\documentclass[num-refs]{wiley-article}

\usepackage{siunitx}

\usepackage{hyperref}

\usepackage{tikz}
\usetikzlibrary{shapes,arrows,calc}
\renewcommand*{\hat}[1]{\widehat{#1}}
\DeclareMathOperator{\dist}{\mathrm{dist}}
\renewcommand*{\*}{\times}

\newcounter{config}
\newcommand{\config}[1]{\refstepcounter{config}\label{#1}C_{\theconfig}}

\newcounter{regle}
\newcommand{\regle}[1]{\refstepcounter{regle}\label{#1}R_{\theregle}}

\newtheorem{conjecture}[theorem]{Conjecture}
\papertype{Original Article}

\title{Every planar graph with $\Delta\geqslant 8$ is totally $(\Delta+2)$-choosable}

\author[1]{Marthe Bonamy}
\author[2]{Théo Pierron}
\author[1]{Éric Sopena}

\affil[1]{Univ.  Bordeaux, Bordeaux INP, CNRS, LaBRI, UMR 5800,
  F-33400 Talence, France}
\affil[2]{ Univ. Lyon, Université Lyon 1, LIRIS UMR CNRS 5205,
  F-69621, Lyon, France}

\corraddress{\phantom{nothing}}
\corremail{marthe.bonamy@labri.fr, eric.sopena@labri.fr, theo.pierron@univ-lyon1.fr}

\fundinginfo{This work is supported by the ANR project HOSIGRA (ANR-17-CE40-0022).}

\runningauthor{Marthe Bonamy, Théo Pierron, Éric Sopena}

\begin{document}

\begin{frontmatter}
\maketitle

\begin{abstract}
Total coloring is a variant of edge coloring where both \mbox{vertices} and
edges are to be colored. A graph is totally $k$-choosable if for any
list assignment of $k$ colors to each vertex and each edge, we can
extract a proper total coloring. In this setting, a graph of maximum
degree $\Delta$ needs at least $\Delta+1$ colors. In the planar
case, Borodin proved in 1989 that $\Delta+2$ colors suffice when
$\Delta$ is at least 9. We show that this bound also holds when
$\Delta$ is $8$.

\end{abstract}
\end{frontmatter}

\section{Introduction}
A graph is totally colored when each vertex and each edge is colored
so that two adjacent vertices or incident elements receive different
colors. The smallest number of colors needed to ensure $G$ has a total
coloring is denoted by $\chi''(G)$. This parameter is a variant of the
chromatic number $\chi$ and the chromatic index $\chi'$ (where only
the vertices or only the edges are to be colored, respectively).

We investigate here the list variant of this coloring: instead of
assigning colors from $\{1,\ldots,k\}$ to elements of the graphs, we
associate to each element a list of available colors and we color each
element with a color from its own list. The question is now to
determine the minimum size of the lists $k$ such that $G$ is totally
$k$-choosable, i.e. such that $G$ has a proper total coloring using
only available colors, regardless of the list
assignment. This new parameter is denoted by $\chi''_\ell(G)$.

Note that by assigning to each element the same list, we simulate the
standard version of coloring. Therefore, list coloring is a
strengthening of standard coloring. In particular,
$\chi''(G)\leqslant \chi''_\ell(G)$. Note that this is not specific to
the total coloring case.

We may first ask how large the gap between a parameter and its list
version can be. While the difference between the chromatic number of a
graph and its list version can be arbitrarily large
(see~\cite{vizing182}), the situation seems to be quite different for
edge and total coloring, as expressed in the following conjecture.

\begin{conjecture}[List coloring
  conjectures~\cite{borodin67,juvan120,vizing182}]
  \label{conj:LC}
  Every simple graph $G$ satisfies $\chi'(G)=\chi'_\ell(G)$ and
  $\chi''(G)=\chi''_\ell(G)$.
\end{conjecture}

Another classical problem is to find bounds for all these
parameters. While there exist graphs with large $\Delta$ and small
chromatic number, the situation is different for edge and total
coloring: all the edges around a vertex must receive different colors,
leading to the bounds $\chi'\geqslant \Delta$ and
$\chi''\geqslant \Delta+1$.

Greedy coloring leads to the following bounds:
$\chi_\ell\leqslant \Delta+1$, $\chi'_\ell\leqslant 2\Delta-1$ and
$\chi''_\ell\leqslant 2\Delta+1$. While the first one can be tight
(for example with a clique), this is far from being the case for edge
coloring, as shown by Vizing's theorem.
\begin{theorem}[Vizing~\cite{vizing184}]
  Every simple graph $G$ satisfies $\chi'(G)\leqslant \Delta(G)+1$.
\end{theorem}

This result is the starting point of some conjectures stated
below. The first one is a weakening of Conjecture~\ref{conj:LC}, and
deals only with choosability.
\begin{conjecture}[Vizing~\cite{vizing182}]
  Every simple graph $G$ satisfies $\chi'_\ell(G)\leqslant \Delta(G)+1$.
\end{conjecture}

The two others generalize the previous ones to the setting of total
coloring.
\begin{conjecture}[\cite{behzad14,vizing182}]
  \label{conj:tot}
  Every simple graph $G$ satisfies $\chi''(G)\leqslant \Delta(G)+2$.
\end{conjecture}

\begin{conjecture}
  \label{conj:tot_list}
  Every simple graph $G$ satisfies $\chi''_\ell(G)\leqslant \Delta(G)+2$.
\end{conjecture}

While all these conjectures are still widely open, many attempts have
been made to tackle them, leading to many results on various graph
classes. For planar graphs, most of the known results are detailed
in~\cite{borodinsurvey}. We summarize the latest ones in the following
table.

\begin{figure}[!h]
\centering
\begin{tabular}{|cc|ll|c|}
  \hline
  & Edge coloring & $\leqslant \Delta+1$& & \cite{vizing184}\\
  & Edge coloring & $\leqslant \Delta$ &when $\Delta \geqslant 7$ (false for $\Delta \leqslant 5$)& \cite{sanders169,vizing185}\\
  \hline
  List & edge coloring & $\leqslant \Delta+1$& when $\Delta \geqslant 8$ or $\Delta\leqslant 4$ & \cite{bonamy,borodin27,juvan120,vizing182}\\
  List & edge coloring & $\leqslant \Delta$ &when $\Delta \geqslant 12$ (false for $\Delta \leqslant 5$) & \cite{borodin27,borodin67}\\
  \hline
  & Total coloring & $\leqslant \Delta+2$ &when $\Delta \neq 6$ & \cite{borodin21,jensen118,kostochka125,sanders169}\\
  & Total coloring & $\leqslant \Delta+1$ &when $\Delta \geqslant 9$ (false for $\Delta \leqslant 3$)& \cite{borodin21,borodin19,borodin67,kowalik130,wang190}\\
  \hline
  List & total coloring & $\leqslant \Delta+2$& when $\Delta \geqslant 9$ or $\Delta\leqslant 3$ & \cite{borodin21,borodin19,juvan120}\\
  List & total coloring & $\leqslant \Delta+1$ &when $\Delta \geqslant 12$ (false for $\Delta\leqslant 3$)& \cite{borodin21,borodin19,borodin67}\\
  \hline
\end{tabular}
\caption{Results for planar graphs with large maximum degree}
\end{figure}
Note that in each case, the bounds on the total-coloring parameters
are one more than their edge-coloring counterpart.

In this paper, we prove the following theorem.
\begin{theorem}
  \label{thm:main}
  Every planar graph of maximum degree at most $8$ is totally
  $10$-choosable.
\end{theorem}
Together with the result of~\cite{borodin19}, we obtain the
following corollary:
\begin{corollary}
  Every planar graph of maximum degree $\Delta$ at least $8$ is
  totally $(\Delta+2)$-choosable.
\end{corollary}
This transposes the results of~\cite{bonamy} and~\cite{jensen118} to
the setting of total list coloring. Moreover, it proves that
Conjecture~\ref{conj:tot_list} holds for planar graphs of maximum
degree at least $8$.

\subsection*{Organization of the paper}

The proof of Theorem~\ref{thm:main} uses the discharging method. This
approach was introduced more than a century ago in~\cite{wernicke} to
study the Four-Color Conjecture, now a theorem. It is especially
well-suited for studying sparse graphs, and leads to many results, as
shown in~\cite{borodinsurvey,cranstonwestsurvey}. 

We proceed by contradiction. Assuming that Theorem~\ref{thm:main} has
a counterexample, we consider the one with the smallest number of
edges. Our goal is to prove that $G$ satisfies structural properties
incompatible with planarity, hence the conclusion. We consider $G$
together with a planar embedding $\mathcal{M}$. Unless specified
otherwise, all the faces discussed in the proof are faces in
$\mathcal{M}$.

In Section~\ref{sec:overview}, we present a set of so-called
\emph{reducible configurations}, which by minimality $G$ cannot
contain. This is proven in Section~\ref{sec:reduction}. We then strive
to reach a contradiction with the planarity of $G$. To do so, we use
the discharging method. This means that we assign some initial weight
to vertices and faces of $G$, then we redistribute those weights, and
obtain a contradiction by double counting the total weight. We present
in Section~\ref{sec:positive} an appropriate collection of discharging
rules, and argue that every element of $G$ ends up with non-negative
weight while the total initial weight was negative.

\section{Proof overview}
\label{sec:overview}

In this section, we describe the reducible configurations. We first
introduce some notation in order to simplify the statements of
configurations.

\subsection{Notation}

We say that a vertex is \emph{triangulated} if all the faces
containing it are triangles. Given a vertex $u$ and two of its
neighbors $v_1,v_2$, the \emph{triangle-distance} between $v_1$ and
$v_2$ around $u$, denoted by $\dist_u(v_1,v_2)$, is the (possibly
infinite) distance between $v_1$ and $v_2$ in the subgraph of $G$
induced by the edges $vw$ such that $uvw$ is a triangular face (see
Figure~\ref{fig:trig_dist}). This distance is the minimum of the
lengths of (at most) two paths in the neighborhood of $u$, each one
turning in one direction. In all the following figures, a node
containing an integer $i$ represents a vertex with degree $i$. An
empty node is a vertex with no degree constraint. Moreover, all the
edges incident to the depicted vertices are not necessarily drawn, and
the drawing does not necessarily corresponds to the chosen embedding
of the graph. When we reduce a configuration, we give a figure with
the names of all the elements we will have to color. It may happen that we do not erase the color of some elements from a coloring obtained using miniamlity. When this happens, the corresponding
element will be depicted in boldface (and may not be given a name).

\begin{figure}[!ht]
  \centering
  \begin{tikzpicture}[every node/.style={draw=black,minimum size = 10pt,ellipse,inner sep=1pt},node distance=1.5cm]
    \node [label=right:{$u$}] (v) at (0,0) {$5$};
    \node [label=right:{$v_3$}](v1) at (108:1) {$2$};
    \node (u) at (180:1)  {$3$};
    \node (v3) at (252:1) {$3$};
    \node [label=right:{$v_1$}](v4) at (324:1) {$2$};
    \node [label=right:{$v_2$}](v2) at (36:1) {$1$};
    \draw (v3) -- (v) -- (v2);
    \draw (v) -- (u);
    \draw (v) -- (v1) -- (u) -- (v3) -- (v4) -- (v);
  \end{tikzpicture}
  \caption{Triangle-distance: $d_u(v_1,v_2)=\infty$ and $d_u(v_1,v_3)=3$.}
  \label{fig:trig_dist}
\end{figure}
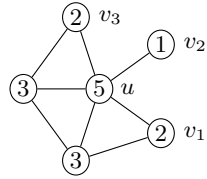

Given an edge $uv$, we say that $v$ is:
\begin{itemize}
\item[$\bullet$] a \emph{weak neighbor} of $u$ if either $v$ is a
  $4^-$-vertex and both faces containing the edge $uv$ are triangles,
  or $v$ is a triangulated $5$-vertex (see Figure~\ref{fig:weak}).

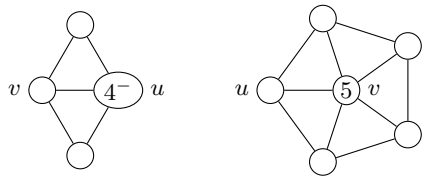
\begin{figure}[!ht]
  \centering
  \begin{tikzpicture}[every node/.style={draw=black,minimum size = 10pt,ellipse,inner sep=1pt},node distance=1.5cm]
    \node [label=right:{$v$}] (v) at (0,0) {$4^-$};
    \node [label=left:{$u$}] (u) at (180:1)  {};
    \node (v1) at (-0.5,0.866){};
    \node (v2) at (-0.5,-0.866){};
    \draw (u) -- (v) -- (v1) -- (u) -- (v2) -- (v);
    \tikzset{xshift=3cm}
    \node [label=right:{$v$}] (v) at (0,0) {5};
    \node [label=left:{$u$}] (u) at (180:1)  {};
    \node (v1) at (252:1){};
    \node (v3) at (324:1){};
    \node (v4) at (36:1){};
    \node (v5) at (108:1){};
    \draw (u) -- (v1) -- (v3) -- (v4) -- (v5) -- (u);
    \draw (v1) -- (v) -- (u);
    \draw (v3) -- (v) -- (v4);
    \draw (v5)-- (v); 
  \end{tikzpicture}
  \caption{A weak neighbor $v$ of $u$.}
  \label{fig:weak}
\end{figure}
\item[$\bullet$] a \emph{semi-weak neighbor} of $u$ if $v$ is a
  $4^-$-vertex and exactly one of the faces containing $uv$ is a
  triangle (see Figure~\ref{fig:semiweak}).

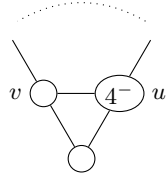
\begin{figure}[!ht]
  \centering
  \begin{tikzpicture}[every node/.style={draw=black,minimum size = 10pt,ellipse,inner sep=1pt},node distance=1.5cm]
    \node [label=right:{$v$}] (v) at (0,0) {$4^-$};
    \node [label=left:{$u$}] (u) at (180:1)  {};
    \node [draw=white] (v1) at (0.5,0.866){};
    \node [draw=white] (v2) at (-1.5,0.866){};
    \node (v3) at (-0.5,-0.866) {};
    \draw (v2) -- (u) -- (v) -- (v1);
    \draw (u) -- (v3) -- (v);
    \draw[dotted, bend right] (v1) to (v2);
  \end{tikzpicture}
  \caption{A semi-weak neighbor $v$ of $u$.}
  \label{fig:semiweak}
\end{figure}
\end{itemize}

Moreover, if $v$ is a weak neighbor of $u$, we often consider the
degree of the common neighbors of $u$ and $v$. We thus define the
following: for any integers $p\leqslant q$, we say that $v$ is a
\emph{$(p,q)$-neighbor} of $u$ if $v$ is a weak neighbor of $u$ and
the two vertices $w_1,w_2$ such that $uvw_1$ and $uvw_2$ are
triangular faces have degree $p$ and $q$. The same holds with $p^+$
(resp. $p^-$), meaning that the degree is at least (resp. at most)
$p$.

We also define special types of $5$-vertices. Consider a $7$-vertex
$u$ with a weak neighbor $v$ of degree 5. We say that $v$ is an:
\begin{enumerate}[label=(\roman*)]
\item \emph{$S_3$-neighbor} of $u$ if one of the following conditions
  holds:
\begin{itemize}
\item[$\bullet$] $v$ is a $(6,6^+)$-neighbor of $u$.
\item[$\bullet$] $v$ is a $(7^+,7^+)$-neighbor of $u$ and $v$ has two
  neighbors $w_1,w_2$ such that $d(w_1)=d(w_2)=6$ and $uvw_1,uvw_2$
  are not triangular faces.
\item[$\bullet$] $v$ has a neighbor $w$ of degree $5$ such that $uvw$
  is not a triangular face.
\end{itemize}
\begin{figure}[!ht]
  \centering
  \begin{tikzpicture}[every node/.style={draw=black,minimum size = 10pt,ellipse,inner sep=1pt},node distance=1.5cm]
    \node [label=right:{$v$}] (v) at (0,0) {$5$};
    \node [label=left:{$u$}] (u) at (180:1)  {$7$};
    \node [label=right:{$w_1$}](v1) at (108:1) {$6$};
    \node [label=right:{$w_2$}](v2) at (252:1) {$6^+$};
    \draw (u) -- (v) -- (v1) -- (u) -- (v2) -- (v);

    \node [xshift=3cm,label=right:{$v$}] (v) at (0,0) {$5$};
    \node [xshift=3cm,label=left:{$u$}] (u) at (180:1)  {$7$};
    \node[xshift=3cm] (v1) at (108:1) {$7^+$};
    \node[xshift=3cm] (v2) at (252:1) {$7^+$};
    \node[xshift=3cm,label=right:{$w_1$}] (v3) at (36:1) {$6$};
    \node[xshift=3cm,label=right:{$w_2$}] (v4) at (324:1) {$6$};
    \draw (v) -- (v1) -- (u) -- (v) -- (v2) -- (v4) -- (v) -- (v3) -- (v4);
    \draw (v1) -- (v3);
    \draw (u) -- (v2);

    \node [xshift=6.5cm,label=above right:{$v$}] (v) at (0,0) {$5$};
    \node [xshift=6.5cm,label=left:{$u$}] (u) at (180:1)  {$7$};
    \node[xshift=6.5cm] (v1) at (108:1) {};
    \node[xshift=6.5cm] (v2) at (252:1) {};
    \node[xshift=6.5cm,label=right:{$w$}] (v3) at (0:1) {$5$};
    \draw (u) -- (v1) -- (v) -- (v2) -- (u) -- (v) -- (v3);
  \end{tikzpicture}
\caption{$v$ is an $S_3$-neighbor of $u$}
\end{figure}
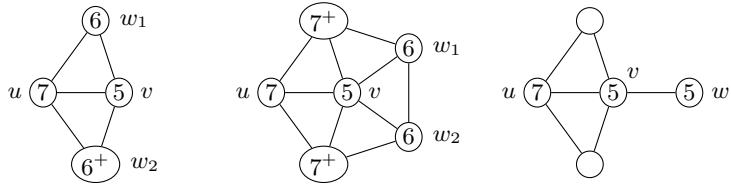
\item \emph{$S_5$-neighbor} of $u$ if every neighbor of $v$ has degree
  $7$.
\item \emph{$S_6$-neighbor} of $u$ if it is not a $(5,6)$-neighbor
  of $u$, nor an $S_3$-neighbor nor an
  $S_5$-neighbor.

\end{enumerate}

We give a similar definition when $u$ is an $8$-vertex with a weak
neighbor $v$ of degree 5. We say that $v$ is an \emph{$E_3$-neighbor}
of $u$ if one of the following conditions holds (see
Figure~\ref{fig:E3neighbor}):
\begin{itemize}
\item[$\bullet$] $v$ is a $(6,7^+)$- or $(7,7)$-neighbor of $u$.
\item[$\bullet$] $v$ is a $(7^+,8)$-neighbor of $u$ and $v$ has two
  neighbors $w_1,w_2$ such that $d(w_1)=d(w_2)=6$ and $uvw_1,uvw_2$
  are not triangular faces.
\item[$\bullet$] $v$ is a $(7^+,8)$-neighbor of $u$ and $v$ has a
  neighbor $w$ of degree $5$ such that $uvw$ is not a triangular face.
\end{itemize}
\begin{figure}[!ht]
  \centering
  \begin{tikzpicture}[every node/.style={draw=black,minimum size = 10pt,ellipse,inner sep=1pt},node distance=1.5cm]
    \node [xshift=-3cm,label=right:{$v$}] (v) at (0,0) {$5$};
    \node [xshift=-3cm,label=left:{$u$}] (u) at (180:1)  {$8$};
    \node [xshift=-3cm,label=right:{$w_1$}](v1) at (108:1) {$6$};
    \node [xshift=-3cm,label=right:{$w_2$}](v2) at (252:1) {$7^+$};
    \draw (u) -- (v) -- (v1) -- (u) -- (v2) -- (v);

    \node [label=right:{$v$}] (v) at (0,0) {$5$};
    \node [label=left:{$u$}] (u) at (180:1)  {$8$};
    \node [label=right:{$w_1$}](v1) at (108:1) {$7$};
    \node [label=right:{$w_2$}](v2) at (252:1) {$7$};
    \draw (u) -- (v) -- (v1) -- (u) -- (v2) -- (v);

    \node [xshift=3cm,label=right:{$v$}] (v) at (0,0) {$5$};
    \node [xshift=3cm,label=left:{$u$}] (u) at (180:1)  {$8$};
    \node[xshift=3cm] (v1) at (108:1) {$7^+$};
    \node[xshift=3cm] (v2) at (252:1) {$8$};
    \node[xshift=3cm,label=right:{$w_1$}] (v3) at (36:1) {$6$};
    \node[xshift=3cm,label=right:{$w_2$}] (v4) at (324:1) {$6$};
    \draw (v) -- (v1) -- (u) -- (v) -- (v2) -- (v4) -- (v) -- (v3) -- (v4);
    \draw (v1) -- (v3);
    \draw (u) -- (v2);

    \node [xshift=6.5cm,label=above right:{$v$}] (v) at (0,0) {$5$};
    \node [xshift=6.5cm,label=left:{$u$}] (u) at (180:1)  {$8$};
    \node[xshift=6.5cm] (v1) at (108:1) {$7^+$};
    \node[xshift=6.5cm] (v2) at (252:1) {8};
    \node[xshift=6.5cm,label=right:{$w$}] (v3) at (0:1) {$5$};
    \draw (u) -- (v1) -- (v) -- (v2) -- (u) -- (v) -- (v3);
  \end{tikzpicture}
\caption{$v$ is an $E_3$-neighbor of $u$}
  \label{fig:E3neighbor}
\end{figure}
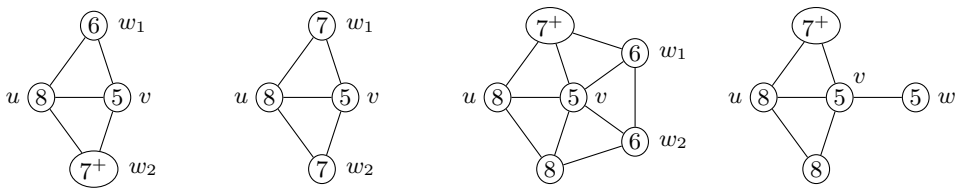

\subsection{Configurations}
\label{sec:config}
We define several configurations, derived from the upcoming
$C_{\ref{C1}},\ldots,C_{\ref{C19}}$ as follows. A configuration $C$ is a
\emph{sub-configuration} of $C'$ if we can obtain $C$ by decreasing
the degree of vertices in $C'$ while preserving the adjacency relation
and the triangle-distance: for every vertices $x,y,z$, $x$ and $y$ are
adjacent in $C$ if and only if they are in $C'$ and $\dist_z(x,y)$ is
the same in $C$ and $C'$. For example, a path $uvw$ where
$d(u)=d(v)=d(w)=4$ is a sub-configuration of $C_{\ref{C3a}}$ but a
path $u_1u_2u_3u_4$ is not a subconfiguration of $C_{\ref{C2}}$ even
if $d(u_1)=d(u_3)=3$ and $d(u_2)=d(u_4)=8$.

\begin{itemize}
\item[$\bullet$] $\config{C1}$ is an edge $(u,v)$ such that
  $d(u)+d(v) \leqslant 10$ and $d(u)\leqslant 4$.
\item[$\bullet$] $\config{C2}$ is an even cycle $v_1\cdots v_{2n}v_1$
  such that for $1\leqslant i\leqslant n$, $d(v_{2i-1})\leqslant 4$
  and $d(v_{2i-1})+d(v_{2i})\leqslant 11$.
\item[$\bullet$] $\config{C3b}$ is a triangle with two vertices of
  degree $5$ and one of degree $6$.
\item[$\bullet$] $\config{C3a}$ is a vertex of degree $5$ with two
  neighbors of degree $5$.
\item[$\bullet$] $\config{C4}$ is a $7$-vertex $u$ with a
  $(5,6)$-neighbor $v_1$ and a $5$-neighbor $v_2$ such that either
  $\dist_u(v_1,v_2)=2$, or $v_2$ is a $(5,6)$-neighbor of $u$ with
  $\dist_u(v_1,v_2)\leqslant 3$, see Figure~\ref{fig:C4init}.
  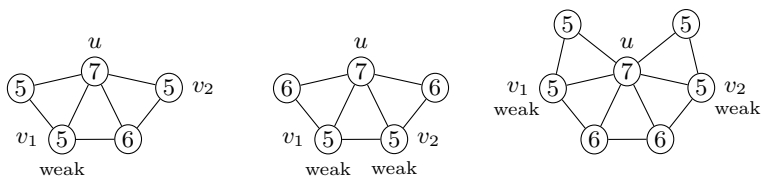
\begin{figure}[!h]
    \centering
    \begin{tikzpicture}[every node/.style={draw=black,minimum size = 10pt,ellipse,inner sep=1pt}]
    \node [label=above:{$u$}] (u) at (0,0) {$7$};
    \node (v1) at (193:1)  {$5$};
    \node[label=left:{$v_1$}, label=below:{\scriptsize weak}] (v2) at (244.5:1) {$5$};
    \node (v3) at (296:1) {$6$};
    \node[label=right:{$v_2$}] (v4) at (347.5:1) {$5$};
    \draw (v2) -- (u) -- (v1) -- (v2) -- (v3) -- (v4) -- (u) -- (v3);

    \node [xshift=3.5cm,label=above:{$u$}] (u) at (0,0) {$7$};
    \node[xshift=3.5cm] (v1) at (193:1)  {$6$};
    \node[xshift=3.5cm,label=left:{$v_1$}, label=below:{\scriptsize weak}] (v2) at (244.5:1) {$5$};
    \node[xshift=3.5cm,label=right:{$v_2$}, label=below:{\scriptsize weak}] (v3) at (296:1) {$5$};
    \node[xshift=3.5cm] (v4) at (347.5:1) {$6$};
    \draw (v2) -- (u) -- (v1) -- (v2) -- (v3) -- (v4) -- (u) -- (v3);

    \node[xshift=7cm,label=above:{$u$}] (u) at (0,0) {$7$};
    \node[xshift=7cm] (v0) at (141.5:1)  {$5$};
    \node[xshift=7cm,label=left:{$v_1$}, label=below left:{\scriptsize weak}] (v1) at (193:1)  {$5$};
    \node[xshift=7cm] (v2) at (244.5:1) {$6$};
    \node[xshift=7cm] (v3) at (296:1) {$6$};
    \node[xshift=7cm,label=right:{$v_2$}, label=below right:{\scriptsize weak}] (v4) at (347.5:1) {$5$};
    \node[xshift=7cm] (v5) at (39:1)  {$5$};
    \draw (v2) -- (u) -- (v1) -- (v2) -- (v3) -- (v4) -- (u) -- (v3);
    \draw (v1) -- (v0) -- (u) -- (v5) -- (v4);
  \end{tikzpicture}
\caption{Configuration $C_{\ref{C4}}$}
  \label{fig:C4init}
  \end{figure}
\item[$\bullet$] $\config{C5}$ is a $5$-vertex $u$ adjacent to three
  $6$-vertices $v_1,v_2,v_3$ and two vertices $v_4,v_5$ such that
  either there are two edges $v_1v_2$ and $v_2v_3$ or $u$ is
  triangulated and $d(v_4)=d(v_5)=7$, see Figure~\ref{fig:C5init}.
  \begin{figure}[!h]
    \centering
  \begin{tikzpicture}[every node/.style={draw=black,minimum size = 10pt,ellipse,inner sep=1pt}]
    \node [label=right:{$v_2$}] (v2) at (0:1) {$6$};
    \node [label=left:{$u$}] (u) at (0,0)  {$5$};
    \node[label=right:{$v_1$}] (v1) at (72:1) {$6$};
    \node[label=right:{$v_3$}] (v3) at (-72:1) {$6$};
    \draw (u) -- (v2) -- (v1) -- (u) -- (v3) -- (v2);

    \node [xshift=4cm,label=left:{$u$}] (u) at (0,0)  {$5$};
    \node [xshift=4cm,label=right:{$v_1$}] (v0) at (90:1) {$6$};
    \node[xshift=4cm,label=left:{$v_4$}] (v1) at (162:1) {$7$};
    \node[xshift=4cm,label=left:{$v_2$}] (v2) at (234:1) {$6$};
    \node[xshift=4cm,label=right:{$v_3$}] (v3) at (306:1) {$6$};
    \node[xshift=4cm,label=right:{$v_5$}] (v4) at (18:1) {$7$};
    \draw (u) -- (v1) -- (v0) -- (u) -- (v2) -- (v3) -- (u) -- (v4) -- (v3);
    \draw (v1) -- (v2);
    \draw (v0) -- (v4);
  \end{tikzpicture}
\caption{Configuration $C_{\ref{C5}}$}
  \label{fig:C5init}
  \end{figure}
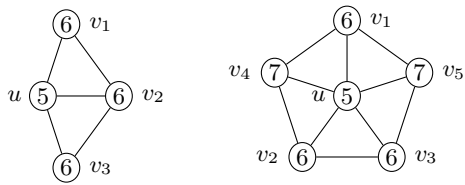
\item[$\bullet$] $\config{C6}$ is a $7$-vertex $u$ with a
  $(5,6)$-neighbor of degree $5$ and a neighbor of degree $4$, see
  Figure~\ref{fig:C6init}.
  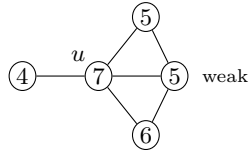
\begin{figure}[!h]
    \centering
  \begin{tikzpicture}[every node/.style={draw=black,minimum size = 10pt,ellipse,inner sep=1pt}]
    \node [label=right:{\scriptsize weak}](v2) at (0:1) {$5$};
    \node [label=above left:{$u$}] (u) at (0,0)  {$7$};
    \node (v1) at (51.5:1) {$5$};
    \node (v3) at (-51.5:1) {$6$};
    \node (v4) at (180:1) {$4$};
    \draw (v4) -- (u) -- (v2) -- (v1) -- (u) -- (v3) -- (v2);
  \end{tikzpicture}
\caption{Configuration $C_{\ref{C6}}$}
  \label{fig:C6init}
  \end{figure}
\item[$\bullet$] $\config{C7}$ is a $7$-vertex $u$ with an
  $S_3$-neighbor $v_1$, a $(7,7^+)$-neighbor $v_3$ of degree $4$, and
  a neighbor $v_4$ of degree $5$ such that $\dist_u(v_1,v_3)=2$ and
  the common neighbor $v_2$ of $u,v_1,v_3$ has degree 7, see Figure~\ref{fig:C7init}.
  \begin{figure}[!h]
    \centering
    \begin{tikzpicture}[every node/.style={draw=black,minimum size = 10pt,ellipse,inner sep=1pt}]
      \node [label=below:{$u$}] (u) at (0,0)  {$7$};
      \node [label=right:{$S_3$}, label=below:{$v_1$}] (v1) at (348.5:1) {$5$};
      \node [label=right:{$v_2$}] (v2) at (38.5:1) {$7$};
      \node[label=above left:{$v_3$}, label=above right:{\scriptsize weak}] (v3) at (90:1) {$4$};
      \node [label=left:{$v_4$}](v5) at (192:1) {$5$};
      \node (v4) at (141.5:1) {$8$};
      \draw (v5) -- (u) -- (v3) -- (v2) -- (v1) -- (u) -- (v2);
      \draw (u) -- (v4) -- (v3);
    \end{tikzpicture}
\caption{Configuration $C_{\ref{C7}}$}
    \label{fig:C7init}
  \end{figure}
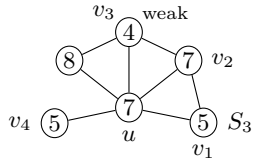
\item[$\bullet$] $\config{C8}$ is a $7$-vertex $u$ with a weak
  neighbor $v_1$, a $(7,7^+)$-neighbor $v_2$ of degree $4$ and a weak
  neighbor $v_3(\neq v_2)$ such that
  $\dist_u(v_1,v_2)=\dist_u(v_1,v_3)=2$ and $v_1$ is a $S_3$- or
  $S_5$-neighbor of $u$, see Figure~\ref{fig:C8init}.
  \begin{figure}[!h]
    \centering
    \begin{tikzpicture}[every node/.style={draw=black,minimum size = 10pt,ellipse,inner sep=1pt}]
      \node [label=left:{$u$}] (u) at (0,0)  {$7$};
      \node [label=below right:{$S_3$}, label=above right :{$v_1$}] (v1) at (0:1) {$5$};
      \node (v2) at (51.5:1) {$8$};
      \node[label=left:{$v_2$}, label=above:{\scriptsize weak}] (v3) at (103:1) {$4$};
      \node (v4) at (154.5:1) {$7$};
      \node (v5) at (308.5:1) {$8$};
      \node[label=left:{$v_3$}, label=below:{\scriptsize weak}] (v6) at (257:1) {$5$};
      \node (v7) at (206.5:1) {$8$};
      \draw (v5) -- (u) -- (v4) -- (v3) -- (v2) -- (v1) -- (u) -- (v2);
      \draw (v1) -- (v5) -- (v6) -- (u) -- (v3);
      \draw (u) -- (v7) -- (v6);

      \node [xshift=4cm,label=left:{$u$}] (u) at (0,0)  {$7$};
      \node [xshift=4cm, label=right :{$v_1$}] (v1) at (0:1) {$5$};
      \node[xshift=4cm] (v2) at (51.5:1) {$7$};
      \node[xshift=4cm, label=above:{\scriptsize weak},label=left:{$v_2$}] (v3) at (103:1) {$4$};
      \node[xshift=4cm] (v7) at (206.5:1) {$8$};
      \node[xshift=4cm] (v4) at (154.5:1) {$8$};
      \node[xshift=4cm] (v5) at (308.5:1) {7};
      \node[xshift=4cm, label=below:{\scriptsize weak},label=left:{$v_3$}] (v6) at (257:1) {$5$};
      \node[xshift=5cm] (w1) at (-36:1) {7};
      \node[xshift=5cm] (w2) at (36:1) {7};
      \draw (v2) -- (w2) -- (w1) -- (v5) -- (u) -- (v3) -- (v2) -- (v1) -- (u) -- (v2);
      \draw (v1) -- (v5) -- (v6) -- (u);
      \draw (w1) -- (v1) -- (w2);
      \draw (v3) -- (v4) -- (u) -- (v7) -- (v6);
    \end{tikzpicture}
\caption{Configuration $C_{\ref{C8}}$}
    \label{fig:C8init}
  \end{figure}
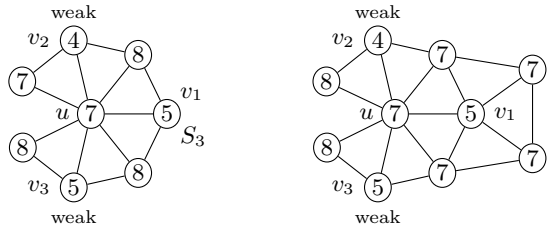
\item[$\bullet$] $\config{C9}$ is a $7$-vertex $u$ with three weak
  neighbors of degree $4$ and a neighbor of degree $7$.
\item[$\bullet$] $\config{C10}$ is a $7$-vertex $u$ with a
  $(7,7^+)$-neighbor $v_1$ of degree $4$, two weak neighbors $v_2$ and
  $v_3$ of degree $4$ and $5$ such that $\dist_u(v_1,v_2)=2$ and
  either $u,v_1$ and $v_3$ have a common neighbor of degree $7$, or
  $v_3$ is an $S_3$-neighbor of $u$ such that $\dist_u(v_1,v_3)=3$,
  see Figure~\ref{fig:C10init}.
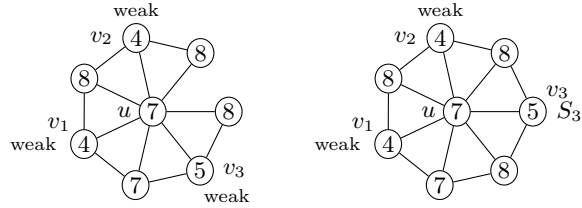
\begin{figure}[!h]
    \centering   
    \begin{tikzpicture}[every node/.style={draw=black,minimum size = 10pt,ellipse,inner sep=1pt}]
      \node [label=left:{$u$}] (u) at (0,0)  {$7$};
       \node [label=above left:{$v_1$}, label=left:{\scriptsize weak}] (v1) at (205.5:1) {$4$};
      \node (v0) at (257:1) {$7$};
      \node (v2) at (154:1) {8};
      \node[label=left:{$v_2$}, label=above:{\scriptsize weak}] (v3) at (102.5:1) {$4$};
      \node[label=right:{$v_3$}, label=273:{\scriptsize weak}] (v6) at (308.5:1) {$5$};
      \node (v4) at (51:1) {$8$};
      \node (v5) at (0:1) {$8$};
      \draw (v2) -- (v3) -- (u) -- (v2) -- (v1) -- (v0) -- (v6) -- (u) -- (v0);
      \draw (v1) -- (u);
      \draw (v3) -- (v4) -- (u) -- (v5) -- (v6);

      \node[xshift=4cm,label=left:{$u$}] (u) at (0,0)  {$7$};
      \node[xshift=4cm, label=left:{\scriptsize weak},label=above left:{$v_1$}] (v1) at (205.5:1) {$4$};
      \node[xshift=4cm] (v0) at (257:1) {$7$};
      \node[xshift=4cm] (v2) at (154:1) {8};
      \node[xshift=4cm,label=left:{$v_2$}, label=above:{\scriptsize weak}] (v3) at (102.5:1) {$4$};
      \node[xshift=4cm,label=right:{$S_3$},label=above right:{$v_3$}] (v5) at (0:1) {$5$};
      \node[xshift=4cm] (v6) at (308.5:1) {$8$};
      \node[xshift=4cm] (v4) at (51:1) {$8$};
      
      \draw (v2) -- (v3) -- (u) -- (v2) -- (v1) -- (v0) -- (v6) -- (v5) -- (u) -- (v6);
      \draw (v1) -- (u) -- (v0);
      \draw (v3)-- (v4) -- (v5);
      \draw (u) -- (v4);
    \end{tikzpicture}
\caption{Configuration $C_{\ref{C10}}$}
    \label{fig:C10init}
  \end{figure}
\item[$\bullet$] $\config{C11}$ is a $7$-vertex $u$ with two weak
  neighbors $v_1,v_2$ of degree $4$ satisfying $\dist_u(v_1,v_2)>2$,
  and a weak neighbor $v_3$ of degree $5$ such that either $v_1$ is a
  $(7,7)$-neighbor of $u$, or so is $v_3$, or $v_1$ is a
  $(7,7^+)$-neighbor and $v_3$ is an $S_3$-neighbor of $u$, see
  Figure~\ref{fig:C11init}.
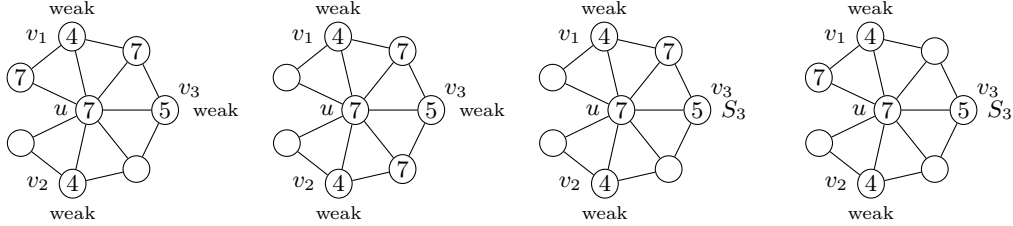
\begin{figure}[!h]
    \centering
    \begin{tikzpicture}[every node/.style={draw=black,minimum size = 10pt,ellipse,inner sep=1pt}]
      \node [label=left:{$u$}] (u) at (0,0)  {$7$};
      \node [label=above right:{$v_3$}, label=right:{\scriptsize weak}] (v1) at (0:1) {$5$};
      \node (v0) at (-51.5:1) {};
      \node (v2) at (51.5:1) {$7$};
      \node[label=left:{$v_1$}, label=above:{\scriptsize weak}] (v3) at (103:1) {$4$};
      \node (v4) at (154.5:1) {$7$};
      \node (v5) at (206:1) {};
      \node[label=left:{$v_2$}, label=below:{\scriptsize weak}] (v6) at (257.5:1) {$4$};
 
      \draw (v3) -- (u) -- (v4) -- (v3) -- (v2) -- (v1) -- (v0) -- (v6) -- (v5) -- (u) -- (v6);
      \draw (v2) -- (u) -- (v0);
      \draw (v1) -- (u);

      \node[xshift=3.5cm,label=left:{$u$}] (u) at (0,0)  {$7$};
      \node[xshift=3.5cm] (v0) at (-51.5:1) {$7$};
      \node[xshift=3.5cm,label=above right:{$v_3$}, label=right:{\scriptsize weak}] (v1) at (0:1) {$5$};
      \node[xshift=3.5cm] (v2) at (51.5:1) {$7$};
      \node[xshift=3.5cm,label=left:{$v_1$}, label=above:{\scriptsize weak}]  (v3) at (103:1) {$4$};
      \node[xshift=3.5cm] (v4) at (154.5:1) {};
      \node[xshift=3.5cm] (v5) at (206:1) {};
      \node[xshift=3.5cm,label=left:{$v_2$}, label=below:{\scriptsize weak}] (v6) at (257.5:1) {4};
 
      \draw (v3) -- (u) -- (v4) -- (v3) -- (v2) -- (v1) -- (v0) -- (v6) -- (v5) -- (u) -- (v6);
      \draw (v2) -- (u) -- (v0);
      \draw (v1) -- (u);

      \node[xshift=7cm,label=left:{$u$}] (u) at (0,0)  {$7$};
      \node[xshift=7cm] (v0) at (-51.5:1) {};
      \node[xshift=7cm,label=above right:{$v_3$},label=right :{$S_3$}] (v1) at (0:1) {$5$};
      \node[xshift=7cm] (v2) at (51.5:1) {$7$};
      \node[xshift=7cm,label=left:{$v_1$}, label=above:{\scriptsize weak}]  (v3) at (103:1) {$4$};
      \node[xshift=7cm] (v4) at (154.5:1) {};
      \node[xshift=7cm] (v5) at (206:1) {};
      \node[xshift=7cm,label=left:{$v_2$}, label=below:{\scriptsize weak}] (v6) at (257.5:1) {4}; 
      \draw (v3) -- (u) -- (v4) -- (v3) -- (v2) -- (v1) -- (v0) -- (v6) -- (v5) -- (u) -- (v6);
      \draw (v2) -- (u) -- (v0);
      \draw (v1) -- (u);

      \node[xshift=10.5cm,label=left:{$u$}] (u) at (0,0)  {$7$};
      \node[xshift=10.5cm] (v0) at (-51.5:1) {};
      \node[xshift=10.5cm,label=above right:{$v_3$},label=right :{$S_3$}] (v1) at (0:1) {$5$};
      \node[xshift=10.5cm] (v2) at (51.5:1) {};
      \node[xshift=10.5cm,label=left:{$v_1$}, label=above:{\scriptsize weak}] (v3) at (103:1) {$4$};
      \node[xshift=10.5cm] (v4) at (154.5:1) {7};
      \node[xshift=10.5cm] (v5) at (206:1) {};
      \node[xshift=10.5cm,label=left:{$v_2$}, label=below:{\scriptsize weak}] (v6) at (257.5:1) {4};
 
      \draw (v3) -- (u) -- (v4) -- (v3) -- (v2) -- (v1) -- (v0) -- (v6) -- (v5) -- (u) -- (v6);
      \draw (v2) -- (u) -- (v0);
      \draw (v1) -- (u);
    \end{tikzpicture}
\caption{Configuration $C_{\ref{C11}}$}
    \label{fig:C11init}
  \end{figure}
\item[$\bullet$] $\config{C12}$ is an $8$-vertex $u$ with either three
  neighbors of degree 6 and five weak neighbors of degree $5$, or four
  neighbors of degree 6 and four weak neighbors of degree 5, two of
  them being $(6,6)$-neighbors, see Figure~\ref{fig:C12init}.
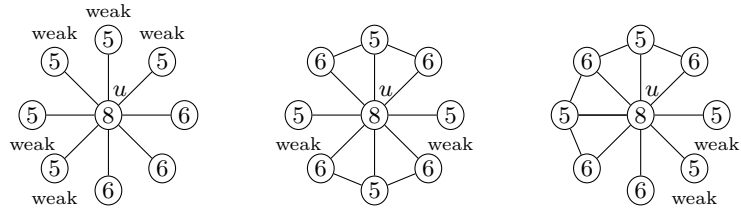
\begin{figure}[!h]
    \centering
    \begin{tikzpicture}[every node/.style={draw=black,minimum size = 10pt,ellipse,inner sep=1pt}]
      \node [label=85:{$u$}] (u) at (0,0)  {8};
      \node [label=above:{\scriptsize weak}] (v0) at (45:1) {5};
      \node [label=above:{\scriptsize weak}] (v1) at (90:1) {5};
      \node [label=above:{\scriptsize weak}] (v2) at (135:1) {5};
      \node [label=below:{\scriptsize weak}] (v3) at (180:1) {5};
      \node [label=below:{\scriptsize weak}] (v4) at (225:1) {5};
      \node (v5) at (270:1) {6};
      \node (v6) at (315:1) {6};
      \node (v7) at (0:1) {6};
      \draw (v0) -- (u) -- (v1);
      \draw (v2) -- (u) -- (v3);
      \draw (v4) -- (u) -- (v5);
      \draw (v6) -- (u) -- (v7);

      \node[xshift=3.5cm, label=85:{$u$}] (u) at (0,0)  {8};
      \node[xshift=3.5cm] (v0) at (45:1) {6};
      \node[xshift=3.5cm] (v1) at (90:1) {5};
      \node[xshift=3.5cm] (v2) at (135:1) {6};
      \node[xshift=3.5cm, label=below:{\scriptsize weak}] (v3) at (180:1) {5};
      \node[xshift=3.5cm] (v4) at (225:1) {6};
      \node[xshift=3.5cm] (v5) at (270:1) {5};
      \node[xshift=3.5cm] (v6) at (315:1) {6};
      \node[xshift=3.5cm, label=below:{\scriptsize weak}] (v7) at (0:1) {5};
      \draw (v1) -- (v0) -- (u) -- (v1) -- (v2) -- (u) -- (v3);
      \draw (v5) -- (v4) -- (u) -- (v5) -- (v6) -- (u) -- (v7);

      \node[xshift=7cm, label=85:{$u$}] (u) at (0,0)  {8};
      \node[xshift=7cm] (v0) at (45:1) {6};
      \node[xshift=7cm] (v1) at (90:1) {5};
      \node[xshift=7cm] (v2) at (135:1) {6};
      \node[xshift=7cm] (v3) at (180:1) {5};
      \node[xshift=7cm] (v4) at (225:1) {6};
      \node[xshift=7cm] (v5) at (270:1) {6};
      \node[xshift=7cm, label=below:{\scriptsize weak}] (v6) at (315:1) {5};
      \node[xshift=7cm, label=below:{\scriptsize weak}] (v7) at (0:1) {5};
      \draw (v1) -- (v0) -- (u) -- (v1) -- (v2) -- (u) -- (v3) -- (v2);
      \draw (u) -- (v3) -- (v4) -- (u) -- (v5);
      \draw (v6) -- (u) -- (v7);
    \end{tikzpicture}
\caption{Configuration $C_{\ref{C12}}$}
    \label{fig:C12init}
  \end{figure}
\item[$\bullet$] $\config{C13}$ is an $8$-vertex $u$ with four
  neighbors $v_1,v_2,v_3,v_4$ of degree $4$ or $5$ such that one of
  the following holds (see Figure~\ref{fig:C13init}):
  \begin{itemize}
  \item $v_1,v_2,v_3,v_4$ are weak neighbors of degree $4$ and $u$ has
    a neighbor of degree $7$,
  \item $v_1,v_2,v_3,v_4$ are $(7,8)$-neighbors, and at most one of
    them has degree $5$
  \item $v_1$ is a $(7,7)$-neighbor of degree $4$, and $v_2$ is a weak
    neighbor of $u$ of degree $4$ such that
    $\dist_u(v_1,v_2)=\dist_u(v_1,v_4)=2$.
  \end{itemize}
\begin{figure}[!h]
    \centering
    \begin{tikzpicture}[every node/.style={draw=black,minimum size = 10pt,ellipse,inner sep=1pt}]
      \node[label=85:{$u$}] (u) at (0,0)  {8};
      \node[label=right:{$w$}] (v0) at (45:1) {7};
      \node[label=above:{$v_1$}] (v1) at (90:1) {4};
      \node (v2) at (135:1) {$8$};
      \node[label=left:{$v_4$}] (v3) at (180:1) {4};
      \node (v4) at (225:1) {$8$};
      \node[label=below:{$v_3$}] (v5) at (270:1) {4};
      \node (v6) at (315:1) {$8$};
      \node[label=right:{$v_2$}] (v7) at (0:1) {4};
      \draw (v6) -- (u);
      \draw (v5) -- (u) ;
      \draw (v4) -- (u) ;
      \draw (v3) -- (u) ;
      \draw (v2) -- (u) ;
      \draw (v1) -- (u) ;
      \draw (v0) -- (u) ;
      \draw (v7) -- (u) ;
      \draw (v5) -- (v6) ;
      \draw (v5) -- (v4) ;
      \draw (v4) -- (v3) ;
      \draw (v3) -- (v2) ;
      \draw (v2) -- (v1) ;
      \draw (v1) -- (v0) ;
      \draw (v0) -- (v7) ;
      \draw (v6) -- (v7) ;

      \node[xshift=3.5cm, label=85:{$u$}] (u) at (0,0)  {8};
      \node[xshift=3.5cm] (v0) at (45:1) {7};
      \node[xshift=3.5cm] (v1) at (90:1) {4};
      \node[xshift=3.5cm] (v2) at (135:1) {8};
      \node[xshift=3.5cm, label = above left:{\scriptsize weak}] (v3) at (180:1) {5};
      \node[xshift=3.5cm] (v4) at (225:1) {7};
      \node[xshift=3.5cm] (v5) at (270:1) {4};
      \node[xshift=3.5cm] (v6) at (315:1) {8};
      \node[xshift=3.5cm] (v7) at (0:1) {4};
      \draw (v1) -- (v0) -- (u) -- (v1) -- (v2) -- (u) -- (v3);
      \draw (v5) -- (v4) -- (u) -- (v5) -- (v6) -- (u) -- (v7);
      \draw (v2) -- (v3) -- (v4);
      \draw (v6) -- (v7) -- (v0);

      \node[xshift=7cm, label=85:{$u$}] (u) at (0,0)  {8};
      \node[xshift=7cm] (v0) at (45:1) {7};
      \node[xshift=7cm, label=above:{$v_1$}, label=above right:{\scriptsize weak}] (v1) at (90:1) {4};
      \node[xshift=7cm] (v2) at (135:1) {7};
      \node[xshift=7cm, label=left:{$v_4$}] (v3) at (180:1) {5};
      \node[xshift=7cm, label=below:{$v_3$}] (v5) at (270:1) {5};
      \node[xshift=7cm] (v6) at (315:1) {$8$};
      \node[xshift=7cm, label=right:{$v_2$},label=273:{\scriptsize weak}] (v7) at (0:1) {4};
      \draw (v1) -- (v0) -- (u) -- (v1) -- (v2) -- (u) -- (v3) -- (v2);
      \draw (u) -- (v3) -- (u) -- (v5);
      \draw (v6) -- (u) -- (v7) -- (v6);
      \draw (v0) -- (v7);
    \end{tikzpicture}
\caption{Configuration $C_{\ref{C13}}$}
    \label{fig:C13init}
  \end{figure}
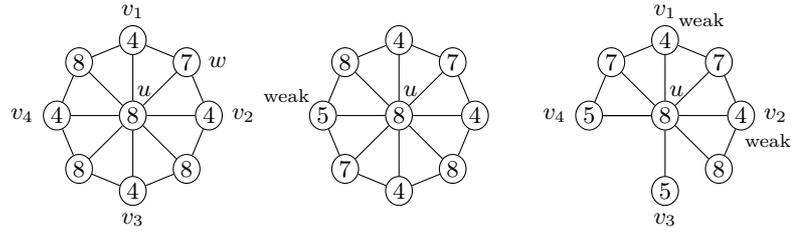
\item[$\bullet$] $\config{C15}$ is an $8$-vertex $u$ with a weak
  neighbor of degree $3$ and either two $(6,6)$-neighbors of degree
  $5$, or one $(6,6)$-neighbor of degree $5$ and two other weak neighbors of degree $5$, or two $(5,6)$-neighbors and a $(6,8)$-neighbor of degree $5$, see Figure~\ref{fig:C15init}.
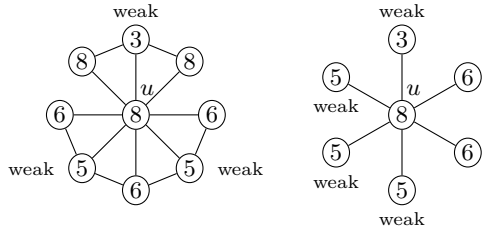
\begin{figure}[!h]
    \centering
    \begin{tikzpicture}[every node/.style={draw=black,minimum size = 10pt,ellipse,inner sep=1pt}]
      \node[ label=85:{$u$}] (u) at (0,0)  {8};
      \node[label=above:{\scriptsize weak}] (v1) at (90:1) {3};
      \node (v3) at (180:1) {6};
      \node[label=left:{\scriptsize weak}] (v4) at (225:1) {5};
      \node (v5) at (270:1) {6};
      \node[label=right:{\scriptsize weak}] (v6) at (315:1) {5};
      \node[label=right:{\scriptsize (weak)}] (v7) at (0:1) {5/6};
      \node (v2) at (45:1) {8};
      \node (v0) at (135:1) {8};
      \draw (v5) -- (v4) -- (u) -- (v5) -- (v6) -- (u) -- (v7) -- (v6);
      \draw (v1) -- (u) -- (v3) -- (v4);
      \draw (u) -- (v0) -- (v1) -- (v2) -- (u);

      \tikzset{xshift=4cm}
        \node[ label=85:{$u$}] (u) at (0,0)  {8};
      \node[label=above:{\scriptsize weak}] (v1) at (90:1) {3};
      \node[label=left:{\scriptsize weak}] (v3) at (180:1) {5};
      \node (v4) at (225:1) {6};
      \node[label=below:{\scriptsize weak}] (v5) at (270:1) {5};
      \node[label=right:{\scriptsize weak}] (v6) at (315:1) {5};
      \node (v7) at (0:1) {6};
      \node (v2) at (45:1) {8};
      \node (v0) at (135:1) {8};
      \draw (v5) -- (v4) -- (u) -- (v5) -- (v6) -- (u) -- (v7) -- (v6);
      \draw (v1) -- (u) -- (v3) -- (v4);
      \draw (u) -- (v0) -- (v1) -- (v2) -- (u);
      \draw (v3) -- (v0);
      
      \tikzset{xshift=4cm}
    \node[ label=85:{$u$}] (u) at (0,0)  {8};
      \node[label=above:{\scriptsize weak}] (v1) at (90:1) {3};
      \node[label=left:{\scriptsize weak}] (v3) at (180:1) {5};
      \node (v4) at (225:1) {6};
      \node[label=below:{\scriptsize weak}] (v5) at (270:1) {5};
      \node (v6) at (315:1) {6};
      \node[label=right:{\scriptsize weak}] (v7) at (0:1) {5};
      \node (v2) at (45:1) {8};
      \node (v0) at (135:1) {8};
      \draw (v5) -- (v4) -- (u) -- (v5) -- (v6) -- (u) -- (v7) -- (v6);
      \draw (v1) -- (u) -- (v3) -- (v4);
      \draw (u) -- (v0) -- (v1) -- (v2) -- (u);
      \draw (v7) -- (v2);
      \draw (v3) -- (v0);
      
    \end{tikzpicture}
\caption{Configuration $C_{\ref{C15}}$}
    \label{fig:C15init}
  \end{figure}
\item[$\bullet$] $\config{C15.5a}$ is a triangulated $8$-vertex $u$
  with a weak 3-neighbor $v_1$, a weak 4-neighbor $v_2$ such that
  $\dist_u(v_1,v_2)=4$, and two weak neighbors $v_3,v_4$ of degree 5
  such that either $v_2$ is a $(7,8)$-neighbor of $u$, or $v_3$ is an
  $E_3$-neighbor of $u$, see Figure~\ref{fig:C155ainit}.
  \begin{figure}[!h]
    \centering
    \begin{tikzpicture}[v/.style={draw=black,minimum size = 10pt,ellipse,inner sep=1pt}]
      \node[v,label=85:{$u$}] (u) at (0,0)  {8};
      \node[v] (v0) at (45:1) {8};
      \node[v] (v1) at (90:1) {3};
      \node[v] (v2) at (135:1) {8};
      \node[v] (v3) at (180:1) {5};
      \node[v] (v4) at (225:1) {7};
      \node[v] (v5) at (270:1) {4};
      \node[v] (v6) at (315:1) {8};
      \node[v] (v7) at (0:1) {5};
      \draw (v6) -- (u) ;
      \draw (v5) -- (u) ;
      \draw (v4) -- (u) ;
      \draw (v3) -- (u) ;
      \draw (v2) -- (u) ;
      \draw (v1) -- (u) ;
      \draw (v0) -- (u) ;
      \draw (v7) -- (u) ;
      \draw (v5) -- (v6);
      \draw (v5) -- (v4);
      \draw (v4) -- (v3);
      \draw (v3) -- (v2);
      \draw (v2) -- (v1);
      \draw (v1) -- (v0);
      \draw (v0) -- (v7);
      \draw (v6) -- (v7);
      \tikzset{xshift=3cm}
      \node[v,label=85:{$u$}] (u) at (0,0)  {8};
      \node[v] (v0) at (45:1) {8};
      \node[v] (v1) at (90:1) {3};
      \node[v] (v2) at (135:1) {8};
      \node[v] (v3) at (180:1) {5};
      \node[v] (v4) at (225:1) {8};
      \node[v] (v5) at (270:1) {4};
      \node[v] (v6) at (315:1) {8};
      \node[v,label=right:{$E_3$}] (v7) at (0:1) {5};
      \draw (v6) -- (u) ;
      \draw (v5) -- (u) ;
      \draw (v4) -- (u) ;
      \draw (v3) -- (u) ;
      \draw (v2) -- (u) ;
      \draw (v1) -- (u) ;
      \draw (v0) -- (u) ;
      \draw (v7) -- (u) ;
      \draw (v5) -- (v6);
      \draw (v5) -- (v4);
      \draw (v4) -- (v3);
      \draw (v3) -- (v2);
      \draw (v2) -- (v1);
      \draw (v1) -- (v0);
      \draw (v0) -- (v7);
      \draw (v6) -- (v7);
    \end{tikzpicture}
    \caption{Configuration $C_{\ref{C15.5a}}$}
    \label{fig:C155ainit}
  \end{figure}
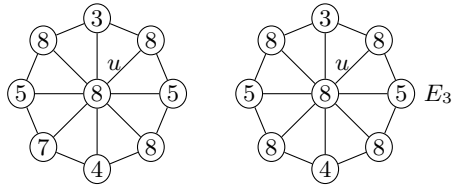

\item[$\bullet$] $\config{C15.5b}$ is a triangulated $8$-vertex $u$
  with a weak 3-neighbor $v_1$, a weak 4-neighbor $v_2$ such that
  $\dist_u(v_1,v_2)=2$, and a weak neighbor $v_3$ of degree 5 such
  that either $v_2$ is a $(7,8)$-neighbor of $u$, or $v_3$ is a
  $(6,6)$-neighbor, or $v_3$ is an $E_3$-neighbor of $u$ and $u$ has
  another weak neighbor of degree 5, see Figure~\ref{fig:C155binit}.
  \begin{figure}[!h]
    \centering
    \begin{tikzpicture}[v/.style={draw=black,minimum size = 10pt,ellipse,inner sep=1pt}]
      \node[v,label=85:{$u$}] (u) at (0,0)  {8};
      \node[v] (v0) at (45:1) {8};
      \node[v] (v1) at (90:1) {3};
      \node[v] (v2) at (135:1) {8};
      \node[v] (v3) at (180:1) {4};
      \node[v] (v4) at (225:1) {7};
      \node[v] (v6) at (315:1) {5};
      \draw (v6) -- (u) ;
      \draw (v4) -- (u) ;
      \draw (v3) -- (u) ;
      \draw (v2) -- (u) ;
      \draw (v1) -- (u) ;
      \draw (v0) -- (u) ;
      \draw (v4) -- (v3) ;
      \draw (v3) -- (v2) ;
      \draw (v2) -- (v1) ;
      \draw (v1) -- (v0) ;
      \tikzset{xshift=4cm}
      \node[v,label=85:{$u$}] (u) at (0,0)  {8};
      \node[v] (v0) at (45:1) {8};
      \node[v] (v1) at (90:1) {3};
      \node[v] (v2) at (135:1) {8};
      \node[v] (v3) at (180:1) {4};
      \node[v] (v4) at (225:1) {8};
      \node[v] (v5) at (270:1) {6};
      \node[v] (v6) at (315:1) {5};
      \node[v] (v7) at (0:1) {6};
      \draw (v6) -- (u) ;
      \draw (v5) -- (u) ;
      \draw (v4) -- (u) ;
      \draw (v3) -- (u) ;
      \draw (v2) -- (u) ;
      \draw (v1) -- (u) ;
      \draw (v0) -- (u) ;
      \draw (v7) -- (u) ;
      \draw (v5) -- (v6) ;
      \draw (v4) -- (v3) ;
      \draw (v3) -- (v2) ;
      \draw (v2) -- (v1) ;
      \draw (v1) -- (v0) ;
      \draw (v6) -- (v7) ;
      \tikzset{xshift=4cm}
      \node[v,label=85:{$u$}] (u) at (0,0)  {8};
      \node[v] (v0) at (45:1) {8};
      \node[v] (v1) at (90:1) {3};
      \node[v] (v2) at (135:1) {8};
      \node[v] (v3) at (180:1) {4};
      \node[v] (v4) at (225:1) {8};
      \node[v,label= below:{$E_3$}] (v5) at (270:1) {5};
      \node[v,label=right:{\scriptsize weak}] (v7) at (0:1) {5};
      \draw (v5) -- (u) ;
      \draw (v4) -- (u) ;
      \draw (v3) -- (u) ;
      \draw (v2) -- (u) ;
      \draw (v1) -- (u) ;
      \draw (v0) -- (u) ;
      \draw (v7) -- (u) ;
      \draw (v4) -- (v3) ;
      \draw (v3) -- (v2) ;
      \draw (v2) -- (v1) ;
      \draw (v1) -- (v0) ;
    \end{tikzpicture}
    \caption{Configuration $C_{\ref{C15.5b}}$}
    \label{fig:C155binit}
  \end{figure}
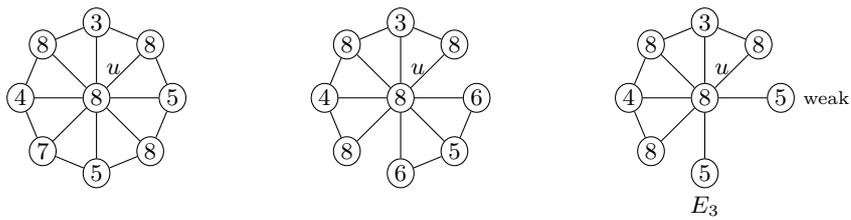

\item[$\bullet$] $\config{C16}$ is an $8$-vertex $u$ with two weak
  neighbors of degree $3$ and another neighbor of degree at most $5$,
  see Figure~\ref{fig:C16init}.
\begin{figure}[!h]
    \centering
    \begin{tikzpicture}[every node/.style={draw=black,minimum size = 10pt,ellipse,inner sep=1pt}]
      \node[label=85:{$u$}] (u) at (0,0)  {8};
      \node[label=above:{\scriptsize weak}] (v1) at (90:1) {3};
      \node[label=above:{\scriptsize weak}] (v3) at (210:1) {3};
      \node (v5) at (330:1) {5};
      \draw (v5) -- (u);
      \draw (v1) -- (u) -- (v3);
    \end{tikzpicture}
\caption{Configuration $C_{\ref{C16}}$}
    \label{fig:C16init}
  \end{figure}
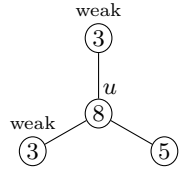
\item[$\bullet$] $\config{C17}$ is an $8$-vertex $u$ with a weak and a
  semi-weak neighbor $v_1,v_2$ of degree $3$ and adjacent to two
  vertices $w_1,w_2$ such that $(d(w_1),d(w_2))$ is $(4,7)$ or
  $(5,6)$, see Figure~\ref{fig:C17init}.
\begin{figure}[!h]
    \centering
    \begin{tikzpicture}[every node/.style={draw=black,minimum size = 10pt,ellipse,inner sep=1pt}]
      \node[label=above right:{$u$}] (u) at (0,0)  {8};
      \node[label=above:{\scriptsize weak}] (v1) at (90:1) {3};
      \node[label=below:{\scriptsize semi-weak}] (v3) at (180:1) {3};
      \node (v5) at (270:1) {4};
      \node (v7) at (0:1) {7};
      \draw (v5) -- (u) -- (v7);
      \draw (v1) -- (u) -- (v3);

      \node[xshift=4cm,label=above right:{$u$}] (u) at (0,0)  {8};
      \node[xshift=4cm,label=above:{\scriptsize weak}] (v1) at (90:1) {3};
      \node[xshift=4cm,label=below:{\scriptsize semi-weak}] (v3) at (180:1) {3};
      \node[xshift=4cm] (v5) at (270:1) {5};
      \node[xshift=4cm] (v7) at (0:1) {6};
      \draw (v5) -- (u) -- (v7);
      \draw (v1) -- (u) -- (v3);
    \end{tikzpicture}
\caption{Configuration $C_{\ref{C17}}$}
    \label{fig:C17init}
  \end{figure}
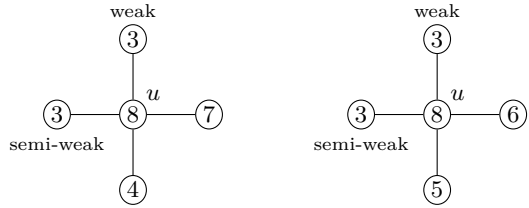

\item[$\bullet$] $\config{C21}$ is an $8$-vertex $u$ with a weak
  neighbor of degree $3$ and four neighbors of degree $4,4,5,7$, see Figure~\ref{fig:C21init}.
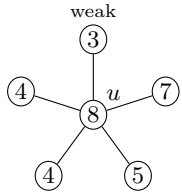
\begin{figure}[!h]
    \centering
    \begin{tikzpicture}[every node/.style={draw=black,minimum size = 10pt,ellipse,inner sep=1pt}]
      \node[label=above right:{$u$}] (u) at (0,0)  {8};
      \node[label=above:{\scriptsize weak}] (v1) at (90:1) {3};
      \node (v3) at (162:1) {4};
      \node (v5) at (234:1) {4};
      \node (v6) at (306:1) {5};
      \node (v7) at (18:1) {7};
      \draw (v5) -- (u) -- (v7);
      \draw (v1) -- (u) -- (v3);
      \draw (u) -- (v6);
    \end{tikzpicture}
    \caption{Configuration $C_{\ref{C21}}$}
    \label{fig:C21init}
  \end{figure}

\item[$\bullet$] $\config{C18}$ is an $8$-vertex $u$ with a weak
  neighbor of degree $3$, two weak neighbors of degree $4$ and a
  neighbor of degree 7, see Figure~\ref{fig:C18init}.
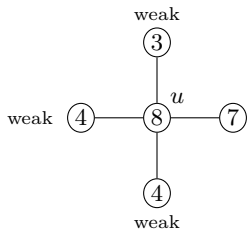
\begin{figure}[!h]
    \centering
    \begin{tikzpicture}[every node/.style={draw=black,minimum size = 10pt,ellipse,inner sep=1pt}]
      \node[label=above right:{$u$}] (u) at (0,0)  {8};
      \node[label=above:{\scriptsize weak}] (v1) at (90:1) {3};
      \node[label=left:{\scriptsize weak}] (v3) at (180:1) {4};
      \node[label=below:{\scriptsize weak}] (v5) at (270:1) {4};
      \node (v7) at (0:1) {7};
      \draw (v5) -- (u) -- (v7);
      \draw (v1) -- (u) -- (v3);
    \end{tikzpicture}
\caption{Configuration $C_{\ref{C18}}$}
    \label{fig:C18init}
  \end{figure}
\item[$\bullet$] $\config{C19}$ is an $8$-vertex $u$ with a neighbor
  of degree~7 and either three semi-weak neighbors of degree $3$ or
  two semi-weak neighbors of degree $3$ and two neighbors of degree
  $4$, see Figure~\ref{fig:C19init}.
  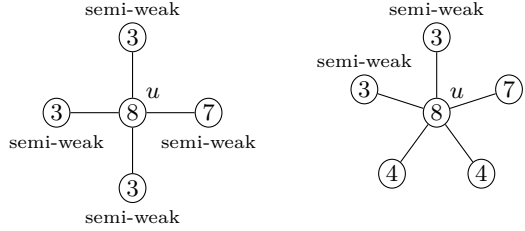
\begin{figure}[!h]
    \centering
   \begin{tikzpicture}[every node/.style={draw=black,minimum size = 10pt,ellipse,inner sep=1pt}]
      \node[label=above right:{$u$}] (u) at (0,0)  {8};
      \node[label=above:{\scriptsize semi-weak}] (v1) at (90:1) {3};
      \node[label=below:{\scriptsize semi-weak}] (v3) at (180:1) {3};
      \node[label=below:{\scriptsize semi-weak}] (v5) at (270:1) {3};
      \node[label=below:{\scriptsize semi-weak}] (v7) at (0:1) {7};
      \draw (v5) -- (u) -- (v7);
      \draw (v1) -- (u) -- (v3);

      \node[xshift=4cm,label=above right:{$u$}] (u) at (0,0)  {8};
      \node[xshift=4cm,label=above:{\scriptsize semi-weak}] (v1) at (90:1) {3};
      \node[xshift=4cm,label=above:{\scriptsize semi-weak}] (v3) at (162:1) {3};
      \node[xshift=4cm] (v5) at (234:1) {4};
      \node[xshift=4cm] (v7) at (306:1) {4};
      \node[xshift=4cm] (v8) at (18:1) {7};
      \draw (v5) -- (u) -- (v7);
      \draw (v1) -- (u) -- (v3);
      \draw (u) -- (v8);
    \end{tikzpicture}
\caption{Configuration $C_{\ref{C19}}$}
    \label{fig:C19init}
  \end{figure}
\end{itemize}

\section{Reduction techniques}
\label{sec:tech}

\subsection{Framework}

We now introduce the generic framework we use to prove that a given
configuration is reducible. Reducing a configuration $C_i$ means to
take a list assignment $L$ of $G$, to find a suitable subgraph $G'$ of
$G$ (often constructed by removing elements of $G$ creating $C_i$),
and to extend any $L$-coloring of $G'$ to $G$. Since $G$ is a minimum
counterexample, we get a contradiction if $G$ contains $C_i$.

There are two non-immediate steps in this proof scheme: first, we have
to find the right subgraph $G'$. Then, the most difficult part is to
extend the coloring. Note that in some cases, we may have to change
the given coloring before extending it. 

We first introduce some terminology. In the previous setting, a
\emph{constraint} for an element $x$ of $G$ (vertex or edge) is an
already colored element $y$ such that $x$ and $y$ are incident (or
adjacent). The \emph{total graph} of $G$ is the graph denoted by
$\mathcal{T}(G)$, whose vertices are $V(G)\cup E(G)$, and there is an
edge between any two elements $x$ and $y$ such that $x$ and $y$ are
adjacent vertices or incident elements of $G$. Observe that finding a
total $L$-coloring of $G$ is equivalent to finding an $L$-coloring of
$\mathcal{T}(G)$.

Given an element $x$ of $G$, we denote by $c_x$ the number of
constraints of $x$, and by $\hat{x}$ a list of $10-c_x$ colors chosen
arbitrarily among the available colors for $x$ after having colored
$G'$. We denote by $\mathcal{T}(G\setminus G')$ the subgraph of
$\mathcal{T}(G)$ induced by the elements that are not already colored,
i.e. the elements of $G\setminus G'$. Note that in order to extend the
coloring from $G'$ to $G$, it is sufficient to produce an
$L'$-coloring of $\mathcal{T}(G\setminus G')$ where $L'$ is defined by
$L'(x)=\hat{x}$ for every element $x$ of $G\setminus G'$. Note that sometimes, we will start from a coloring of a subgraph of $G$ and remove the color of some elements to obtain $G'$. 

In order to limit the number of variables used when reducing
configurations, we will think of the lists $\hat{x}$ as dynamic
objects, meaning that they may decrease each time we color some
element. The following observation allows us to assume when
appropriate that, when we color an uncolored element $x$ of $G$, the
lists of all its neighbors in $\mathcal{T}(G)$ always lose a color.

\begin{remark}
  Let $x,y$ be adjacent elements in $\mathcal{T}(G)$. Unless otherwise
  stated (i.e. if we assume explicitly that the color of $x$ does not
  appear in $\hat{y}$, for example if $\hat{x}$ and $\hat{y}$ are
  disjoint), coloring $x$ makes $|\hat{y}|$ decrease by $1$. Moreover, if $y$ is element whose color $\alpha$ was removed when creating $G'$, then unless $x$ gets color $\alpha$, we assume that $\alpha$ remains in $\hat{y}$ after coloring $x$.
\end{remark}

We sometimes \emph{forget} elements. Forgetting $x$ means that for
every coloring of its neighbors in $\mathcal{T}(G)$, we can always
find an available color for $x$. For example, this happens when $x$
has more available colors than uncolored neighbors in
$\mathcal{T}(G)$. Therefore, when we forget $x$, we postpone the
coloring of $x$ to the end of the coloring process: we implicitly
assign a color to $x$ when all the remaining elements are colored. We
extend this terminology to lists of elements: forgetting
$x_1,\ldots,x_p$ means that we forget $x_1$, then $x_2$, \dots, then
$x_p$ (observe that the order matters). Note that we can always forget
uncolored vertices of degree at most $4$ in $G$, since they have at
most eight neighbors in $\mathcal{T}(G)$.

\subsection{Combinatorial Nullstellensatz}

Most of the proofs of this paper rely on more or less involved case
analyses, depending on the lists $\hat{x}$. This may lead to rather
long proofs. To deal with this issue, we introduce another approach to
reduce the corresponding configurations. As we will see, this method
relies on an algebraic criterion that can be computer checked. This
leads to much shorter proofs, with the downside of not being
human-checkable. We now describe how to reduce a given configuration,
or more precisely how to extend a coloring from a subgraph of $G$ to
$G$ itself. The method uses the Combinatorial Nullstellensatz stated
below.

\begin{theorem}[\cite{alon1992colorings}]
  \label{thm:nss}
  Let $\mathbb{K}$ be a field, and $P\in\mathbb{K}[X_1,\ldots,X_n]$ a
  multivariate polynomial. Let $X_1^{a_1}\cdots X_n^{a_n}$ be a
  monomial with a non-zero coefficient in $P$, and of maximal
  degree. Then, for any family $S_1,\ldots,S_n$ of subsets of
  $\mathbb{K}$ satisfying $|S_i|>a_i$ for $i=1,\ldots, n$, there
  exists a non-zero value of $P$ in $S_1\times\cdots \times S_n$.
\end{theorem}

At first glance, this result has nothing to do with graph coloring. We
present here a standard generic framework that translates graph
coloring into an algebraic problem that can be solved using this
result. We also explain how to use it to reduce configurations.

For each uncolored element $x$ in $G$ which is not forgotten, we
associate a variable $X$ (we use the same letter but capitalized). We
fix an arbitrary order $<$ on the variables. The polynomial $P_{G'}$
is then defined as the product of all $(X-Y)$ when $X<Y$ and $x$ and
$y$ are adjacent and not forgotten uncolored vertices of
$\mathcal{T}(G)$. Observe that each possible coloring of $G$ gives an
evaluation of $P_{G'}$ by replacing each variable $X$ with the color
of the corresponding element $x$. Moreover, observe that $P_{G'}$ does
not evaluate to 0 if and only if the corresponding coloring is proper,
i.e. if the coloring of $G'$ extends to $G$. Since we now look for a
non-zero value of $P_{G'}$, applying Theorem~\ref{thm:nss} to the
subsets $\hat{x}$ gives a sufficient condition in terms of the
monomials in $P_{G'}$: to prove that the coloring extends from $G'$ to
$G$, it is sufficient to find a monomial $m$ in $P_{G'}$ such that the
three following conditions hold:
\begin{enumerate}
\item $\deg(m)=\deg(P)$.
\item $\deg_X(m) < |\hat{x}|$ for every uncolored and not forgotten
  element $x$ of $G$.
\item The coefficient of $m$ in $P$ is non-zero.
\end{enumerate}

Therefore, proving that a configuration is reducible using the
Combinatorial Nullstellensatz amounts to finding a suitable monomial
in $P_{G'}$. For the sake of readability, we do not state the
polynomial $P_{G'}$ in each of the reduction proofs.

Note that we do not believe that finding a suitable monomial, as well
as checking Condition 3, can be done without a computer. For the
former problem, we use an exhaustive search algorithm that produces an
output in a reasonable time on most of the instances, but not for all,
hence we do not have a reduction proof using Combinatorial
Nullstellensatz for each configuration. To check Condition 3, a Maple
code is available at
\href{https://github.com/tpierron/Delta8}{\url{https://github.com/tpierron/Delta8}}.

Finally, observe that Theorem~\ref{thm:nss} is not an equivalence in
general: a polynomial may satisfy the conclusion of the theorem even
if it has no suitable monomial. However, we do not know whether there
exist reducible configurations such that the associated polynomial
contains no suitable monomial.

\subsection{Recoloring approach}

Together with the two techniques previously discussed (case analysis
and Nullstellensatz), we use a third one to reduce the
configurations. This technique is based on the idea that we can
sometimes permute the colors of the surrounding of a configuration in
order to make our life easier while extending a coloring to the full
graph. This idea is definitely not new: it has already been used many
times to obtain graph colouring results. We build here upon the method
used to reduce a few configurations in~\cite{bonamy}. In this
subsection, we develop a generic framework allowing us to use
recoloring arguments in our setting.

We start with a preliminary definition. Let $L$ be a list assignment
on $\mathcal{T}(G)$ and $\gamma$ a partial $L$-coloring of
$\mathcal{T}(G)$. Let $S$ be a properly colored clique in
$\mathcal{T}(G)$. The \emph{color shifting graph} of $S$ with respect
to $\gamma$ is the loopless digraph $H_{S,\gamma}$ defined as follows:
\begin{itemize}
\item[$\bullet$] Each element of $S$ is a vertex of $H_{S,\gamma}$.
\item[$\bullet$] We add a vertex $s_\alpha$ to $H_{S,\gamma}$ for each
  color $\alpha\in\cup_{x\in S}\hat{x}$, where $\hat{x}$ is the set of
  available colors for $x$ when we uncolor $S$.
\item[$\bullet$] If $x,y\in S$ with $x\neq y$, there is an arc
  $x\to y$ if the color of $x$ lies in $\hat{y}$ once $S$ is
  uncolored.
\item[$\bullet$] For any $x,\alpha$, there is an arc $s_\alpha\to x$
  if $\alpha\in \hat{x}$ and $\alpha\notin \gamma(S)$ which means that
  the color $\alpha$ could replace the color of $x$.
\item[$\bullet$] For any $x,\alpha$, there is an arc $x\to s_\alpha$.
\item[$\bullet$] For any $\alpha\neq\beta$, there is an arc
  $s_\alpha\to s_\beta$.
\end{itemize}

The terminology comes from the fact that any directed cycle in
$H_{S,\gamma}$ allows us to shift the colors of the elements of $S$ as
stated in the following lemma.
\begin{lemma}
  \label{lem:recolor}
  Let $L$ be a list assignment of $\mathcal{T}(G)$, let $\gamma$ be a
  partial $L$-coloring of $\mathcal{T}(G)$ and $S$ be a colored clique
  of $\mathcal{T}(G)$. Assume that there is a directed cycle
  $x_1\to\cdots\to x_n\to x_1$ in the color shifting graph
  $H_{S,\gamma}$.

  Then there exists a partial $L$-coloring $\gamma'$, defined on the
  same elements of $\mathcal{T}(G)$ as $\gamma$, and that differs from
  $\gamma$ exactly on $S\cap\{x_1,\ldots,x_n\}$.
\end{lemma}

\begin{proof}
  We define $\gamma'$ by taking $\gamma'(x)=\gamma(x)$ for all the
  vertices $x$ of $S$ outside the directed cycle. It remains to define
  $\gamma'$ on $S\cap\{x_1,\ldots,x_n\}$.

  If none of the $x_i$'s is some $s_\alpha$, we move the colors
  following the arrows: for $1\leqslant i\leqslant n$, we define
  $\gamma'(x_{i+1})=\gamma(x_i)$ (the indices are taken modulo
  $n$). This is allowed since we have $\gamma(x_i)\in \hat{x_{i+1}}$
  by the definition of the arc $x_i\to x_{i+1}$. Moreover, $\gamma'$
  is still a proper coloring since the color $\gamma(x_i)$ appears in
  $S$ only on $x_i$ according to $\gamma$ since $S$ is a clique in
  $G$, hence it appears only on $x_{i+1}$ according to $\gamma'$.

  Otherwise, we decompose the directed cycle into (maximal) directed
  paths of the form $s_\alpha\to x_1\to \cdots\to x_p$. We then apply
  a similar approach to each of these paths: for
  $2\leqslant i\leqslant p$, we define
  $\gamma'(x_i)=\gamma(x_{i-1})$. Similarly, this gives a proper
  coloring. It remains to color $x_1$. Note that $s_\alpha\to x_1$, so
  $\alpha\in \hat{x_1}$ and $\alpha\notin\gamma(S)$. Therefore, we can
  take $\gamma'(x_1)=\alpha$ and keep a proper coloring.

  Since we consider a cycle, for every $\alpha$, the vertex $s_\alpha$
  is the source of at most one such sub-path of the directed
  cycle. Therefore, color $\alpha$ appears in at most one vertex of
  $S$ according to $\gamma'$, and the coloring $\gamma'$ is proper.

  In both cases, we thus obtain a proper coloring $\gamma'$ satisfying
  $\gamma'(x)\neq \gamma(x)$ for each vertex $x$ of the considered
  directed cycle.
\end{proof}

We are now ready to describe the generic way used to reduce
configurations in this approach. The framework is the same as before:
our goal is to extend a coloring of a subgraph $G'$ of $G$ to the
entire graph $G$. To this end, we first identify some conditions on
the color lists impeding the coloring to extend directly to $G$. If
these conditions are not satisfied, then we can extend the coloring
and we are done.

On the contrary, if they are satisfied, we look for some elements of
$G'$ to recolor in order to change the available colors of the
uncolored elements of $G$, and hence break the previous conditions. We
finally use the previous lemma to reduce the initial problem to
finding a suitable directed cycle in the color shifting graph of a
well-chosen set of elements.

To find such directed cycles, we first state a simple but useful
property of color shifting graphs: if $H_{S,\gamma}$ is the color
shifting graph of a set $S$ with respect to $\gamma$, then the
in-degree of any vertex $x\in S$ of $H_{S,\gamma}$ is at least
$|\hat{x}|-1$. We often use this property together with the following
lemma to find the required directed cycles. Recall that a strong
component of $H_{S,\gamma}$ is a maximal set of vertices $C$ such that
any two of them are linked by a directed path in $C$.

\begin{lemma}
  \label{lem:SCC}
  Every simple directed graph $H$ has a strong component $C$
  satisfying
  \[|C|>\max_{x\in C} d^-(x)\]
\end{lemma}

\begin{proof}
  Consider the graph $\pi(H)$ obtained by contracting each strong
  component of $H$ to a single vertex.

  Note that $\pi(H)$ is an acyclic digraph, therefore it contains a
  vertex $C$ of in-degree $0$. Take $x\in C$. Then note that due to
  the definition of $\pi(H)$, for each arc $y\to x$, we also have
  $y\in C$. Therefore $C$ contains every in-neighbor of $x$. Since $H$
  is a simple graph, there are $d^-(x)$ such in-neighbors, and $x$ is
  not one of them. Thus $|C|> d^-(x)$. This is valid for any
  $x\in C$, thus we obtain the result.
\end{proof}

Our goal is to prove that the elements we want to recolor are not
alone in their strong component in the color shifting graph we
consider (so that one of these elements is contained in a directed
cycle, and we can recolor it using Lemma~\ref{lem:recolor}). With the
previous result, we have a case split: if there is a strong component
containing a vertex with large in-degree, then it is a large
component, and it is likely to contain an element we want to
recolor. Otherwise, we remove all the vertices with large in-degree
and apply recursively the same argument until (hopefully) a suitable
directed cycle is found.

In order to apply this method, we need to compute the in-degree of
every vertex in a color shifting graph. This is the goal of the last
lemma of this section.

\begin{lemma}
  \label{lem:degmin}
  Let $L$ be a list assignment of $\mathcal{T}(G)$, let $\gamma$ be a
  partial $L$-coloring of $\mathcal{T}(G)$ and $S$ be a colored clique
  of $\mathcal{T}(G)$. Let $x$ be a vertex of $H_{S,\gamma}$. We have
  \[d^-(x)=\begin{cases}  |\hat{x}|-1&\text{ if } x\in S\\ |V(H_{S,\gamma})|-1&\text{ otherwise.}
    \end{cases}
  \]
\end{lemma}

\begin{proof}
  Let $x\in V(H_{S,\gamma})$. If $x$ is some $s_\alpha$, then by
  definition, there is an arc $y\to x$ for every other vertex $y$ of
  $H_{S,\gamma}$. Therefore, $d^-(x)=|V(H_{S,\gamma})|-1$.

  Otherwise, assume that $x\in S$. By definition, every in-neighbor of
  $x$ is either an element of $S$ colored with some color in
  $\hat{x}$, or a vertex $s_\alpha$ with
  $\alpha\in\hat{x}\setminus\gamma(S)$. Observe that since $\gamma$ is
  proper and $S$ is a clique, then for every $\alpha\in \gamma(S)$,
  there is exactly one vertex of $S$ colored with $\alpha$. In
  particular, there is no vertex $y\neq x$ with
  $\gamma(y)=\gamma(x)$. Therefore, there is one in-neighbor of $x$
  for every color of $\hat{x}\setminus\{\gamma(x)\}$.

  Conversely, let $\alpha\in\hat{x}\setminus\{\gamma(x)\}$. If
  $\alpha$ does not appear on $S$, i.e. $\alpha\notin\gamma(S)$, then
  we have an arc $s_\alpha\to x$ in $H_{S,\gamma}$. Otherwise,
  $\alpha=\gamma(y)$ for some $y\in S\setminus\{x\}$, and we have an
  arc $y\to x$ in $H_{S,\gamma}$.

  Therefore, the number $d^-(x)$ of in-neighborss of $x$ is
  $|\hat{x}\setminus\{\gamma(x)\}|=|\hat{x}|-1$.
\end{proof}

\subsection{Generic reducible patterns}

We conclude this section by giving some other generic (and standard)
tools that are used many times in the reductions.

Recall that proving that a configuration is reducible amounts to
extending a coloring of a subgraph $G'$ of $G$ to the entire graph
$G$. This can be rephrased in terms of \emph{$f$-choosability}. This
variant of the choosability problem is defined as follows. Let $H$ be
a graph and $f:V(H)\to\mathbb{N}$. We say that $H$ is $f$-choosable if
we can produce a vertex $L$-coloring of $H$ from any list assignment
$L$ satisfying $|L(v)|\geqslant f(v)$ for every vertex $v$ of $G$.

To extend a coloring from $G'$ to $G$, we often prove that
$\mathcal{T}(G\setminus G')$ is $f$-choosable, where $f(x)$ is the
number of available colors of the element $x$ (in our case, $f(x)$ is
ten minus the number of elements of $G'$ incident to $x$). This point
of view gives another tool to extend colorings, as shown by the
following theorem.
\begin{theorem}[\cite{borodin67,erdos1979choosability}]
  \label{thm:=deg}
  Let $G$ be a connected graph such that at least one of its blocks is
  neither a complete graph nor an odd cycle. For any function
  $f:V(G)\to \mathbb{N}$ such that $f(v)\geqslant d(v)$ for each
  vertex $v$, the graph $G$ is $f$-choosable.
\end{theorem}

Despite the fact that Theorem~\ref{thm:=deg} is about vertex
choosability while we focus on total choosability,
Theorem~\ref{thm:=deg} will turn out to be helpful when looking at the
constraint graphs.

As a consequence, we get this classical result about choosability of
even cycles.
\begin{corollary}
  \label{cor:evencycle}
  Any even cycle is $2$-choosable.
\end{corollary}

We introduce some other useful results. The first one is based on
Corollary~\ref{cor:evencycle}.

\begin{lemma}
  \label{lem:fryingpan}
  Let $G$ be the graph composed of a cycle $v_1\cdots v_nv_1$ such
  that $v_1,v_n$ share a common neighbor $u$, see
  Figure~\ref{fig:fryingpan}.  Let $L$ be a list assignment satisfying
  that for every vertex $v$, we have $|L(v)|\geqslant 2$. Then $G$ is
  $L$-choosable if either $|L(v_1)|\geqslant 3$ or $n$ is even and
  $L(v_1)\neq L(u)$.
\end{lemma}

\begin{figure}[!ht]
  \centering
  \begin{tikzpicture}[every node/.style={draw=black,minimum size = 10pt,ellipse,inner sep=1pt},node distance=1.5cm]
    \node (v1) at (234:1) {$v_1$};
    \node (v2) at (306:1) {$v_2$};
    \node (vn1) at (90:1) {$v_{n-1}$};
    \node (vn) at (162:1) {$v_n$};
    \node (u) at (198:2)  {$u$};
    \draw (v2) -- (v1) -- (u) -- (vn) -- (v1);
    \draw (vn) -- (vn1);
    \draw[dotted,bend left=40] (vn1) to (v2);
  \end{tikzpicture}
  \caption{Configuration of Lemma~\ref{lem:fryingpan}}
  \label{fig:fryingpan}
\end{figure}

\begin{proof}
  Without loss of generality, we may assume that for any $v\neq v_1$,
  $|L(v)|=2$, and that $|L(v_1)|$ is 2 or 3. First assume that the
  cycle has odd length, thus $|L(v_1)|=3$. If $L(v_n)=L(u)$, we color
  $v_1$ with a color not in $L(u)$, then
  $v_2,\ldots,v_n,u$. Otherwise, we color $v_n$ with a color not in
  $L(u)$, then color $v_{n-1},\ldots,v_1,u$.

  If the cycle has even length, we distinguish two cases:
  \begin{itemize}
  \item $L(v_2)=\cdots=L(v_n)$, then we color $v_2,v_4,\ldots,v_n$
    with a color, $v_3,v_5\ldots,v_{n-1}$ with another color. Denote
    by $\hat{L}$ the list assignment obtained from $L$ by removing the
    colors of the neighbors of each vertex. Observe that we have
    $|\hat{L}(v_1)|=1$ or $2$. If $|\hat{L}(v_1)|=2$, we color $u$
    then $v_1$. Otherwise, since $\hat{L}(v_1)\neq \hat{L}(u)$, we can
    color $v_1$ then $u$.
  \item Otherwise, there exists $i$ such that $L(v_i)\neq
    L(v_{i+1})$. Color $v_{i+1}$ with a color not in $L(v_i)$, then
    color $v_{i+1},\ldots,v_n$. With $\hat{L}$ defined as previously,
    we now have $|\hat{L}(u)|=1$ and $|\hat{L}(v_1)|=1$ or $2$. If
    $|\hat{L}(v_1)|=2$, we color $u,v_1,v_2,\ldots,v_i$. Otherwise,
    since we have $L(u)\neq L(v_1)$, we can color $u$ with a color not
    in $L(v_1)$, then $v_1,v_2,\ldots,v_i$.
\end{itemize}
\end{proof}

The next result is a consequence of Hall's necessary and sufficient
condition for a perfect matching to exist in a bipartite
graph. Finding an $L$-coloring of a graph $G$ can be reduced to
finding a perfect matching in the following graph. It has one vertex
per color $c$ and per vertex $x$ of $G$, and an edge $(c,x)$ when
$c\in L(x)$. Since this graph is bipartite, Hall's criterion gives a
condition for an $L$-coloring to exist.

\begin{theorem}[Hall's marriage theorem]
  \label{thm:clique}
  Let $G$ be a clique. Then for every list assignment $L$, the graph
  $G$ is $L$-choosable if and only if for all $S\subset V(G)$,
  $|S|\leqslant |\cup_{x\in S} L(x)|$.
\end{theorem}

We end this subsection with a last configuration, depicted in
Figure~\ref{fig:diam}.
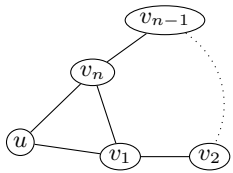
\begin{figure}[!ht]
  \centering
  \begin{tikzpicture}[every node/.style={draw=black,minimum size = 10pt,ellipse,inner sep=1pt},scale=0.66]
    \node (v1) at (0,0) {$v_1$};
    \node (v2) at (1,1.72) {$v_2$};
    \node (v3) at (2,0) {$v_3$};
    \node (v4) at (3,1.72) {$v_4$};
    
    \node (vn3) at (5,0) {$v_{n-3}$};
    \node (vn2) at (6,1.72) {$v_{n-2}$};
    \node (vn1) at (7,0) {$v_{n-1}$};
    \node (vn) at (8,1.72) {$v_n$};
    \draw (v1) -- (v2) -- (v3) -- (v4);
    \draw (vn3) -- (vn2) -- (vn1) -- (vn);
    \draw (v1) -- (v3);
    \draw (v2) -- (v4);
    \draw (vn3) -- (vn1);
    \draw (vn2) -- (vn);
    \draw[dotted] (v4) -- (vn2);
    \draw[dotted] (v3) -- (vn3);
  \end{tikzpicture}
  \caption{Configuration of Lemma~\ref{lem:diam}}
  \label{fig:diam}
\end{figure}

\begin{lemma}
  \label{lem:diam}
  Let $n\geqslant 4$ be an integer such that $n\not\equiv 0\mod
  3$. Let $G$ be the graph formed by a path $v_1\ldots v_n$ with
  additional edges $v_iv_{i+2}$ for $1\leqslant i\leqslant n-2$ (see
  Figure~\ref{fig:diam}). Let $L$ be a list assignment $L$ such that
  $|L(v)|\geqslant 2$ for $v\in\{v_1,v_{n-1},v_n\}$, and
  $|L(v)|\geqslant 3$ for any other $v$. Then $G$ is $L$-choosable.
\end{lemma}

\begin{proof}
  We proceed by induction on $n$.
  \begin{itemize}
  \item Assume $n=4$. If $L(v_3)=L(v_4)$, we color $v_2$ with a color
    not in $L(v_3)$, then $v_1,v_3,v_4$. Otherwise, we color $v_3$
    with a color not in $L(v_4)$, then $v_1,v_2,v_4$.
  \item Assume $n=5$. If $L(v_4)=L(v_5)$, we color $v_3$ with a color
    not in $L(v_4)$, then $v_1,v_2,v_4,v_5$. Otherwise, we color $v_5$
    with a color not in $L(v_4)$ and use the case $n=4$ to color
    $v_1,\ldots,v_4$.
  \item Assume $n>6$. If $L(v_{n-1})=L(v_n)$, we color $v_{n-2}$ with
    a color not in $L(v_n)$, then apply the case $n-3$ to color
    $v_1,\ldots,v_{n-3}$, then color $v_{n-1}$ and $v_n$. Otherwise if
    $n\equiv 2\mod 3$, we color $v_n$ with a color not in
    $L(v_{n-1})$, then apply the use case $n-1$ to color
    $v_1,\ldots,v_{n-1}$. If $n\equiv 1\mod 3$, color $v_{n-1}$ with a
    color not in $L(v_n)$, then use case $n-2$ to color
    $v_1,\ldots,v_{n-2}$, then color $v_n$.
  \end{itemize}
\end{proof}

\section{The configurations are reducible}
\label{sec:reduction}
In this section, we prove that configurations $C_{\ref{C1}}$ to
$C_{\ref{C19}}$ are reducible. We devote a subsection to each
configuration. We use the recoloring approach to reduce
$C_{\ref{C16}}$ to $C_{\ref{C19}}$. We use $C_{\ref{C1}}$ and
$C_{\ref{C2}}$ as examples of the template for the case analysis and
the Nullstellensatz approachs. Then, for $C_{\ref{C3b}}$ to
$C_{\ref{C15}}$, we omit the case analysis argument whenever the
Nullstellensatz approach manages to conclude in a reasonable time. In
that case, we refer to
\href{https://github.com/tpierron/Delta8}{\url{https://github.com/tpierron/Delta8}}
for the suitable monomial(s), and to the arxiv version of this
paper~\cite{arxiv} for the reader who would be interested by the case
analysis.

\subsection{Configuration $C_{\ref{C1}}$}

\begin{lemma}
The graph $G$ does not contain $C_{\ref{C1}}$.
\end{lemma}

\begin{proof}
  Assume that $G$ has an edge $e=uv$ such that
  $d_G(u)+d_G(v) \leqslant 10$ and $d_G(u) \leqslant 4$.

  Color $G\setminus \{e\}$ by minimality and uncolor $u$. We may
  assume that $|\hat{e}|=1$ and $|\hat{u}|=3$. 

  We can extend the coloring to $G$ using the following argument. We
  first forget $u$. Then $uv$ has at most $d_G(u)+d_G(v)-1<10$
  constraints, so we can color it.

  We can also conclude using the Nullstellensatz: note that $P_G$ is
  $E-U$. Then the monomial $m=U$ satisfies:
  \begin{enumerate}
  \item $\deg(m)=1=\deg(P)$.
  \item $\deg_E(m)=0< 1 = |\hat{e}|$ and $\deg_U(m)=1<3=|\hat{u}|$.
  \item $m$ has coefficient $-1$ in $P$. 
  \end{enumerate}
  Hence we can color $G$ using Theorem~\ref{thm:nss}.
\end{proof}

\subsection{Configuration $C_{\ref{C2}}$}
\begin{lemma}
The graph $G$ does not contain $C_{\ref{C2}}$.
\end{lemma}

\begin{proof}
  Assume that $G$ has an even cycle $v_1\cdots v_{2n}v_1$ such that
  for $1\leqslant i\leqslant n$, $d(v_{2i-1})\leqslant 4$ and
  $d(v_{2i-1})+d(v_{2i})\leqslant 11$.

  Denote by $G'$ the graph obtained from $G$ by removing the edges of
  the cycle. Using the minimality of $G$, we can color $G'$. Remove
  the color of vertices with odd subscript, and forget them since they
  have degree at most $4$. Observe that each edge of the cycle has now
  $d(v_{2i})-1+d(v_{2i+1})-2=11-3=8$ constraints.

  By Corollary~\ref{cor:evencycle}, we can color the edges of the cycle
  and obtain a valid coloring of $G$.

  We can also conclude using the Nullstellensatz: we have
  $P_G=(E_1-E_2)\cdots (E_{2n-1}-E_{2n})(E_1-E_{2n})$ and
  $m=E_1\cdots E_{2n}$, where $e_1,\ldots,e_{2n}$ are the uncolored
  edges of $G$. We have:
  \begin{enumerate}
  \item $\deg(m)=2n=\deg(P_G)$.
  \item For $1\leqslant i\leqslant 2n$,
    $\deg_{E_i}(m)=1 < 2 = |\hat{e_i}|$.
  \item The coefficient of $m=E_1\cdots E_{2n}$ in $P_G$ is then $-2$.
  \end{enumerate}
  Using Theorem~\ref{thm:nss}, we can extend the coloring to $G$.
\end{proof}

\subsection{Configuration $C_{\ref{C3b}}$}

\begin{lemma}
  \label{lem:3b}
The graph $G$ does not contain $C_{\ref{C3b}}$.
\end{lemma}

\begin{proof}
  We name the elements according to Figure~\ref{fig:C3b}. By
  minimality, we color $G'=G\setminus\{a,c\}$ and we remove the
  color of $b,u,v,w$.
  \begin{figure}[!ht]
    \centering
    \begin{tikzpicture}
      \tikzset{v/.style={draw=black,minimum size = 10pt,ellipse,inner sep=1pt}}
      \node [v,label=left:{$v$}] (b) at (0,0) {$5$};
      \node [v,label=right:{$u$}] (f) at (0:1)  {$6$};
      \node [v,label=above:{$w$}] (d) at (60:1) {$5$};
      \draw (b) -- (f) node[midway, below]{$a$};
      \draw (b) -- (d) node[midway, above left]{$b$};
      \draw (d) -- (f) node[midway, above right]{$c$};
    \end{tikzpicture}
\caption{Notation for Lemma~\ref{lem:3b}}
    \label{fig:C3b}
  \end{figure}

  We have $|\hat{a}|=|\hat{c}|=3$, $|\hat{v}|=|\hat{w}|=|\hat{b}|=4$
  and $|\hat{u}|=2$. We conclude using Theorem~\ref{thm:nss}.
\end{proof}

\subsection{Configuration $C_{\ref{C3a}}$}

\begin{lemma}
The graph $G$ does not contain $C_{\ref{C3a}}$.
\end{lemma}

\begin{proof}
  Assume that $G$ contains a path $uvw$ such that
  $d(u)=d(v)=d(w)=5$. Note that we may assume that $uw\notin E(G)$ due
  to $C_{\ref{C3b}}$. We denote by $a,b$ the edges $uv$ and $vw$.

  We color by minimality the graph $G'$ obtained by removing $a$ and
  $b$ from $G$. Then we uncolor $u,v$ and $w$. We have $|\hat{v}|=4$,
  $|\hat{a}|=|\hat{b}|=3$ and $|\hat{u}|=|\hat{w}|=2$. We conclude
  using Theorem~\ref{thm:nss}.
\end{proof}

\subsection{Configuration $C_{\ref{C4}}$}
To prove that $G$ does not contain $C_{\ref{C4}}$, it is sufficient to
prove the three following lemmas, one for each possible minimal
triangle-distance between neighbors of $u$ satisfying the hypothesis
of $v_1$ and $v_2$.

\begin{lemma}
  \label{lem:4a}
The graph $G$ does not contain a $7$-vertex $u$ with two $(5,6)$-neighbors
  $v_1,v_2$ such that $uv_1v_2$ is a triangle.
\end{lemma}

\begin{proof}
  We use the notation depicted in Figure~\ref{fig:C4a}. By minimality,
  we color $G'=G\setminus\{a,b,c,d,e,f,g\}$ and uncolor
  $u,v_1,v_2,w_1,w_2$. By $C_{\ref{C3b}}$, there is no edge $w_1v_2$
  nor $w_2v_1$. Therefore, the only possible edge of $G$ not on the
  drawing is $w_1w_2$. We have $|\hat{d}|=|\hat{g}|=3$,
  $|\hat{a}|=|\hat{c}|=|\hat{u}|=4$, $|\hat{e}|=|\hat{f}|=5$ and
  $|\hat{v_1}|=|\hat{v_2}|=|\hat{b}|=6$. Moreover, if
  $w_1w_2\notin E(G)$, we have $|\hat{w_1}|=|\hat{w_2}|=2$ and
  $|\hat{w_1}|=|\hat{w_2}|=3$ otherwise.
  \begin{figure}[!h]
    \centering
    \begin{tikzpicture}[v/.style={draw=black,minimum size = 10pt,ellipse,inner sep=1pt}]
    \node [v,label=above:{$u$}] (u) at (0,0) {$7$};
    \node [v,label=left:{$w_1$}] (w1) at (193:1.5)  {$6$};
    \node [v,label=below left:{$v_1$}] (v1) at (244.5:1.5) {$5$};
    \node [v,label=below right:{$v_2$}] (v2) at (296:1.5) {$5$};
    \node [v,label=right:{$w_2$}] (w2) at (347.5:1.5) {$6$};
    \draw (u) -- (w1) node[midway,above] {$d$};
    \draw (u) -- (w2) node[midway,above] {$g$};
    \draw (w1) -- (v1) node[midway,below left] {$a$};
    \draw (u) -- (v1) node[midway,left] {$e$};
    \draw (u) -- (v2) node[midway,right] {$f$};
    \draw (w2) -- (v2) node[midway,below right] {$c$};
    \draw (v1) -- (v2) node[midway,above] {$b$};
  \end{tikzpicture}
\caption{Notation for Lemma~\ref{lem:4a}}
    \label{fig:C4a}
  \end{figure}
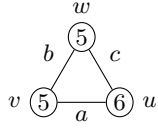
  We conclude using Theorem~\ref{thm:nss}.
\end{proof}

\begin{lemma}
  \label{lem:4b}
  The graph $G$ does not contain a $7$-vertex $u$ with four neighbors
  $v_1,v_2,v_3,v_4$ such that $d(v_1)=d(v_2)=d(v_4)=5$, $d(v_3)=6$ and
  $v_iv_{i+1}\in E(G)$ for $i=1,2,3$.
\end{lemma}

\begin{proof}
  We use the notation depicted in Figure~\ref{fig:C4b}. By minimality,
  we color $G'=G\setminus\{a,b,c,d,e,f,g\}$ and uncolor
  $u,v_1,v_2,v_3,v_4$. Due to $C_{\ref{C3a}}$, $v_1v_4$ and $v_2v_4$
  are not edges of $G$. Moreover, by $C_{\ref{C3b}}$,
  $v_1v_3\notin E(G)$. Therefore, all the edges between
  $u,v_1,\ldots,v_4$ in $G$ are drawn in the figure. We have
  $|\hat{v_1}|=|\hat{v_3}|=|\hat{v_4}|=|\hat{u}|=|\hat{c}|=|\hat{d}|=|\hat{f}|=|\hat{g}|=4$,
  $|\hat{a}|=|\hat{b}|=|\hat{e}|=5$ and $|\hat{v_2}|=6$.
  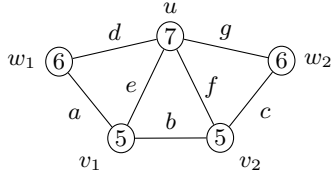
\begin{figure}[!h]
    \centering
   \begin{tikzpicture}[v/.style={draw=black,minimum size = 10pt,ellipse,inner sep=1pt}]
    \node[v,label=above:{$u$}] (u) at (0,0) {$7$};
    \node[v,label=left:{$v_1$}] (v1) at (193:1.5)  {$5$};
    \node[v,label=below left:{$v_2$}] (v2) at (244.5:1.5) {$5$};
    \node[v,label=below right:{$v_3$}] (v3) at (296:1.5) {$6$};
    \node[v,label=right:{$v_4$}] (v4) at (347.5:1.5) {$5$};
    \draw (v2) -- (v3) node[midway,above]{$b$};
    \draw (v4) -- (v3) node[midway,below right]{$c$};
    \draw (v2) -- (v1) node[midway,below left]{$a$};
    \draw (v2) -- (u) node[midway,above left]{$e$};
    \draw (u) -- (v3) node[midway,above right]{$f$};
    \draw (u) -- (v4) node[midway,above]{$g$};
    \draw (u) -- (v1) node[midway,above]{$d$};
  \end{tikzpicture}
\caption{Notation for Lemma~\ref{lem:4b}}
    \label{fig:C4b}

  \end{figure}
  We conclude using Theorem~\ref{thm:nss}.
\end{proof}

\begin{lemma}
  \label{lem:4c}
  The graph $G$ does not contain a $7$-vertex $u$ with six neighbors
  $v_1,\ldots,v_6$ such that $d(v_1)=d(v_2)=d(v_5)=d(v_6)=5$,
  $d(v_3)=d(v_4)=6$ and $v_iv_{i+1}\in E(G)$ for
  $1\leqslant i\leqslant 6$.
\end{lemma}

\begin{proof}
  We use the notation depicted in Figure~\ref{fig:C4c}. By minimality,
  we color $G'=G\setminus\{a,\ldots,k\}$ and uncolor
  $u,v_1,v_2,v_5,v_6$. We have $|\hat{c}|=2$, $|\hat{b}|=|\hat{d}|=4$,
  $|\hat{a}|=|\hat{e}|=|\hat{h}|=|\hat{i}|=5$,
  $|\hat{u}|=|\hat{f}|=|\hat{k}|=6$ and $|\hat{g}|=|\hat{j}|=7$.

  Note that due to $C_{\ref{C3a}}$, there is no edge $v_iv_j$ for
  $i=1,2$ and $j=5,6$. Thus, we have $|\hat{v_1}|=|\hat{v_6}|=4$ and
  $|\hat{v_2}|=|\hat{v_5}|=5$.
  \begin{figure}[!h]
    \centering
    \begin{tikzpicture}[v/.style={draw=black,minimum size = 10pt,ellipse,inner sep=1pt}]
      \node[v,label=above:{$u$}] (u) at (0,0) {$7$};
      \node[v,label=left:{$v_1$}] (v1) at (141.5:1.5)  {$5$};
      \node[v,label=left:{$v_2$}] (v2) at (193:1.5)  {$5$};
      \node[v, very thick,label=below left:{$v_3$}] (v3) at (244.5:1.5) {$6$};
      \node[v, very thick,label=below right:{$v_4$}] (v4) at (296:1.5) {$6$};
      \node[v,label=right:{$v_5$}] (v5) at (347.5:1.5) {$5$};
      \node[v,label=right:{$v_6$}] (v6) at (39:1.5)  {$5$};
      \draw (v1) -- (v2) node[midway,left]{$a$};
      \draw (v3) -- (v2) node[midway,left]{$b$};
      \draw (v3) -- (v4) node[midway,above]{$c$};
      \draw (v5) -- (v4) node[midway,right]{$d$};
      \draw (v5) -- (v6) node[midway,right]{$e$};
      \draw (v1) -- (u) node[midway,above]{$f$};
      \draw (u) -- (v2) node[midway,above]{$g$};
      \draw (u) -- (v3) node[midway,left]{$h$};
      \draw (u) -- (v4) node[midway,right]{$i$};
      \draw (u) -- (v5) node[midway,above]{$j$};
      \draw (u) -- (v6) node[midway,above]{$k$};
    \end{tikzpicture} 
\caption{Notation for Lemma~\ref{lem:4c}}
    \label{fig:C4c}
  \end{figure}
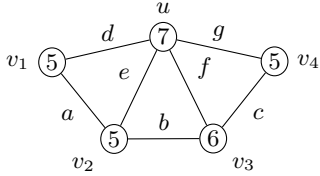
  We did not succeed in finding a suitable monomial for the
  Nullstellensatz approach, hence we only present a case analysis
  proof. We color $e$ with a color not in $\hat{v_6}$. We forget $v_6$
  and $v_5$, then color $d,i$ and $h$ with colors not in $\hat{c}$ and
  forget $c$. We color $j,k,u,f,g$, then apply Lemma~\ref{lem:diam} on
  $\mathcal{T}(G)$ with the path $v_1av_2b$.
\end{proof}

\subsection{Configuration $C_{\ref{C5}}$}
To prove that $G$ does not contain $C_{\ref{C5}}$, we prove the two
following lemmas.

\begin{lemma}
  \label{lem:C5a}
  The graph $G$ does not contain a $5$-vertex $u$ adjacent to three
  consecutive $6$-vertices $v_1,v_2,v_3$.
\end{lemma}

\begin{proof}
  We use the notation depicted in Figure~\ref{fig:C5a}. We color
  $G\setminus\{a,\ldots,e\}$ by minimality, and then uncolor
  $u,v_1,v_2,v_3$. 
  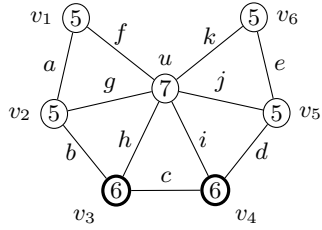
\begin{figure}[!h]
    \centering
    \begin{tikzpicture}[v/.style={draw=black,minimum size = 10pt,ellipse,inner sep=1pt}]
      \node[v,label=below:{$v_2$}] (v2) at (-90:1.5) {$6$};
      \node[v,label=above:{$u$}] (u) at (0,0)  {$5$};
      \node[v,label=left:{$v_1$}] (v1) at (-150:1.5) {$6$};
      \node[v,label=right:{$v_3$}] (v3) at (-30:1.5) {$6$};
      \draw (u) -- (v1) node [midway, above]{$a$};
      \draw (u) -- (v2) node [midway, right]{$b$};
      \draw (u) -- (v3) node [midway, above]{$c$};
      \draw (v2) -- (v1) node [midway, below]{$d$};
      \draw (v2) -- (v3) node [midway, below]{$e$};
    \end{tikzpicture}
\caption{Notation for Lemma~\ref{lem:C5a}}
    \label{fig:C5a}
  \end{figure}

  We have $|\hat{v_1}|=|\hat{v_3}|=3$ or $2$ depending on whether
  $v_1v_3\in E(G)$, $|\hat{d}|=|\hat{e}|=3$,
  $|\hat{v_2}|=|\hat{a}|=|\hat{c}|=4$, $|\hat{b}|=5$ and
  $|\hat{u}|=6$.   We conclude using Theorem~\ref{thm:nss}.
\end{proof}

\begin{lemma}
  \label{lem:C5b}
The graph $G$ does not contain a triangulated $5$-vertex $u$ with neighbors
  $v_1,\ldots,v_5$ satisfying $d(v_1)=d(v_3)=d(v_5)=6$ and
  $d(v_2)=d(v_4)=7$.
\end{lemma}

\begin{proof}
  We use the notation depicted in Figure~\ref{fig:C5b}. We color
  $G'=G\setminus\{a,\ldots,j\}$ by minimality, and then uncolor
  $u,v_1,\ldots,v_5$.   
  \begin{figure}[!h]
    \centering
    \begin{tikzpicture}[v/.style={draw=black,minimum size = 10pt,ellipse,inner sep=1pt}]
      \node[v,label=left:{$u$}] (u) at (0,0)  {$5$};
      \node[v,label=above:{$v_3$}] (v3) at (90:1.5) {$6$};
      \node[v,label=left:{$v_2$}] (v2) at (162:1.5) {$7$};
      \node[v,label=left:{$v_1$}] (v1) at (234:1.5) {$6$};
      \node[v,label=right:{$v_5$}] (v5) at (306:1.5) {$6$};
      \node[v,label=right:{$v_4$}] (v4) at (18:1.5) {$7$};
      \draw (u) -- (v3) node[midway,right] {$a$};
      \draw (u) -- (v4) node[midway,above] {$b$};
      \draw (u) -- (v5) node[midway,above right] {$c$};
      \draw (u) -- (v1) node[midway,above left] {$d$};
      \draw (u) -- (v2) node[midway,above] {$e$};
      \draw (v4) -- (v3) node[midway,above] {$f$};
      \draw (v4) -- (v5) node[midway,right] {$g$};
      \draw (v1) -- (v5) node[midway,above] {$h$};
      \draw (v1) -- (v2) node[midway,left] {$i$};
      \draw (v2) -- (v3) node[midway,above] {$j$};
    \end{tikzpicture}
\caption{Notation for Lemma~\ref{lem:C5b}}
    \label{fig:C5b}
  \end{figure}
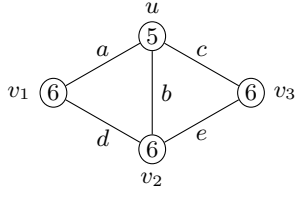
  We have $|\hat{v_2}|=|\hat{v_4}|=2$,
  $|\hat{f}|=|\hat{g}|=|\hat{i}|=|\hat{j}|=3$,
  $|\hat{v_1}|=|\hat{v_3}|=|\hat{v_5}|=|\hat{h}|=4$,
  $|\hat{b}|=|\hat{e}|=6$, $|\hat{a}|=|\hat{c}|=|\hat{d}|=7$ and
  $|\hat{u}|=10$. Moreover, for $1\leqslant i\leqslant 5$,
  $|\hat{v_i}|$ may differ depending on the presence of edges between
  the $v_i$'s that are not on the figure, but we may assume that
  $|\hat{v_2}|,|\hat{v_4}|$ are at least $2$ and
  $|\hat{v_1}|,|\hat{v_3}|,|\hat{v_5}|$ are at least $4$.

  In each case, we conclude using Theorem~\ref{thm:nss}.
\end{proof}

\subsection{Configuration $C_{\ref{C6}}$}
 \begin{lemma}
  \label{lem:C6}
The graph $G$ does not contain a $7$-vertex $u$ with four neighbors
  $v_1,\ldots,v_4$ satisfying $d(v_1)=d(v_2)=5$, $d(v_3)=6$,
  $d(v_4)=4$, and $v_1v_2,v_2v_3\in E(G)$.
\end{lemma}

\begin{proof}
  We consider the notation depicted in Figure~\ref{fig:C6}. By
  minimality, we color $G'=G\setminus\{a,\ldots,f\}$ and uncolor
  $u,v_1,v_2,v_3,v_4$.
  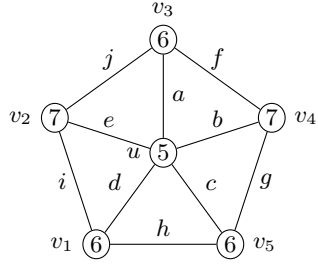
\begin{figure}[!h]
    \centering
    \begin{tikzpicture}[v/.style={draw=black,minimum size = 10pt,ellipse,inner sep=1pt}]
      \node[v,label=right:{$v_2$}] (v2) at (0:1.5) {$5$};
      \node[v,label=above:{$u$}] (u) at (0,0)  {$7$};
      \node[v,label=left:{$v_1$}] (v1) at (60:1.5) {$5$};
      \node[v,label=left:{$v_3$}] (v3) at (-60:1.5) {$6$};
      \node[v,label=left:{$v_4$}] (v4) at (180:1.5) {$4$};
      \draw (u) -- (v1) node[midway,above left]{$a$};
      \draw (u) -- (v2) node[midway,above]{$b$};
      \draw (u) -- (v3) node[midway,left]{$c$};
      \draw (u) -- (v4) node[midway,above]{$d$};
      \draw (v2) -- (v1) node[midway,above right]{$e$};
      \draw (v2) -- (v3) node[midway,right]{$f$};
    \end{tikzpicture}
\caption{Notation for Lemma~\ref{lem:C6}}
    \label{fig:C6}
  \end{figure}

  We have $|\hat{c}|=3$,
  $|\hat{u}|=|\hat{v_4}|=|\hat{a}|=|\hat{d}|=|\hat{f}|=4$,
  $|\hat{b}|=|\hat{e}|=5$ and $|\hat{v_2}|=6$. Moreover, none of
  $v_1,v_2,v_3$ are adjacent to $v_4$ because of $C_{\ref{C1}}$, and
  $v_1v_3\notin E(G)$ because of $C_{\ref{C3b}}$. Hence, we have
  $|\hat{v_1}|=4$, $|\hat{v_3}|=2$. We forget $v_4$, and conclude using Theorem~\ref{thm:nss}.
\end{proof}

\subsection{Configuration $C_{\ref{C7}}$}
According to the definition of a $S_3$-neighbor, if $G$ contains
$C_{\ref{C7}}$, $v_1$ is triangulated and we are in one of the
following cases:
\begin{itemize}
\item[$\bullet$]$C_{\ref{C7}a}$: $u$ and $v_1$ have a common neighbor $w$ of
  degree six.
\item[$\bullet$]$C_{\ref{C7}b}$: $v_1$ has two neighbors $w_1,w_2$ of degree
  six such that $uv_1w_1$ and $uv_1w_2$ are not triangular
  faces. Moreover, due to $C_{\ref{C7}a}$, we know that
  $w_1w_2, v_2w_2\in E(G)$ and $uw_1,uw_2\notin E(G)$.
\item[$\bullet$]$C_{\ref{C7}c}$: $v_1$ has a neighbor $w$ of degree five such
  that $uv_1w$ is not a triangular face.
\end{itemize}
We dedicate a lemma to each of these configurations.

\begin{lemma}
  \label{lem:C7a}
  The graph $G$ does not contain $C_{\ref{C7}a}$.
\end{lemma}

\begin{proof}
  We use the notation depicted in Figure~\ref{fig:C7a}. By minimality, we
  color $G\setminus\{a,\ldots,j\}$ and uncolor
  $u,v_1,\ldots,v_4,w$. 
  \begin{figure}[!h]
    \centering
    \begin{tikzpicture}[v/.style={draw=black,minimum size = 10pt,ellipse,inner sep=1pt}]
      \node[v,label=below:{$u$}] (u) at (0,0)  {$7$};
      \node[v,label=right:{$w$}] (w) at (-64.5:1.5) {6};
      \node[v,label=left:{$v_4$}] (v4) at (193:1.5) {5};
      \node[very thick, v] (v5) at (141.5:1.5) {$8$};
      \node[v,label=above:{$v_3$}] (v1) at (90:1.5) {4};
      \node[v,label=right:{$v_2$}] (v2) at (38.5:1.5) {7};
      \node[v,label=right:{$v_1$}] (v3) at (-13:1.5) {5};
      \draw (u) -- (v1) node[midway,left] {$c$};
      \draw (u) -- (v2) node[midway,above] {$d$};
      \draw (u) -- (v3) node[midway,above] {$e$};
      \draw (u) -- (w) node[midway,right] {$f$};
      \draw (u) -- (v4) node[midway,below] {$a$};
      \draw (u) -- (v5) node[midway,below] {$b$};
      \draw (v2) -- (v1) node[midway,above] {$h$};
      \draw (v2) -- (v3) node[midway,right] {$i$};
      \draw (w) -- (v3) node[midway,right] {$j$};
      \draw (v5) -- (v1) node[midway,above] {$g$};
    \end{tikzpicture}
\caption{Notation for Lemma~\ref{lem:C7a}}
    \label{fig:C7a}
  \end{figure}
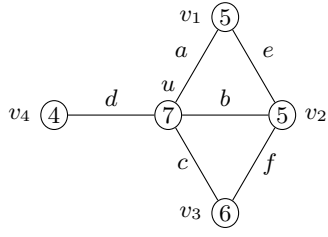

  We have $|\hat{b}|=|\hat{g}|=2$, $|\hat{i}|=|\hat{j}|=4$,
  $|\hat{d}|=|\hat{f}|=|\hat{a}|=|\hat{h}|=5$,
  $|\hat{u}|=|\hat{v_3}|=|\hat{e}|=7$ and $|\hat{c}|=8$. Moreover, the
  sizes $|\hat{v_1}|,|\hat{v_2}|,|\hat{w}|,|\hat{v_4}|$ depend on the
  presence of edges between the vertices $v_1,v_2,v_4,w$. Note that
  $v_3$ is not adjacent to any of $v_1,v_4,w$ by
  $C_{\ref{C1}}$. Moreover, $|N(v_4)\cap\{v_1,w\}|\leqslant 1$ by
  $C_{\ref{C3b}}$. We may thus assume that
  $|\hat{v_1}|=6+|N(v_1)\cap\{v_4\}|$,
  $|\hat{v_2}|=2+|N(v_2)\cap\{v_4,w\}|$,
  $|\hat{v_4}|=2+|N(v_4)\cap\{v_1,v_2,w\}|$ and
  $|\hat{w}|=2+|N(w)\cap\{v_2,v_4\}|$. We forget $v_3$ and conclude
  using Theorem~\ref{thm:nss}.
\end{proof}

\begin{lemma}
  \label{lem:C7b}
  The graph $G$ does not contain $C_{\ref{C7}b}$.
\end{lemma}

\begin{proof}
  We consider the notation of Figure~\ref{fig:C7b}. By minimality, we
  color $G\setminus\{a,\ldots,h\}$ and uncolor $v_1,v_2,v_3,w_1,w_2$.
  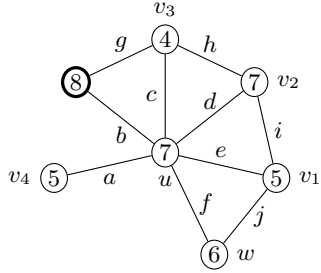
\begin{figure}[!h]
    \centering
    \begin{tikzpicture}[v/.style={draw=black,minimum size = 10pt,ellipse,inner sep=1pt}]
      \node[v,very thick,label=below:{$u$}] (u) at (0,0)  {$7$};
      \node[v,very thick] (v5) at (141.5:1.5) {$8$};
      \node[v,label=above:{$v_3$}] (v3) at (90:1.5) {4};
      \node[v,label=above:{$v_2$}] (v2) at (38.5:1.5) {7};
      \node[v,label=below:{$v_1$}] (v1) at (-13:1.5) {5};
      \node[v,xshift=1.5cm,label=right:{$w_1$}] (w2) at (-13:1.5) {6};
      \node[v,xshift=1.5cm,label=right:{$w_2$}] (w1) at (38.5:1.5) {6};
      \draw (u) -- (v1) node[midway,above] {$a$};
      \draw (v2) -- (v1) node[midway,right] {$b$};
      \draw (v1) -- (w1) node[midway,right] {$c$};
      \draw (v1) -- (w2) node[midway,above] {$d$};
      \draw (u) -- (v3) node[midway,left] {$e$};
      \draw (v2) -- (v3) node[midway,above] {$f$};
      \draw (v2) -- (w1) node[midway,above] {$g$};
      \draw (w1) -- (w2) node[midway,right] {$h$};
      \draw[very thick] (u) -- (v2);
      \draw[very thick] (u) -- (v5);
      \draw[very thick] (v5) -- (v3);
    \end{tikzpicture}
\caption{Notation for Lemma~\ref{lem:C7b}}
    \label{fig:C7b}
  \end{figure}
  We have $|\hat{e}|=2$, $|\hat{a}|=|\hat{g}|=|\hat{h}|=3$,
  $|\hat{w_2}|=|\hat{f}|=4$, $|\hat{v_3}|=|\hat{b}|=|\hat{d}|=5$,
  $|\hat{c}|=6$, $|\hat{v_1}|=7$. Moreover, $|\hat{v_2}|,|\hat{w_1}|$
  are $2$ or $3$ depending on the presence of the edge $v_2w_1$ in
  $G$.  We forget $v_3$.

  We color $a$ with a color not in $\hat{e}$, then forget $e,f$. The
  resulting configuration is now the same as in Lemma~\ref{lem:C5a}.
\end{proof}

\begin{lemma}
  \label{lem:C7c} 
  $G$ does not contain $C_{\ref{C7}c}$.
\end{lemma}

\begin{proof}
  We use the notation depicted in Figure~\ref{fig:C7c}. 
  \begin{figure}[!h]
    \centering
     \begin{tikzpicture}[v/.style={draw=black,minimum size = 10pt,ellipse,inner sep=1pt}]
      \node[v,label=below:{$u$}] (u) at (0,0)  {$7$};
      \node[very thick,v] (v5) at (141.5:1.5) {$8$};
      \node[v,label=above:{$v_3$}] (v3) at (90:1.5) {4};
      \node[v,label=above:{$v_2$}] (v2) at (38.5:1.5) {7};
      \node[v,label=below:{$v_1$}] (v1) at (-13:1.5) {5};
      \node[v,label=left:{$v_4$}] (v4) at (193:1.5) {5};
      \node[v,xshift=1.5cm,label=above:{$w$}] (w) at (-13:1.5) {5};
      \draw (u) -- (v3) node[midway,left] {$b$};
      \draw (u) -- (v2) node[midway,above] {$c$};
      \draw (u) -- (v1) node[midway,above] {$d$};
      \draw (u) -- (v4) node[midway,above] {$a$};
      \draw (v3) -- (v2) node[midway,above] {$e$};
      \draw (v2) -- (v1) node[midway,right] {$f$};
      \draw (w) -- (v1) node[midway,above] {$g$};
      \draw[very thick] (u) -- (v5);
      \draw[very thick] (v3) -- (v5);
    \end{tikzpicture}
\caption{Notation for Lemma~\ref{lem:C7c}}
    \label{fig:C7c}
  \end{figure}
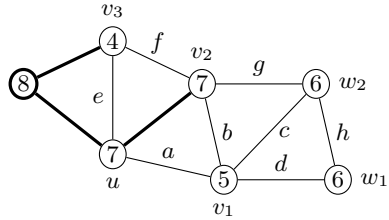
  Note that we may assume that, if $w\neq v_4$, $uw\notin E(G)$ due to
  Lemma~\ref{lem:C7a}. By minimality, we color
  $G'=G\setminus\{a,\ldots,g\}$ and uncolor $u,v_1,v_2,v_3,v_4,w$. We
  have $|\hat{c}|=3$, $|\hat{u}|=|\hat{e}|=|\hat{f}|=4$,
  $|\hat{b}|=|\hat{d}|=5$ and $|\hat{v_1}|=|\hat{v_3}|=6$.

  Moreover, if $w= v_4$, we have $|\hat{a}|=4$ and $|\hat{g}|=5$,
  otherwise, $|\hat{a}|=3$ and $|\hat{g}|=4$. We may also assume that
  $|\hat{v_2}|,|\hat{v_4}|,|\hat{w}|$ are $2,3,4$ or $5$ depending on
  whether $v_2v_4,v_2w\in E(G)$. We forget $v_3$ and conclude using
  Theorem~\ref{thm:nss}.
\end{proof}

\subsection{Configuration $C_{\ref{C8}}$}
Note that $G$ does not contain $C_{\ref{C7}}$, therefore, if $G$
contains $C_{\ref{C8}}$, we are in one of the following cases:
\begin{itemize}
\item[$\bullet$] $C_{\ref{C8}a}$: the common neighbor of $v_1,u$ and $v_2$ has
  degree $7$ and $v_1$ is an $S_5$-neighbor of $u$.
\item[$\bullet$] $C_{\ref{C8}b}$: the common neighbor of $v_1,u$ and $v_2$ has
  degree $8$ and $v_1$ is a $(6,8)$-neighbor of $u$.
\item[$\bullet$] $C_{\ref{C8}c}$: the common neighbor of $v_1,u$ and $v_2$ has
  degree $8$ and $v_1$ has two neighbors $w_1,w_2$ of degree $6$ such
  that $uv_1w_1$ and $uv_1w_2$ are not triangular faces
\item[$\bullet$] $C_{\ref{C8}d}$: the common neighbor of $v_1,u$ and $v_2$ has
  degree $8$ and $v_1$ has a neighbor $w$ of degree $5$ such that
  $uv_1w$ is not a triangular face.
\end{itemize}
We dedicate a lemma to each of these configurations.

\begin{lemma}
  \label{lem:C8a}
The graph $G$ does not contain $C_{\ref{C8}a}$.
\end{lemma}

\begin{proof}
  We use the notation depicted in Figure~\ref{fig:C8a}. Since $G$ is planar we cannot have both edges $w_1v_2$ and $w_2v_3$. Up to flipping the configuration, we assume that $w_1v_2$ is not an edge. Now by minimality, we
  color $G\setminus\{a,\ldots,m\}$ and uncolor
  $u,v_1,v_2,v_3,w_1,w_2$.
  \begin{figure}[!h]
    \centering
    \begin{tikzpicture}[v/.style={draw=black,minimum size = 10pt,ellipse,inner sep=1pt}]
      \node[v,label=left:{$u$}] (u) at (0,0)  {$7$};
      \node[v,label=right:{$v_1$}] (v1) at (0:1.5) {$5$};
      \node[v,label=above:{$v_2$}] (v2) at (103:1.5) {$4$};
      \node[v,label=left:{$v_3$}] (v3) at (257:1.5) {$5$};
      \node[v,very thick] (v4) at (154.5:1.5) {8};
      \node[v,label=above:{$w_2$}] (w1) at (51.5:1.5) {7};
      \node[v, very thick] (v6) at (206:1.5) {8};
      \node[v,label=below:{$w_1$}] (w2) at (308.5:1.5) {7};
      \node[v,xshift=1.5cm, very thick] (x1) at (36:1.5) {7};
      \node[v,xshift=1.5cm, very thick] (x2) at (-36:1.5) {7};
      \draw (u) -- (v6) node[midway,below] {$g$};
      \draw (u) -- (v4) node[midway,above] {$a$};
      \draw (u) -- (v2) node[midway,right] {$b$};
      \draw (u) -- (w1) node[midway,right] {$c$};
      \draw (u) -- (v1) node[midway,above] {$d$};
      \draw (u) -- (w2) node[midway,right] {$e$};
      \draw (u) -- (v3) node[midway,right] {$f$};
      \draw (v2) -- (v4) node[midway,above] {$h$};
      \draw (v2) -- (w1) node[midway,above] {$i$};
      \draw (v1) -- (w1) node[midway,right] {$j$};
      \draw (v1) -- (w2) node[midway,right] {$k$};
      \draw (v3) -- (w2) node[midway,above] {$\ell$};
      \draw (v1) -- (x1) node[midway,left] {$m$};
      \draw (v1) -- (x2) node[midway,left] {$n$};
      \draw[very thick] (x2) -- (x1);
      \draw[very thick] (w1) -- (x1);
      \draw[very thick] (x2) -- (w2);
      \draw[very thick] (v6) -- (v3);
    \end{tikzpicture}
\caption{Notation for Lemma~\ref{lem:C8a}}
    \label{fig:C8a}
  \end{figure}
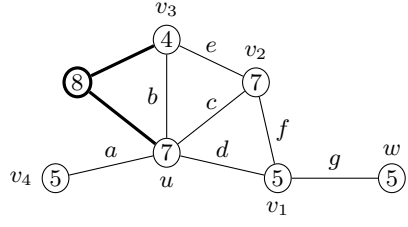
  We have $|\hat{g}|=|\hat{h}|=2$,
  $|\hat{a}|=|\hat{\ell}|=|\hat{m}|=|\hat{n}|=3$, $|\hat{i}|=5$,
  $|\hat{c}|=|\hat{e}|=|\hat{j}|=|\hat{k}|=6$,
  $|\hat{f}|=|\hat{v_2}|=7$, $|\hat{u}|=|\hat{v_1}|=8$, $|\hat{b}|=9$
  and $|\hat{d}|=10$.

  Moreover, note that the only edges of $G$ between uncolored vertices
  that may not be present on the figure are $w_1w_2$ and
  $w_2v_3$. Depending on the presence of these edges, $|\hat{w_2}|$ is
  $2,3$ or $4$, $|\hat{w_1}|$ is $2$ or $3$ and $|\hat{v_3}|$ is $4$
  or $5$.

  We first forget $v_2$, and color $j$ with a color not in
  $\hat{i}$. Then we color $m,n,w_2$, and forget $i,b,h$.

  We then conclude using Theorem~\ref{thm:nss} to color the remaining
  graph.
\end{proof}

\begin{lemma}
  \label{lem:C8b}
The graph $G$ does not contain $C_{\ref{C8}b}$.
\end{lemma}

\begin{proof}
  We use the notation depicted in Figure~\ref{fig:C8b}. By minimality,
  we color $G'=G\setminus\{a,\ldots,\ell\}$ and uncolor
  $u,v_1,v_2,v_3,w$.
  \begin{figure}[!h]
    \centering
    \begin{tikzpicture}[v/.style={draw=black,minimum size = 10pt,ellipse,inner sep=1pt}]
      \node[v,label=left:{$u$}] (u) at (0,0)  {$7$};
      \node[v,label=right:{$v_1$}] (v1) at (0:1.5) {$5$};
      \node[v,label=above:{$v_2$}] (v2) at (103:1.5) {$4$};
      \node[v,label=left:{$v_3$}] (v3) at (257:1.5) {$5$};
      \node[v, very thick] (v4) at (154.5:1.5) {7};
      \node[v,very thick] (v5) at (51.5:1.5) {8};
      \node[v, very thick] (v6) at (206:1.5) {$8$};
      \node[v,label=right:{$w$}] (w) at (308.5:1.5) {6};
      \draw (u) -- (v4) node[midway,above]{$a$};
      \draw (u) -- (v2) node[midway,right]{$b$};
      \draw (u) -- (v5) node[midway,right]{$c$};
      \draw (u) -- (v1) node[midway,above]{$d$};
      \draw (u) -- (w) node[midway,right]{$e$};
      \draw (u) -- (v3) node[midway,right]{$f$};
      \draw (u) -- (v6) node[midway,below]{$g$};
      \draw (v2) -- (v4) node[midway,above]{$h$};
      \draw (v2) -- (v5) node[midway,above]{$i$};
      \draw (v1) -- (v5) node[midway,right]{$j$};
      \draw (v1) -- (w) node[midway,right]{$k$};
      \draw (w) -- (v3) node[midway,above]{$\ell$};
      \draw (v6)--(v3);
    \end{tikzpicture}
\caption{Notation for Lemma~\ref{lem:C8b}}
    \label{fig:C8b}
  \end{figure}
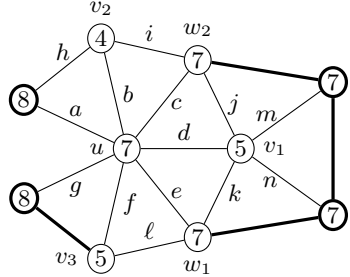
  We have $|\hat{g}|=|\hat{j}|=2$, $|\hat{h}|=|\hat{i}|=3$,
  $|\hat{w}|=|\hat{a}|=|\hat{c}|=|\hat{\ell}|=4$, $|\hat{k}|=5$,
  $|\hat{v_2}|=6$, $|\hat{u}|=|\hat{e}|=|\hat{f}|=7$, $|\hat{d}|=8$
  and $|\hat{b}|=9$.

  Note that, due to $C_{\ref{C3b}}$, $v_1v_3\notin E(G)$. Moreover,
  due to $C_{\ref{C1}}$, $v_2$ is not adjacent to $v_1,v_3,w$. Since
  the graph $G$ is simple, all the edges of $G$ between uncolored
  vertices are drawn in the figure. We may thus assume that
  $|\hat{v_1}|=5$ and $|\hat{v_3}|=4$.
  
  We forget $v_2$ and conclude using Theorem~\ref{thm:nss}.
\end{proof}

\begin{lemma}
  \label{lem:C8c}
The graph $G$ does not contain $C_{\ref{C8}c}$.
\end{lemma}

\begin{proof}
  We use the notation depicted in Figure~\ref{fig:C8c}. By minimality, we
  color $G\setminus\{a,\ldots,q\}$ and uncolor
  $u,v_1,v_2,v_3,w_1,w_2$.
  \begin{figure}[!h]
    \centering
    \begin{tikzpicture}[v/.style={draw=black,minimum size = 10pt,ellipse,inner sep=1pt},scale=.8]
      \node[v,label=left:{$u$}] (u) at (0,0)  {$7$};
      \node[v,label=right:{$v_1$}] (v1) at (0:1.5) {$5$};
      \node[v,label=above:{$v_2$}] (v2) at (103:1.5) {$4$};
      \node[v,label=left:{$v_3$}] (v3) at (257:1.5) {$5$};
      \node[v, very thick] (v4) at (154.5:1.5) {7};
      \node[v, very thick] (w1) at (51.5:1.5) {8};
      \node[v, very thick] (v6) at (206:1.5) {$8$};
      \node[v, very thick] (w2) at (308.5:1.5) {8};
      \node[v,xshift=1.5cm,label=above:{$w_2$}] (x1) at (36:1.5) {6};
      \node[v,xshift=1.5cm,label=below:{$w_1$}] (x2) at (-36:1.5) {6};
      \draw (u) -- (v4) node[midway,above] {$a$};
      \draw (u) -- (v2) node[midway,right] {$b$};
      \draw (u) -- (w1) node[midway,right] {$c$};
      \draw (u) -- (v1) node[midway,above] {$d$};
      \draw (u) -- (w2) node[midway,right] {$e$};
      \draw (u) -- (v3) node[midway,right] {$f$};
      \draw (u) -- (v6) node[midway,below] {$g$};
      \draw (v1) -- (w1) node[midway,right] {$j$};
      \draw (v1) -- (x1) node[midway,above] {$n$};
      \draw (v1) -- (x2) node[midway,below] {$o$};
      \draw (v1) -- (w2) node[midway,right] {$k$};
      \draw (v2) -- (v4) node[midway,above] {$h$};
      \draw (v2) -- (w1) node[midway,above] {$i$};
      \draw (w1) -- (x1) node[midway,above] {$m$};
      \draw (x2) -- (x1) node[midway,right] {$q$};
      \draw (x2) -- (w2) node[midway,below] {$p$};
      \draw (v3) -- (w2) node[midway,below] {$\ell $};
      \draw[very thick] (v3) -- (v6);
    \end{tikzpicture}
        \begin{tikzpicture}[v/.style={draw=black,minimum size = 10pt,ellipse,inner sep=1pt}, scale=.8]
      \node[v,label=left:{$u$}] (u) at (0,0)  {$7$};
      \node[v,label=right:{$v_1$}] (v1) at (0:1.5) {$5$};
      \node[v,label=above:{$v_2$}] (v2) at (103:1.5) {$4$};
      \node[v,label=left:{$v_3$}] (v3) at (257:1.5) {$5$};
      \node[v, very thick] (v4) at (154.5:1.5) {7};
      \node[v, very thick] (w1) at (51.5:1.5) {8};
      \node[v, very thick] (w2) at (308.5:1.5) {8};
      \node[v,xshift=1.5cm,label=above:{$w_2$}] (x1) at (36:1.5) {6};
      \node[v,xshift=1.5cm,label=below:{$w_1$}] (x2) at (-36:1.5) {6};
      \draw (u) -- (v4) node[midway,above] {$a$};
      \draw (u) -- (v2) node[midway,right] {$b$};
      \draw (u) -- (w1) node[midway,right] {$c$};
      \draw (u) -- (v1) node[midway,above] {$d$};
      \draw (u) -- (w2) node[midway,right] {$e$};
      \draw (u) -- (v3) node[midway,right] {$f$};
      \draw plot [smooth, tension=0.6] coordinates {(u.south west) (230:2) (308.5:2.5) (x2.south west)};
      \node () at (206:1) {$g$};
      \draw (v1) -- (w1) node[midway,right] {$j$};
      \draw (v1) -- (x1) node[midway,above] {$n$};
      \draw (v1) -- (x2) node[midway,below] {$o$};
      \draw (v1) -- (w2) node[midway,right] {$k$};
      \draw (v2) -- (v4) node[midway,above] {$h$};
      \draw (v2) -- (w1) node[midway,above] {$i$};
      \draw (w1) -- (x1) node[midway,above] {$m$};
      \draw (x2) -- (x1) node[midway,right] {$q$};
      \draw (x2) -- (w2) node[midway,below] {$p$};
      \draw (v3) -- (w2) node[midway,above] {$\ell $};
      \draw[very thick,bend right] (v3) to (x2);
    \end{tikzpicture}
    \begin{tikzpicture}[v/.style={draw=black,minimum size = 10pt,ellipse,inner sep=1pt}, scale=.8]
      \node[v,label=left:{$u$}] (u) at (0,0)  {$7$};
      \node[v,label=right:{$v_1$}] (v1) at (0:1.5) {$5$};
      \node[v,label=above:{$v_2$}] (v2) at (103:1.5) {$4$};
      \node[v,label=left:{$v_3$}] (v3) at (257:1.5) {$5$};
      \node[v, very thick] (v4) at (154.5:1.5) {7};
      \node[v, very thick] (w1) at (51.5:1.5) {8};
      \node[v, very thick] (w2) at (308.5:1.5) {8};
      \node[v,xshift=1.5cm,label=above:{$w_2$}] (x1) at (36:1.5) {6};
      \node[v,xshift=1.5cm,label=below:{$w_1$}] (x2) at (-36:1.5) {6};
      \draw (u) -- (v4) node[midway,above] {$a$};
      \draw (u) -- (v2) node[midway,right] {$b$};
      \draw (u) -- (w1) node[midway,right] {$c$};
      \draw (u) -- (v1) node[midway,above] {$d$};
      \draw (u) -- (w2) node[midway,right] {$e$};
      \draw (u) -- (v3) node[midway,right] {$f$};
      \draw plot [smooth, tension=0.6] coordinates {(u.south west) (170:2) (130:2) (51.5:2.1) (x1.north west)};
      \node () at (190:1) {$g$};
      \draw (v1) -- (w1) node[midway,right] {$j$};
      \draw (v1) -- (x1) node[midway,above] {$n$};
      \draw (v1) -- (x2) node[midway,below] {$o$};
      \draw (v1) -- (w2) node[midway,right] {$k$};
      \draw (v2) -- (v4) node[midway,above] {$h$};
      \draw (v2) -- (w1) node[midway,above] {$i$};
      \draw (w1) -- (x1) node[midway,above] {$m$};
      \draw (x2) -- (x1) node[midway,right] {$q$};
      \draw (x2) -- (w2) node[midway,below] {$p$};
      \draw (v3) -- (w2) node[midway,above] {$\ell $};
       \draw[very thick] plot [smooth, tension=0.6] coordinates {(v3.south east) (3.5,-1.5) (x1.south east)};
      
    \end{tikzpicture}
\caption{Notation for Lemma~\ref{lem:C8c}}
    \label{fig:C8c}
  \end{figure}
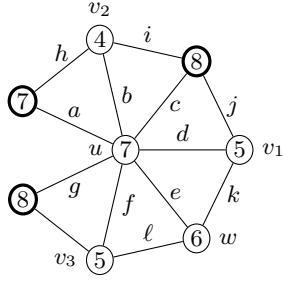
  We have $|\hat{m}|=|\hat{p}|=|\hat{\ell }|=2$, $|\hat{h}|=3$,
  $|\hat{a}|=|\hat{i}|=|\hat{q}|=4$,
  $|\hat{c}|=|\hat{e}|=|\hat{j}|=|\hat{k}|=5$, $|\hat{v_2}|=6$,
  $|\hat{f}|=|\hat{n}|=|\hat{o}|=7$, $|\hat{v_1}|=8$, $|\hat{b}|=9$
  and $|\hat{d}|=10$.
  
  Moreover, $|\hat{g}|$, $|\hat{u}|$, $|\hat{v_3}|$, $|\hat{w_1}|$, 
  $|\hat{w_2}|$ may differ depending on the presence of edges between
  these vertices that are not on the figure, and whether $g$ is
  incident to $w_1$ or $w_2$. However, we still have at least $2$
  colors in $\hat{g}$, $3$ in $\hat{w_1},\hat{v_3},\hat{w_2}$ and $6$
  in $\hat{u}$. Similarly, if $g$ is incident to $w_1$ or $w_2$, the edges incident to $w_1$ or $w_2$ get additional colors. 

  We forget $v_2$, then color $a$ with a color not in $\hat{h}$ and
  $g$ arbitrarily. Then we forget $h,b,i$. Note that afterwards, $m,n,o,p,q$ have the number of available colors given at the beginning of the proof, regardless of whether $g$ is incident to $w_1,w_2$ or not. We color $\ell$ such that
  $\hat{u}$ and $\hat{f}$ are not the same set of size 4 afterwards,
  then $p,e$. We color $m$ such that $\hat{w_1},\hat{w_2}$ are not the
  same set of size 2, then $c,k,j$. We then separate three cases:
  \begin{enumerate}
  \item Assume that $g$ is not incident to $w_1,w_2$ and that
    $\hat{f}=\hat{v_3}$. We color $u$ with a color not in $\hat{f}$
    and forget $v_3,f$. We have three cases:
    \begin{enumerate}
    \item If $\hat{w_2}=\hat{q}$ (or $\hat{w_1}=\hat{q}$ by symmetry),
      we color $w_1$ with a color not in $\hat{q}$, then $o$ with a
      color not in $\hat{q}$, and we apply Lemma~\ref{lem:diam} on
      $\mathcal{T}(G)$ with the path $dv_1nw_2q$.
    \item If $\hat{w_2}\neq\hat{q}$ and moreover,
      $\hat{q}\not\subset\hat{w_1}\cup\hat{w_2}$, we color $q$ with a
      color not in this union. We color $w_2$ with a color not in
      $\hat{w_1}$, color $n$, and apply Lemma~\ref{lem:diam} on $\mathcal{T}(G)$
      with the path $dv_1ow_1$.
    \item Otherwise, we have $\hat{w_2}=\{\alpha,\beta\}$,
      $\hat{w_1}=\{\gamma,\delta\}$ and $\hat{q}=\{\alpha,\gamma\}$
      (with possibly $\beta=\delta$). Therefore, there are two
      possible colorings for $\{w_1,w_2,q\}$ hence at least one of
      them ensures that $\hat{v_1}\neq\hat{d}$. We then apply
      Theorem~\ref{thm:clique} on $\{v_1,d,n,o\}$.
    \end{enumerate}
  \item Assume that $g$ is not incident to $w_1,w_2$ and that
    $\hat{f}\neq\hat{v_3}$.

    Since $|\hat{w_1}|\neq|\hat{w_2}|$, $\{w_1,w_2,q\}$ is
    colorable. Moreover, there are at least two different colorings
    for this set. Therefore, we may always color $w_1,w_2,q$ such that
    afterwards we have $\hat{n}\neq\hat{o}$ if they are lists of size
    two.

    If $|\hat{n}\cup\hat{o}|=3$, we can color $v_1$ with a color not
    in $\hat{n}\cup\hat{o}$, then color $u$. Since
    $\hat{f}\neq\hat{v_3}$, we can color $f,v_3$, then $d$. Finally,
    we can color $n$ and $o$ since $\hat{n}\neq\hat{o}$.
    
    Otherwise, we have $|\hat{n}\cup\hat{o}|>3$. We may thus color
    $v_3,f,u$ (since $\hat{v_3}\neq\hat{f}$) and apply
    Theorem~\ref{thm:clique} on $\{v_1,d,n,o\}$.
  \item Assume that $g$ is incident to $w_1$ or $w_2$. Free to
    exchange $w_1$ and $w_2$, we may assume that $g=uw_1$. The
    situation is depicted on Figure~\ref{fig:C8caux}. We may thus
    assume that $|\hat{f}|=|\hat{q}|=|\hat{w_2}|=2$,
    $|\hat{u}|=|\hat{v_3}|=3$,
    $|\hat{d}|=|\hat{n}|=|\hat{o}|=|\hat{w_1}|=4$ and $|\hat{v_1}|=6$.

  \begin{figure}[!h]
    \centering
    \begin{tikzpicture}[v/.style={draw=black,minimum size = 10pt,ellipse,inner sep=1pt}]
      \node[v,label=left:{$u$}] (u) at (0,0)  {$7$};
      \node[v,label=right:{$v_1$}] (v1) at (0:1.5) {$5$};
      \node[v,label=left:{$v_3$}] (v3) at (257:1.5) {$5$};
      \node[v,xshift=1.5cm,label=above:{$w_2$}] (x1) at (36:1.5) {6};
      \node[v,xshift=1.5cm,label=below:{$w_1$}] (x2) at (-36:1.5) {6};
      \draw (u) -- (v1) node[midway,above] {$d$};
      \draw (u) -- (v3) node[midway,right] {$f$};
      \draw (v1) -- (x1) node[midway,above] {$n$};
      \draw (v1) -- (x2) node[midway,right=1pt] {$o$};
      \draw (x2) -- (x1) node[midway,right] {$q$};
      \draw[very thick] (v3) -- (x2) -- (u);
    \end{tikzpicture}
    \caption{Remaining elements for Lemma~\ref{lem:C8c}}
    \label{fig:C8caux}
  \end{figure}

  If $\hat{n}=\hat{o}$, we color $w_2$ arbitrarily, otherwise, there
  exists $\alpha\in\hat{n}\setminus\hat{o}$ and we color $w_2$ not with
  $\alpha$.  We then color $q$ and $w_1$ such that
  $\hat{f}\neq\hat{v_3}$ afterwards.

  Due to the choice for the color of $w_2$, we now have
  $\hat{n}\neq\hat{o}$ if they have size two. We have
  $\hat{n}\neq\hat{o}$ and $\hat{f}\neq\hat{v_3}$, hence may now apply
  the same procedure as in the previous item.
  \end{enumerate}

\end{proof}

\begin{lemma}
  \label{lem:C8d}
The graph $G$ does not contain $C_{\ref{C8}d}$.
\end{lemma}

\begin{proof}
  We use the notation depicted in Figure~\ref{fig:C8d}. By definition,
  there is an edge $m$ between $w$ and either $w_1$ or $w_2$. 
  \begin{figure}[!h]
    \centering
    \begin{tikzpicture}[v/.style={draw=black,minimum size = 10pt,ellipse,inner sep=1pt}]
      \node[v,label=left:{$u$}] (u) at (0,0)  {$7$};
      \node[v,label=above right:{$v_1$}] (v1) at (0:1.5) {$5$};
      \node[v,label=above:{$v_2$}] (v2) at (103:1.5) {$4$};
      \node[v,label=left:{$v_3$}] (v3) at (257:1.5) {$5$};
      \node[v, very thick] (v4) at (154.5:1.5) {7};
      \node[v,label=right:{$w_2$},very thick] (w1) at (51.5:1.5) {8};
      \node[v, very thick] (v6) at (206:1.5) {$8$};
      \node[v,label=right:{$w_1$},very thick] (w2) at (308.5:1.25) {8};
      \draw (u) -- (v4) node[midway,above] {$a$};
      \draw (u) -- (v2) node[midway,right] {$b$};
      \draw (u) -- (w1) node[midway,right] {$c$};
      \draw (u) -- (v1) node[midway,above] {$d$};
      \draw (u) -- (w2) node[midway,right] {$e$};
      \draw (u) -- (v3) node[midway,right] {$f$};
      \draw (u) -- (v6) node[midway,below] {$g$};
      \draw (v2) -- (v4) node[midway,above] {$h$};
      \draw (v2) -- (w1) node[midway,above] {$i$};
      \draw  (v1) -- (w1) node[midway,right] {$j$};
      \draw (v1) -- (w2) node[midway,right] {$k$};
      \draw (v1) [out=-30,in=-90,bend left=90] to (v3);
      \node (l) at (308.5:2) {$\ell$};
      \draw[ very thick] (v3) -- (w2) node[midway,below] {};
      \draw[ very thick] (v3) -- (v6);
      \tikzset{xshift=4cm}
      \node[v,label=left:{$u$}] (u) at (0,0)  {$7$};
      \node[v,label=above right:{$v_1$}] (v1) at (0:1.5) {$5$};
      \node[v,label=above:{$v_2$}] (v2) at (103:1.5) {$4$};
      \node[v,label=left:{$v_3$}] (v3) at (257:1.5) {$5$};
      \node[v, very thick] (v4) at (154.5:1.5) {7};
      \node[v,label=right:{$w_2$}, very thick] (w1) at (51.5:1.5) {8};
      \node[v, very thick] (v6) at (206:1.5) {$8$};
      \node[v,label=below:{$w_1$}] (w2) at (308.5:1.5) {8};
      \node[v,xshift=1.5cm,label=above:{$w$}] (w) at (-45:1.5) {5};
      \draw (u) -- (v4) node[midway,above] {$a$};
      \draw (u) -- (v2) node[midway,right] {$b$};
      \draw (u) -- (w1) node[midway,right] {$c$};
      \draw (u) -- (v1) node[midway,above] {$d$};
      \draw (u) -- (w2) node[midway,right] {$e$};
      \draw (u) -- (v3) node[midway,right] {$f$};
      \draw (u) -- (v6) node[midway,below] {$g$};
      \draw (v2) -- (v4) node[midway,above] {$h$};
      \draw (v2) -- (w1) node[midway,above] {$i$};
      \draw (v1) -- (w1) node[midway,right] {$j$};
      \draw (v1) -- (w2) node[midway,right] {$k$};
      \draw (v1) -- (w) node[midway,above] {$\ell$};
      \draw (w2) -- (w) node[midway,below] {$m$};
      \draw (v3) -- (w2) node[midway,below] {$n$};
      \draw [ very thick] (v3) -- (v6);    
      \tikzset{xshift=5cm}
      \node[v,label=left:{$u$}] (u) at (0,0)  {$7$};
      \node[v,label=below right:{$v_1$}] (v1) at (0:1.5) {$5$};
      \node[v,label=above:{$v_2$}] (v2) at (103:1.5) {$4$};
      \node[v,label=left:{$v_3$}] (v3) at (257:1.5) {$5$};
      \node[v, very thick] (v4) at (154.5:1.5) {7};
      \node[v,label=above:{$w_2$}] (w1) at (51.5:1.5) {8};
      \node[v, very thick] (v6) at (206:1.5) {$8$};
      \node[v,label=right:{$w_1$},very thick] (w2) at (308.5:1.5) {8};
      \node[v,xshift=1.5cm,label=above:{$w$}] (w) at (45:1.5) {5};
      \draw (u) -- (v4) node[midway,above] {$a$};
      \draw (u) -- (v2) node[midway,right] {$b$};
      \draw (u) -- (w1) node[midway,right] {$c$};
      \draw (u) -- (v1) node[midway,above] {$d$};
      \draw (u) -- (w2) node[midway,right] {$e$};
      \draw (u) -- (v3) node[midway,right] {$f$};
      \draw (u) -- (v6) node[midway,below] {$g$};
      \draw (v2) -- (v4) node[midway,above] {$h$};
      \draw (v2) -- (w1) node[midway,above] {$i$};
      \draw (v1) -- (w1) node[midway,right] {$j$};
      \draw (v1) -- (w2) node[midway,right] {$k$};
      \draw (v1) -- (w) node[midway,above] {$\ell$};
      \draw (w1) -- (w) node[midway,above] {$m$};
      \draw[ very thick] (v3) -- (w2) node[midway,below] {};
      \draw[ very thick] (v3) -- (v6);
    \end{tikzpicture}
\caption{Notation for Lemma~\ref{lem:C8d}}
    \label{fig:C8d}
  \end{figure}
  We separate three cases depending on whether $w=v_3$, and whether
  $m=ww_1$ or $m=ww_2$. In each case, we color by minimality the graph
  $G'$ obtained from $G$ by removing $a,\ldots,\ell$ and the labeled
  edges incident to $m$ if $w\neq v_3$. We then uncolor
  $u,v_1,v_2,v_3$ and the endpoints of $m$ if $w\neq v_3$.
  
  Observe that if $w\neq v_3$, there is no edge $v_3w$ nor $v_1v_3$ in
  $E(G)$ due to $C_{\ref{C3a}}$. Moreover, since $v_3$ is a weak
  neighbor of $u$, we cannot have $g=uw$ either (otherwise,
  $v_1wv_3$ creates $C_{\ref{C3a}}$).

  We have
  \begin{enumerate}
  \item If $w=v_3$, $|\hat{g}|=|\hat{k}|=2$,
    $|\hat{e}|=|\hat{h}|=|\hat{i}|=|\hat{j}|=3$,
    $|\hat{v_3}|=|\hat{a}|=|\hat{c}|=4$,
    $|\hat{u}|=|\hat{v_1}|=|\hat{v_2}|=|\hat{\ell}|=6$, $|\hat{f}|=7$
    and $|\hat{b}|=|\hat{d}|=9$.
  \item If $m=ww_1$, $|\hat{w_1}|=|\hat{g}|=2$,
    $|\hat{h}|=|\hat{i}|=|\hat{j}|=|\hat{m}|=|\hat{n}|=3$,
    $|\hat{v_3}|=|\hat{w}|=|\hat{a}|=|\hat{c}|=4$, $|\hat{k}|=5$,
    $|\hat{v_2}|=|\hat{e}|=|\hat{\ell}|=6$,
    $|\hat{u}|=|\hat{v_1}|=|\hat{f}|=7$, and $|\hat{b}|=|\hat{d}|=9$.
  \item If $m=ww_2$, $|\hat{g}|=|\hat{k}|=2$,
    $|\hat{e}|=|\hat{h}|=|\hat{m}|=3$, $|\hat{w}|=|\hat{a}|=4$,
    $|\hat{c}|=|\hat{i}|=|\hat{j}|=5$, $|\hat{f}|=|\hat{\ell}|=6$,
    $|\hat{v_1}|=|\hat{v_2}|=|\hat{u}|=7$ and $|\hat{b}|=|\hat{d}|=9$. Moreover, $|\hat{v_3}|=|\hat{w_2}|$ is $2$ or $3$ depending on whether there is an edge $v_3w_2$.
  \end{enumerate}
  
  In each case, we forget $v_2$ and  conclude using Theorem~\ref{thm:nss}.
\end{proof}

 \subsection{Configuration $C_{\ref{C9}}$}
\begin{lemma}
  \label{lem:C9}
The graph $G$ does not contain a $7$-vertex $u$ with three weak neighbors
  $v_1,v_2,v_3$ of degree $4$ and a neighbor $v_4$ of degree $7$.
\end{lemma}

\begin{proof}
  As $G$ does not contain $C_{\ref{C2}}$, we may assume that $v_4$ is
  adjacent to only one vertex among $\{v_1,v_2,v_3\}$. Moreover, due
  to $C_{\ref{C1}}$, we may assume (up to renaming the vertices) that
  the situation is depicted in Figure~\ref{fig:C9}. By minimality, we
  color $G'=G\setminus\{a,\ldots,m\}$ and uncolor $u,v_1,v_2,v_3$.
  \begin{figure}[!h]
    \centering
    \begin{tikzpicture}[v/.style={draw=black,minimum size = 10pt,ellipse,inner sep=1pt}]
      \node[v,label=left:{$u$}] (u) at (0,0)  {$7$};
      \node[v,label=below right:{$v_2$}] (v1) at (0:1.5) {4};
      \node[v,label=above:{$v_1$}] (v2) at (103:1.5) {$4$};
      \node[v,label=left:{$v_3$}] (v3) at (257:1.5) {4};
      \node[v, very thick] (v4) at (154.5:1.5) {$8$};
      \node[v, very thick] (w1) at (51.5:1.5) {8};
      \node[v,label=left:{$v_4$}] (v6) at (206:1.5) {7};
      \node[v, very thick] (w2) at (308.5:1.5) {8};
      \draw (u) -- (v4) node[midway,above] {$a$};
      \draw (u) -- (v2) node[midway,right] {$b$};
      \draw (u) -- (w1) node[midway,right] {$c$};
      \draw (u) -- (v1) node[midway,above] {$d$};
      \draw (u) -- (w2) node[midway,right] {$e$};
      \draw (u) -- (v3) node[midway,right] {$f$};
      \draw (u) -- (v6) node[midway,below] {$g$};
      \draw (v2) -- (v4) node[midway,above] {$h$};
      \draw (v2) -- (w1) node[midway,above] {$i$};
      \draw (v1) -- (w1) node[midway,right] {$j$};
      \draw (v1) -- (w2) node[midway,right] {$k$};
      \draw (v3) -- (w2)  node[midway,below] {$\ell$};
      \draw (v3) -- (v6)  node[midway,above] {$m$};
    \end{tikzpicture}
\caption{Notation for Lemma~\ref{lem:C9}}
    \label{fig:C9}
  \end{figure}
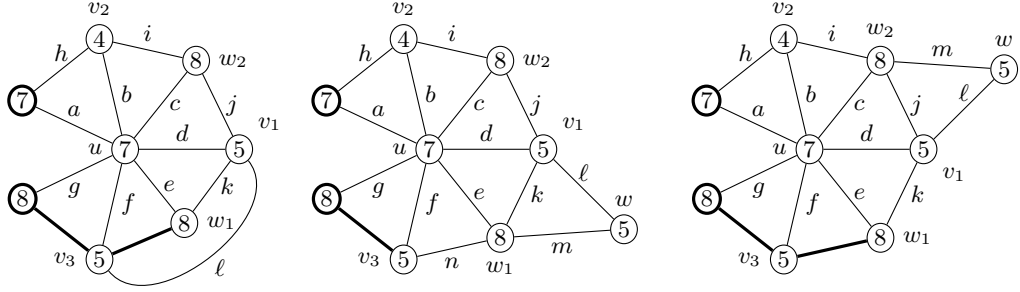

  We have $|\hat{h}|=2$,
  $|\hat{a}|=|\hat{i}|=|\hat{j}|=|\hat{k}|=|\hat{\ell}|=|\hat{m}|=3$,
  $|\hat{c}|=|\hat{e}|=|\hat{g}|=4$,
  $|\hat{u}|=|\hat{v_1}|=|\hat{v_2}|=|\hat{v_3}|=6$ and
  $|\hat{b}|=|\hat{d}|=|\hat{f}|=9$.

  We forget $v_1,v_2,v_3$ and  conclude using Theorem~\ref{thm:nss}.
\end{proof}

\subsection{Configuration $C_{\ref{C10}}$}
Due to the definitions of $C_{\ref{C10}}$ and $S_3$-neighbor, if $G$
contains $C_{\ref{C10}}$, then we are in one of the following cases:
\begin{itemize}
\item[$\bullet$] $C_{\ref{C10}a}$: $\dist_u(v_1,v_3)=2$ and the common neighbor
  $w$ of $v_1,u$ and $v_3$ has degree seven.
\item[$\bullet$] $C_{\ref{C10}b}$: $\dist_u(v_1,v_3)=3$ and $u,v_3$ share a
  common neighbor $w_1$ of degree six.
\item[$\bullet$] $C_{\ref{C10}c}$: $\dist_u(v_1,v_3)=3$ and $v_3$ has two
  neighbors $w_2,w_3$ of degree six.
\item[$\bullet$] $C_{\ref{C10}d}$: $\dist_u(v_1,v_3)=3$ and $v_3$ has a neighbor
  $w$ of degree five.
\end{itemize}
We dedicate a lemma to each of these configurations.

\begin{lemma}
\label{lem:C10a}
The graph $G$ does not contain $C_{\ref{C10}a}$.
\end{lemma}

\begin{proof}
  We use the notation depicted in Figure~\ref{fig:C10a}. By
  minimality, we color $G'=G\setminus\{a,\ldots,\ell\}$ and uncolor
  $u,v_1,v_2,v_3,w$.
  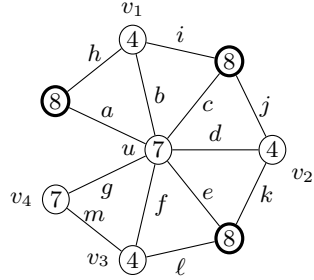
\begin{figure}[!h]
    \centering
    \begin{tikzpicture}[v/.style={draw=black,minimum size = 10pt,ellipse,inner sep=1pt}]
      \node[v,label=left:{$u$}] (u) at (0,0)  {$7$};
      \node[v,label=left:{$v_1$}] (v1) at (205.5:1.5) {4};
      \node[v,label=above:{$v_2$}] (v2) at (103.5:1.5) {$4$};
      \node[v,label=right:{$v_3$}] (v3) at (-51.5:1.5) {5};
      \node[v, very thick] (v4) at (51.5:1.5) {$8$};
      \node[v, very thick] (w1) at (154:1.5) {8};
      \node[v, very thick] (v6) at (0:1.5) {$8$};
      \node[v,label=below:{$w$}] (w2) at (-103:1.5) {7};
      \draw (u) -- (w2) node[midway,left] {$a$};
      \draw (u) -- (v1) node[midway,above] {$b$};
      \draw (u) -- (w1) node[midway,above] {$c$};
      \draw (u) -- (v2) node[midway,right] {$d$};
      \draw (u) -- (v4) node[midway,right] {$e$};
      \draw (u) -- (v6) node[midway,below] {$f$};
      \draw (u) -- (v3) node[midway,left] {$g$};
      \draw (v1) -- (w2) node[midway,below] {$h$};
      \draw (v1) -- (w1) node[midway,left] {$i$};
      \draw (v2) -- (w1) node[midway,above] {$j$};
      \draw (v2) -- (v4) node[midway,above] {$k$};
      \draw (v3) -- (w2)  node[midway,below] {$\ell$};
      \draw[ very thick] (v3) -- (v6);
    \end{tikzpicture}
\caption{Notation for Lemma~\ref{lem:C10a}}
    \label{fig:C10a}
  \end{figure}

We have $|\hat{f}|=|\hat{k}|=|\hat{w}|=2$,
  $|\hat{e}|=|\hat{i}|=|\hat{j}|=|\hat{\ell}|=3$,
  $|\hat{v_3}|=|\hat{c}|=4$, $|\hat{h}|=5$, $|\hat{a}|=|\hat{v_2}|=6$,
  $|\hat{u}|=|\hat{v_1}|=|\hat{g}|=7$ and $|\hat{b}|=|\hat{d}|=9$.

  We forget $v_1,v_2$ and  conclude using Theorem~\ref{thm:nss}.
\end{proof}

\begin{lemma}
  \label{lem:C10b}
  The graph $G$ does not contain $C_{\ref{C10}b}$.
\end{lemma}

\begin{proof}
  We use the notation depicted in Figure~\ref{fig:C10b}. By
  minimality, we color $G'=G\setminus\{a,\ldots,n\}$ and uncolor
  $u,v_1,v_2,v_3,w_1,w_2$.
  \begin{figure}[!h]
    \centering
    \begin{tikzpicture}[v/.style={draw=black,minimum size = 10pt,ellipse,inner sep=1pt}]
      \node[v,label=left:{$u$}] (u) at (0,0)  {$7$};
      \node[v,label=above:{$v_2$}] (v1) at (-257:1.5) {4};
      \node[v,label=left:{$v_1$}] (v2) at (-154:1.5) {$4$};
      \node[v,label=right:{$v_3$}] (v3) at (0:1.5) {5};
      \node[v,label=below:{$w_2$}] (v4) at (-102.5:1.5) {7};
      \node[v, very thick] (w1) at (-205.5:1.5) {8};
      \node[v,label=below:{$w_1$}] (v6) at (-51:1.5) {6};
      \node[v, very thick] (w2) at (51.5:1.5) {8};
      \draw (u) -- (v4) node[midway,left] {$a$};
      \draw (u) -- (v2) node[midway,above] {$b$};
      \draw (u) -- (w1) node[midway,above] {$c$};
      \draw (u) -- (v1) node[midway,right] {$d$};
      \draw (u) -- (w2) node[midway,right] {$e$};
      \draw (u) -- (v3) node[midway,below] {$f$};
      \draw (u) -- (v6) node[midway,left] {$g$};
      \draw (v2) -- (v4) node[midway,below] {$h$};
      \draw (v2) -- (w1) node[midway,left] {$i$};
      \draw (v1) -- (w1) node[midway,above] {$j$};
      \draw (v1) -- (w2) node[midway,above] {$k$};
      \draw (v3) -- (w2)  node[midway,right] {$\ell$};
      \draw (v3) -- (v6)  node[midway,right] {$m$};
      \draw (v4) -- (v6)  node[midway,below] {$n$};
    \end{tikzpicture}
\caption{Notation for Lemma~\ref{lem:C10b}}
    \label{fig:C10b}
  \end{figure}
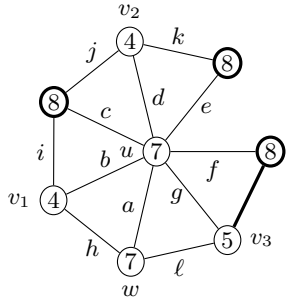
We have $|\hat{\ell}|=2$,
  $|\hat{i}|=|\hat{j}|=|\hat{k}|=|\hat{n}|=3$,
  $|\hat{w_1}|=|\hat{c}|=|\hat{e}|=4$, $|\hat{h}|=|\hat{m}|=5$,
  $|\hat{a}|=|\hat{v_2}|=6$, $|\hat{g}|=|\hat{v_1}|=7$,
  $|\hat{u}|=|\hat{f}|=8$ and $|\hat{b}|=|\hat{d}|=9$.
  
  If $v_3w_2\in E(G)$, we have $|\hat{v_3}|=6$ and
  $|\hat{w_2}|=3$. Otherwise, we may assume that $|\hat{v_3}|=5$ and
  $|\hat{w_2}|=2$.

  We forget $v_1,v_2$ and conclude using Theorem~\ref{thm:nss}.
\end{proof}

\begin{lemma}
  \label{lem:C10c}
  The graph $G$ does not contain $C_{\ref{C10}c}$.
\end{lemma}

\begin{proof}
  We use the notation depicted in Figure~\ref{fig:C10c}. By minimality,
  we color $G\setminus\{a,\ldots,q\}$ and uncolor
  $u,v_1,v_2,v_3,w_1,w_2,w_3$.
  \begin{figure}[!h]
    \centering
    \begin{tikzpicture}[v/.style={draw=black,minimum size = 10pt,ellipse,inner sep=1pt}]
      \node[v,label=left:{$u$}] (u) at (0,0)  {$7$};
      \node[v,label=right:{$v_3$}] (v1) at (0:1.5) {$5$};
      \node[v,label=above:{$v_2$}] (v2) at (103:1.5) {$4$};
      \node[v, very thick] (v3) at (257:1.5) {7};
      \node[v, very thick] (v4) at (154.5:1.5) {8};
      \node[v,label=above:{$w_1$}] (w1) at (51.5:1.5) {8};
      \node[v,label=left:{$v_1$}] (v6) at (206:1.5) {4};
      \node[v, very thick] (w2) at (308.5:1.5) {$8$};
      \node[v,xshift=1.5cm,label=above:{$w_2$}] (x1) at (36:1.5) {6};
      \node[v,xshift=1.5cm,label=below:{$w_3$}] (x2) at (-36:1.5) {6};
      \draw (u) -- (v3) node[midway,right] {$a$};
      \draw (u) -- (v6) node[midway,below] {$b$};
      \draw (u) -- (v4) node[midway,above] {$c$};
      \draw (u) -- (v2) node[midway,right] {$d$};
      \draw (u) -- (w1) node[midway,right] {$e$};
      \draw (u) -- (v1) node[midway,above] {$f$};
      \draw (u) -- (w2) node[midway,right] {$g$};
      \draw (v6) -- (v3) node[midway,below] {$h$};      
      \draw (v4) -- (v6) node[midway,left] {$i$};
      \draw (v2) -- (v4) node[midway,above] {$j$};
      \draw (v2) -- (w1) node[midway,above] {$k$};
      \draw (v1) -- (w1) node[midway,right] {$\ell$};
      \draw (v1) -- (w2) node[midway,right] {$m$};
      \draw (w1) -- (x1) node[midway,above] {$n$};
      \draw (v1) -- (x1) node[midway,above] {$o$};
      \draw (v1) -- (x2) node[midway,below] {$p$};
      \draw (x2) -- (x1) node[midway,right] {$q$};
      \draw [ very thick] (x2) -- (w2) ;
      \draw[ very thick] (v3) -- (w2) ;
    \end{tikzpicture}
\caption{Notation for Lemma~\ref{lem:C10c}}
    \label{fig:C10c}
  \end{figure}
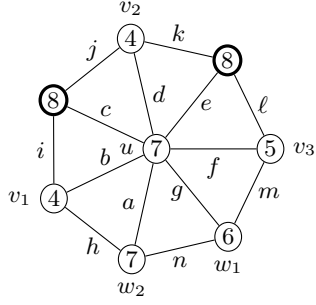  
We have
  $|\hat{g}|=|\hat{h}|=|\hat{i}|=|\hat{j}|=|\hat{m}|=|\hat{n}|=|\hat{q}|=3$,
  $|\hat{w_2}|=|\hat{a}|=|\hat{c}|=4$, $|\hat{k}|=5$,
  $|\hat{v_1}|=|\hat{e}|=|\hat{\ell}|=|\hat{p}|=6$,
  $|\hat{u}|=|\hat{v_2}|=|\hat{o}|=7$,
  $|\hat{v_3}|=|\hat{b}|=|\hat{d}|=9$ and $|\hat{f}|=10$. Moreover,
  depending on the presence of the edge $w_1w_3$ in $E(G)$ we may
  assume that $|\hat{w_1}|,|\hat{w_3}|$ are $2$ or $3$.

  We forget $v_1,v_2$ and we color $a$ with a color not in
  $\hat{h}$. Then we forget $h,b,i,d,j,k$ and color $c$ with a color
  not in $\hat{g}$.

  Note that if $w_1w_3\notin E(G)$, then we may assume that
  $\hat{w_1}\cap\hat{w_3}=\varnothing$, since otherwise we color them
  with the same color, then color $g,m,n,q,w_2,e,\ell,o,p,u,f,v_3$.

  We remove a color $\alpha\in\hat{v_3}\setminus\hat{f}$ from
  $\hat{w_1},\hat{w_2},\hat{w_3}$. Assume that we can color every
  element excepted $v_3$ and $f$. Either $\alpha$ appears on
  $\ell,m,o,p$ or $u$ so $|\hat{v_3}|=1$ and $|\hat{f}|=2$, or $\alpha$
  is still in $\hat{v_3}$ at the end, so $\hat{v_3}\neq\hat{f}$. Thus
  we can extend the coloring to $v_3$ and $f$.  We may thus forget
  $v_3$ and $f$, and then $o,p,\ell,m,u,e,g$ also. Considering the
  edge $w_1w_3$, we have two cases:
  \begin{enumerate}
  \item If $w_1w_3\notin E(G)$, $\hat{w_1}$ and $\hat{w_3}$ are
    disjoint, so at most one of them (say $w$) loses a color when we
    removed $\alpha$. We color $w$, then apply Lemma~\ref{lem:diam} to
    $\mathcal{T}(G)$ with the path $nw_2qw_3$ if $w=w_1$ or $qw_2nw_1$ if
    $w=w_3$.
  \item Assume that $w_1w_3\in E(G)$. If $\hat{w_1}=\hat{w_3}$, we
    color $w_2$ with a color not in $\hat{w_1}$, then apply
    Corollary~\ref{cor:evencycle} on the cycle $w_1w_3qn$ in
    $\mathcal{T}(G)$. Otherwise, we color $w_1$ with a color not in $\hat{w_3}$,
    then apply Lemma~\ref{lem:diam} on $\mathcal{T}(G)$ with the path
    $w_3qw_2n$.
  \end{enumerate}
\end{proof}

\begin{lemma}
  \label{lem:C10d}
  The graph $G$ does not contain $C_{\ref{C10}d}$.
\end{lemma}

\begin{proof}
  We use the notation depicted in Figure~\ref{fig:C10d}. By
  minimality, we color $G'=G\setminus\{a,\ldots,n\}$ and uncolor
  $u,v_1,v_2,v_3,w$.
  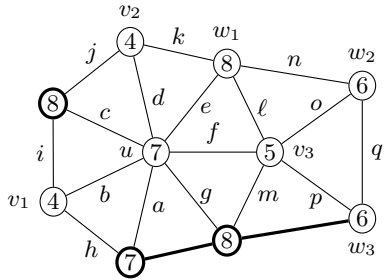
\begin{figure}[!h]
    \centering
    \begin{tikzpicture}[v/.style={draw=black,minimum size = 10pt,ellipse,inner sep=1pt}]
      \node[v,label=left:{$u$}] (u) at (0,0)  {$7$};
      \node[v,label=below right:{$v_3$}] (v1) at (0:1.5) {$5$};
      \node[v,label=above:{$v_2$}] (v2) at (103:1.5) {$4$};
      \node[v, very thick] (v3) at (257:1.5) {7};
      \node[v, very thick] (v4) at (154.5:1.5) {8};
      \node[v,label=above:{$w_1$}] (w1) at (51.5:1.5) {8};
      \node[v,label=left:{$v_1$}] (v6) at (206:1.5) {4};
      \node[v, very thick] (w2) at (308.5:1.5) {$8$};
      \node[v,label=above:{$w$}] (x1) at (0:3) {5};
      \draw (u) -- (v3) node[midway,right] {$a$};
      \draw (u) -- (v6) node[midway,below] {$b$};
      \draw (u) -- (v4) node[midway,above] {$c$};
      \draw (u) -- (v2) node[midway,right] {$d$};
      \draw (u) -- (w1) node[midway,right] {$e$};
      \draw (u) -- (v1) node[midway,above] {$f$};
      \draw (u) -- (w2) node[midway,right] {$g$};
      \draw (v6) -- (v3) node[midway,below] {$h$};      
      \draw (v4) -- (v6) node[midway,left] {$i$};
      \draw (v2) -- (v4) node[midway,above] {$j$};
      \draw (v2) -- (w1) node[midway,above] {$k$};
      \draw (v1) -- (w1) node[midway,right] {$\ell$};
      \draw (v1) -- (w2) node[midway,right] {$m$};
      \draw (v1) -- (x1) node[midway,above] {$n$};
      \draw [ very thick](v3) -- (w2) ;
    \end{tikzpicture}
\caption{Notation for Lemma~\ref{lem:C10d}}
    \label{fig:C10d}
  \end{figure}
  We have $|\hat{w}|=|\hat{m}|=2$,
  $|\hat{g}|=|\hat{h}|=|\hat{i}|=|\hat{j}|=|\hat{k}|=|\hat{\ell}|=3$,
  $|\hat{a}|=|\hat{c}|=|\hat{e}|=4$, $|\hat{n}|=5$,
  $|\hat{u}|=|\hat{v_1}|=|\hat{v_2}|=|\hat{v_3}|=6$ and
  $|\hat{b}|=|\hat{d}|=|\hat{f}|=9$.

  We forget $v_1,v_2$ and conclude using Theorem~\ref{thm:nss}.
\end{proof}

\subsection{Configuration $C_{\ref{C11}}$}
Note that, when $v_3$ is an $S_3$-neighbor of $u$, it cannot be a
$(6,6^+)$-neighbor of $u$ otherwise we obtain $C_{\ref{C1}}$ since the
$6$-vertex would be adjacent to $v_1$ or $v_2$. Thus, due to the
definitions of $C_{\ref{C11}}$ and $S_3$-neighbor, if $G$ contains
$C_{\ref{C11}}$, then we are in one of the following cases:
\begin{itemize}
\item[$\bullet$] $C_{\ref{C11}a}$: $v_1$ is a $(7,7)$-neighbor of $u$.
\item[$\bullet$] $C_{\ref{C11}b}$: $v_3$ is a $(7,7)$-neighbor of $u$.
\item[$\bullet$] $C_{\ref{C11}c}$: $v_1$ is a $(7,8)$-neighbor of $u$ and $v_3$
  has two neighbors $w_2,w_3$ of degree $6$ such that $uv_3w_2$,
  $uv_3w_3$ are not triangular faces.
\item[$\bullet$] $C_{\ref{C11}d}$: $v_1$ is a $(7,8)$-neighbor of $u$ and $v_3$
  has a neighbor $w$ of degree $5$ such that $uv_3w$ is not a
  triangular face.
\end{itemize}
We dedicate a lemma to each of these configurations.

\begin{lemma}
  \label{lem:C11a}
The graph $G$ does not contain $C_{\ref{C11}a}$.
\end{lemma}

\begin{proof}
  We use the notation depicted in Figure~\ref{fig:C11a}. By
  minimality, we color $G'=G\setminus\{a,\ldots,n\}$ and uncolor
  $u,v_1,v_2,v_3,w$.
  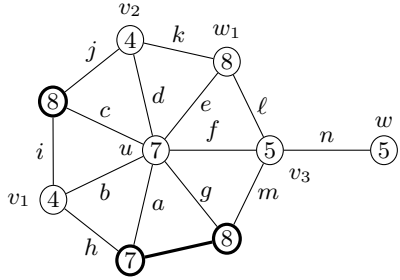
\begin{figure}[!h]
    \centering
    \begin{tikzpicture}[v/.style={draw=black,minimum size = 10pt,ellipse,inner sep=1pt}]
      \node[v,label=left:{$u$}] (u) at (0,0)  {$7$};
      \node[v,label=right:{$v_3$}] (v1) at (0:1.5) {5};
      \node[v,label=above:{$v_1$}] (v2) at (103:1.5) {4};
      \node[v,label=below:{$v_2$}] (v3) at (257:1.5) {4};
      \node[v, very thick] (v4) at (154.5:1.5) {7};
      \node[v,label=right:{$w$}] (w1) at (51.5:1.5) {7};
      \node[v, very thick] (v6) at (206:1.5) {8};
      \node[v, very thick] (w2) at (308.5:1.5) {8};
      \draw (u) -- (v4) node[midway,above] {$a$};
      \draw (u) -- (v2) node[midway,right] {$b$};
      \draw (u) -- (w1) node[midway,right] {$c$};
      \draw (u) -- (v1) node[midway,above] {$d$};
      \draw (u) -- (w2) node[midway,right] {$e$};
      \draw (u) -- (v3) node[midway,right] {$f$};
      \draw (u) -- (v6) node[midway,below] {$g$};
      \draw (v2) -- (v4) node[midway,above] {$h$};
      \draw (v2) -- (w1) node[midway,above] {$i$};
      \draw (v1) -- (w1) node[midway,right] {$j$};
      \draw (v1) -- (w2) node[midway,right] {$k$};
      \draw (v3) -- (w2)  node[midway,below] {$\ell$};
      \draw (v3) -- (v6)node[midway,below] {$m$};
    \end{tikzpicture}
\caption{Notation for Lemma~\ref{lem:C11a}}
    \label{fig:C11a}
  \end{figure}

  We have $|\hat{w}|=|\hat{k}|=|\hat{m}|=2$,
  $|\hat{g}|=|\hat{h}|=|\hat{\ell}|=3$,
  $|\hat{a}|=|\hat{e}|=|\hat{j}|=4$, $|\hat{v_3}|=|\hat{i}|=5$,
  $|\hat{v_2}|=|\hat{c}|=6$, $|\hat{u}|=|\hat{v_1}|=7$, $|\hat{d}|=8$
  and $|\hat{b}|=|\hat{f}|=9$.

  We forget $v_1,v_2$ and conclude using Theorem~\ref{thm:nss}.
\end{proof}

\begin{lemma}
  \label{lem:C11b}
The graph $G$ does not contain $C_{\ref{C11}b}$.
\end{lemma}

\begin{proof}
  We use the notation depicted in Figure~\ref{fig:C11b}. By
  minimality, we color $G'=G\setminus\{a,\ldots,m\}$ and uncolor
  $u,v_1,v_2,v_3,w_1,w_2$.
  \begin{figure}[!h]
    \centering
    \begin{tikzpicture}[v/.style={draw=black,minimum size = 10pt,ellipse,inner sep=1pt}]
      \node[v,label=left:{$u$}] (u) at (0,0)  {$7$};
      \node[v,label=right:{$v_3$}] (v1) at (0:1.5) {5};
      \node[v,label=above:{$v_1$}] (v2) at (103:1.5) {4};
      \node[v,label=below:{$v_2$}] (v3) at (257:1.5) {4};
      \node[v, very thick] (v4) at (154.5:1.5) {8};
      \node[v,label=right:{$w_2$}] (w1) at (51.5:1.5) {7};
      \node[v, very thick] (v6) at (206:1.5) {8};
      \node[v,label=right:{$w_1$}] (w2) at (308.5:1.5) {7};
      \draw (u) -- (v4) node[midway,above] {$a$};
      \draw (u) -- (v2) node[midway,right] {$b$};
      \draw (u) -- (w1) node[midway,right] {$c$};
      \draw (u) -- (v1) node[midway,above] {$d$};
      \draw (u) -- (w2) node[midway,right] {$e$};
      \draw (u) -- (v3) node[midway,right] {$f$};
      \draw (u) -- (v6) node[midway,below] {$g$};
      \draw (v2) -- (v4) node[midway,above] {$h$};
      \draw (v2) -- (w1) node[midway,above] {$i$};
      \draw (v1) -- (w1) node[midway,right] {$j$};
      \draw (v1) -- (w2) node[midway,right] {$k$};
      \draw (v3) -- (w2)  node[midway,below] {$\ell$};
      \draw (v3) -- (v6)node[midway,below] {$m$};
    \end{tikzpicture}
\caption{Notation for Lemma~\ref{lem:C11b}}
    \label{fig:C11b}
  \end{figure}
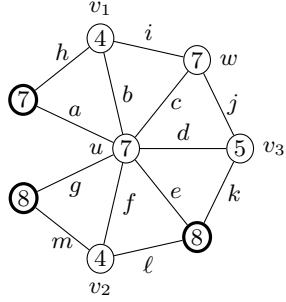

We have $|\hat{h}|=|\hat{m}|=2$,
  $|\hat{a}|=|\hat{g}|=3$, $|\hat{j}|=|\hat{k}|=4$,
  $|\hat{i}|=|\hat{\ell}|=5$, $|\hat{v_3}|=|\hat{c}|=|\hat{e}|=6$,
  $|\hat{v_1}|=|\hat{v_2}|=7$, $|\hat{u}|=|\hat{d}|=8$ and
  $|\hat{b}|=|\hat{f}|=9$. Moreover, $|\hat{w_1}|=|\hat{w_2}|$ is $2$ or $3$ depending on whether there is an edge $w_1w_2$.

  We forget $v_1,v_2$ and conclude using Theorem~\ref{thm:nss}.
\end{proof}

\begin{lemma}
  \label{lem:C11c}
The graph $G$ does not contain $C_{\ref{C11}c}$.
\end{lemma}

\begin{proof}
  We use the notation depicted in Figure~\ref{fig:C11c}. Recall that
  $v_1$ is a $(7,8)$-neighbor of $u$, hence $w$ or $w_4$ has degree
  7. By minimality, we color $G\setminus\{a,\ldots,r\}$ and uncolor
  $u,v_1,v_2,v_3,w_1,w_2,w_3,w_4$.
  \begin{figure}[!h]
    \centering
     \begin{tikzpicture}[v/.style={draw=black,minimum size = 10pt,ellipse,inner sep=1pt}]
      \node[v,label=left:{$u$}] (u) at (0,0)  {$7$};
      \node[v,label=right:{$v_3$}] (v1) at (0:1.5) {$5$};
      \node[v,label=above:{$v_1$}] (v2) at (103:1.5) {$4$};
      \node[v,label=below:{$v_2$}] (v3) at (257:1.5) {4};
      \node[v, very thick,label=above:{$w$}] (v4) at (154.5:1.5) {};
      \node[v,label=above:{$w_4$}] (w1) at (51.5:1.5) {};
      \node[v, very thick] (v6) at (206:1.5) {8};
      \node[v,label=below:{$w_1$}] (w2) at (308.5:1.5) {8};
      \node[v,xshift=1.5cm,label=above:{$w_3$}] (x1) at (36:1.5) {6};
      \node[v,xshift=1.5cm,label=below:{$w_2$}] (x2) at (-36:1.5) {6};
      \draw (u) -- (v4) node[midway,above] {$a$};
      \draw (u) -- (v2) node[midway,right] {$b$};
      \draw (u) -- (w1) node[midway,right] {$c$};
      \draw (u) -- (v1) node[midway,above] {$d$};
      \draw (u) -- (w2) node[midway,right] {$e$};
      \draw (u) -- (v3) node[midway,right] {$f$};
      \draw (u) -- (v6) node[midway,below] {$g$};
      \draw (v1) -- (w1) node[midway,right] {$j$};
      \draw (v1) -- (x1) node[midway,above] {$o$};
      \draw (v1) -- (x2) node[midway,below] {$p$};
      \draw (v1) -- (w2) node[midway,right] {$k$};
      \draw (v2) -- (v4) node[midway,above] {$h$};
      \draw (v2) -- (w1) node[midway,above] {$i$};
      \draw (w1) -- (x1) node[midway,above] {$n$};
      \draw (x2) -- (x1) node[midway,right] {$r$};
      \draw (x2) -- (w2) node[midway,below] {$q$};
      \draw (v3) -- (w2) node[midway,below] {$\ell$};
      \draw (v6) -- (v3) node[midway,below] {$m$};      
    \end{tikzpicture}
\caption{Notation for Lemma~\ref{lem:C11c}}
    \label{fig:C11c}
  \end{figure}
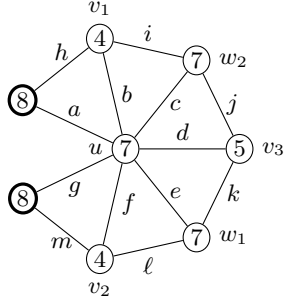
  We have $|\hat{m}|=2$, $|\hat{g}|=|\hat{q}|=3$, $|\hat{r}|=4$,
  $|\hat{\ell}|=5$, $|\hat{e}|=|\hat{k}|=6$,
  $|\hat{o}|=|\hat{p}|=|\hat{v_1}|=|\hat{v_2}|=7$, $|\hat{u}|=8$,
  $|\hat{b}|=|\hat{f}|=9$ and $|\hat{v_3}|=|\hat{d}|=10$. Moreover,
  depending on the presence of edges between the $w_i$'s, their lists
  size may vary, but we may assume that $|\hat{w_1}|\geqslant 2$ and
  $|\hat{w_2}|,|\hat{w_3}|$ are at least $4$. We forget $v_1,v_2$.

  We separate two cases depending on the degrees of $w_4$ and $w$:
  \begin{enumerate}
  \item We first assume that $d(w_4)=8$ and $d(w)=7$. Then we may also
    assume that $|\hat{h}|=|\hat{n}|=3$, $|\hat{a}|=4$, $|\hat{i}|=5$,
    $|\hat{c}|=|\hat{j}|=6$ and $|\hat{w_4}|\geqslant 2$.

    We remove from $\hat{w_3}$ and $\hat{r}$ a color
    $\alpha\in\hat{o}\setminus\hat{j}$, if it appears in these lists. We then color $w_2$
    with a color not in $\hat{r}$, then $w_1$ and $q$, and apply
    Lemma~\ref{lem:diam} on $\mathcal{T}(G)$ with the path $w_4nw_3r$.

    Due to the choice of $\alpha$, we may now color $j$ with a color
    not in $\hat{o}$, then color $i,c,h,a,g,e,m,\ell ,k$. We color $u$
    such that $\hat{v_3}\neq \hat{o}$, then $b,f,d,p$. Since
    $\hat{v_3}\neq\hat{o}$, we can finally color $v_3$ and $o$.
  \item Assume that $d(w_4)=7$ and $d(w)=8$. We may assume that
    $|\hat{h}|=2$, $|\hat{a}|=3$, $|\hat{n}|=4$, $|\hat{i}|=6$,
    $|\hat{c}|=|\hat{j}|=7$ and $|\hat{w_4}|\geqslant 4$.

    We color $g,\ell $ with a color not in $\hat{m}$. Then we forget
    $m,f,b,i,h$ and color $w_1$, $q$, $a$, $e$, $k$, $w_2$, $r$, $w_3$,
    $n$, $o$, $p$. We then color $v_3$ with a color not in
    $\hat{w_4}$. If $\hat{w_4}=\hat{j}$, we color $c$ with a color not
    in $\hat{j}$, then apply Corollary~\ref{cor:evencycle} on the
    cycle $uw_4jd$ in $\mathcal{T}(G)$. Otherwise, we color $w_4$ with
    a color not in $\hat{j}$, then apply Lemma~\ref{lem:diam} on
    $\mathcal{T}(G)$ with the path $jdcu$.
  \end{enumerate}
\end{proof}

\begin{lemma}
  \label{lem:C11d}
The graph $G$ does not contain $C_{\ref{C11}d}$.
\end{lemma}

\begin{proof}
  We use the notation depicted in Figure~\ref{fig:C11d}. By
  minimality, we color $G'=G\setminus\{a,\ldots,n\}$ and uncolor
  $u,v_1,v_2,v_3,w_1$.
  \begin{figure}[!h]
    \centering
    \begin{tikzpicture}[v/.style={draw=black,minimum size = 10pt,ellipse,inner sep=1pt}]
      \node[v,label=left:{$u$}] (u) at (0,0)  {$7$};
      \node[v,label=below right:{$v_3$}] (v1) at (0:1.5) {$5$};
      \node[v,label=above:{$v_1$}] (v2) at (103:1.5) {$4$};
      \node[v,label=below:{$v_2$}] (v3) at (257:1.5) {4};
      \node[v,label=left:{$w_3$},very thick] (v4) at (154.5:1.5) {};
      \node[v,label=above:{$w_2$},very thick] (w1) at (51.5:1.5) {};
      \node[v, very thick] (v6) at (206:1.5) {8};
      \node[v, very thick] (w2) at (308.5:1.5) {8};
      \node[v,label=above:{$w_1$}] (x1) at (0:3) {5};
      \draw (u) -- (v4) node[midway,above] {$a$};
      \draw (u) -- (v2) node[midway,right] {$b$};
      \draw (u) -- (w1) node[midway,right] {$c$};
      \draw (u) -- (v1) node[midway,above] {$d$};
      \draw (u) -- (w2) node[midway,right] {$e$};
      \draw (u) -- (v3) node[midway,right] {$f$};
      \draw (u) -- (v6) node[midway,below] {$g$};
      \draw (v1) -- (w1) node[midway,right] {$j$};
      \draw (v1) -- (x1) node[midway,above] {$n$};
      \draw (v1) -- (w2) node[midway,right] {$k$};
      \draw (v2) -- (v4) node[midway,above] {$h$};
      \draw (v2) -- (w1) node[midway,above] {$i$};
      \draw (v3) -- (w2) node[midway,above] {$\ell$};
      \draw (v6) -- (v3) node[midway,below] {$m$};      
    \end{tikzpicture}
\caption{Notation for Lemma~\ref{lem:C11d}}
    \label{fig:C11d}
  \end{figure}
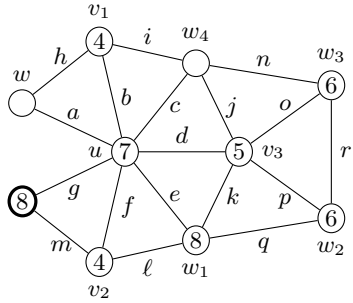

We have $|\hat{w_1}|=|\hat{m}|=2$,
  $|\hat{g}|=|\hat{k}|=|\hat{\ell}|=3$, $|\hat{e}|=4$, $|\hat{n}|=5$,
  $|\hat{u}|=|\hat{v_1}|=|\hat{v_2}|=|\hat{v_3}|=6$ and
  $|\hat{b}|=|\hat{d}|=|\hat{f}|=9$.

  Moreover, if $d(w_2)=8$ and $d(w_3)=7$, we have
  $|\hat{h}|=|\hat{i}|=|\hat{j}|=3$ and
  $|\hat{a}|=|\hat{c}|=4$. Otherwise, $d(w_2)=7$ and $d(w_3)=8$ so
  $|\hat{h}|=2$, $|\hat{a}|=3$, $|\hat{i}|=|\hat{j}|=4$ and
  $|\hat{c}|=5$.  We forget $v_1,v_2$ and conclude using
  Theorem~\ref{thm:nss}.
\end{proof}

\subsection{Configuration $C_{\ref{C12}}$}
By definition, if $G$ contains $C_{\ref{C12}}$, then we are in one of
the following cases ($v_1,\ldots,v_8$ denote the neighbors of $u$ in
cyclic ordering around $u$):
\begin{itemize}
\item[$\bullet$] $C_{\ref{C12}a}$: $u$ has four neighbors of degree
  $6$, and four $(6,6)$-neighbors of degree $5$. We may assume that
  $d(v_{2i})=5$ and $d(v_{2i-1})=6$ for $1\leqslant i\leqslant 4$ and
  that $v_1v_2,\ldots,v_7v_8,v_8v_1$ are in $E(G)$.
\item[$\bullet$] $C_{\ref{C12}b}$: $u$ has five weak neighbors of degree $5$ and
  three neighbors of degree $6$. Due to $C_{\ref{C3a}}$ and
  $C_{\ref{C3b}}$, we may assume that $v_1,v_2,v_4,v_6,v_7$ have
  degree $5$, $v_3,v_5,v_8$ have degree $6$ and that
  $v_1v_2,\ldots,v_7v_8,v_8v_1$ are in $E(G)$.
\item[$\bullet$] $C_{\ref{C12}c}$: $u$ has four neighbors of degree $6$, two
  $(6,6)$-neighbors of degree $5$ at triangle-distance 2, and two
  $(5,6)$-neighbors of degree $5$. We may assume that
  $v_2,v_4,v_6,v_7$ have degree $5$, $v_1,v_3,v_5,v_8$ have degree $6$
  and that $v_1v_2,\ldots,v_7v_8$ are in $E(G)$.
\item[$\bullet$] $C_{\ref{C12}d}$: $u$ has four neighbors of degree $6$, two
  $(6,6)$-neighbors of degree $5$ at triangle-distance at least $3$,
  and two $(5,6)$-neighbors of degree $5$. We may assume that
  $v_2,v_4,v_5,v_7$ have degree $5$, $v_1,v_3,v_6,v_8$ have degree $6$
  and that $v_1v_2,\ldots,v_7v_8$ are in $E(G)$.
\end{itemize}
We dedicate a lemma to each of these configurations. In each of them,
we did not succeed in finding a suitable monomial for the
Nullstellensatz approach, hence we only present case analysis proofs.

\begin{lemma}
  \label{lem:C12a}
The graph $G$ does not contain $C_{\ref{C12}a}$.
\end{lemma}

\begin{proof}
  We use the notation depicted in Figure~\ref{fig:C12a}.  By
  minimality, we color $G\setminus\{a,\ldots,p\}$ and uncolor
  $u,v_1,\ldots,v_8$.
  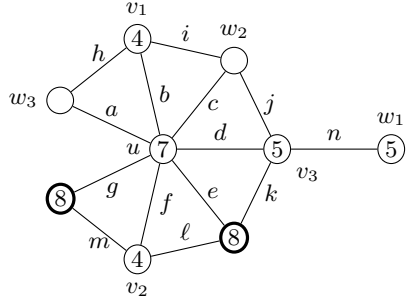
\begin{figure}[!h]
    \centering
    \begin{tikzpicture}[v/.style={draw=black,minimum size = 10pt,ellipse,inner sep=1pt}]
      \node[v,label=85:{$u$}] (u) at (0,0)  {8};
      \node[v,label=right:{$v_7$}] (v0) at (45:1.5) {6};
      \node[v,label=above:{$v_6$}] (v1) at (90:1.5) {5};
      \node[v,label=left:{$v_5$}] (v2) at (135:1.5) {6};
      \node[v,label=left:{$v_4$}] (v3) at (180:1.5) {5};
      \node[v,label=left:{$v_3$}] (v4) at (225:1.5) {6};
      \node[v,label=below:{$v_2$}] (v5) at (270:1.5) {5};
      \node[v,label=right:{$v_1$}] (v6) at (315:1.5) {6};
      \node[v,label=right:{$v_8$}] (v7) at (0:1.5) {5};
      \draw (v6) -- (u) node[midway,right] {$a$};
      \draw (v5) -- (u) node[midway,right] {$b$};
      \draw (v4) -- (u) node[midway,below] {$c$};
      \draw (v3) -- (u) node[midway,below] {$d$};
      \draw (v2) -- (u) node[midway,left] {$e$};
      \draw (v1) -- (u) node[midway,left] {$f$};
      \draw (v0) -- (u) node[midway,above] {$g$};
      \draw (v7) -- (u) node[midway,above] {$h$};
      \draw (v5) -- (v6) node[midway,below] {$i$};
      \draw (v5) -- (v4) node[midway,below] {$j$};
      \draw (v4) -- (v3) node[midway,left] {$k$};
      \draw (v3) -- (v2) node[midway,left] {$\ell$};
      \draw (v2) -- (v1) node[midway,above] {$m$};
      \draw (v1) -- (v0) node[midway,above] {$n$};
      \draw (v0) -- (v7) node[midway,right] {$o$};
      \draw (v6) -- (v7) node[midway,right] {$p$};
    \end{tikzpicture}
\caption{Notation for Lemma~\ref{lem:C12a}}
    \label{fig:C12a}
  \end{figure}
  First note that there is no edge between $5$-vertices excepted
  maybe $v_2v_6$ and $v_4v_8$ since otherwise, it would create
  $C_{\ref{C3b}}$.

  Using that $G$ is planar, we first show the following:
  \begin{enumerate}
  \item We may assume (up to symmetry) that there is no edge between
    $v_2$ and $v_5,v_6,v_7$.

    Assume that $v_6$ or $v_7$ is a neighbor of $v_2$. Then there is
    no edge between $v_8$ and $v_3,v_4,v_5$, otherwise,
    $\{\{u\},\{v_3,v_4,v_5\},\{v_8\},\{v_6,v_7\},\{v_1,v_2\}\}$ is a
    $K_5$-minor of $G$. By exchanging $v_2$ and $v_8$, we obtain that
    $v_2$ has no neighbor among $v_5,v_6,v_7$.
    
    If $v_2v_5$ is an edge, we obtain the same result by exchanging
    $v_2$ and $v_4$.
  \item With such a $v_2$, we may also assume that $v_4$ has at most
    one neighbor among $v_1,v_7,v_8$. First note that if
    $v_4v_8\in E(G)$, then $v_1,v_7$ are not neighbors of $v_4$ due to
    $C_{\ref{C3b}}$. In this case, $v_4$ has thus only one neighbor
    among $v_1,v_7,v_8$.

    Otherwise, both $v_1$ and $v_7$ are neighbors of $v_4$, so there
    is no edge between $vv_8$ with $v\in\{v_3,v_5\}$. Indeed,
    otherwise, $\{u,v,v_1,v_4,v_7,v_8\}$ would be a $K_{3,3}$ minor of
    $G$. Thus, by exchanging $v_4$ and $v_8$, we obtain that $v_4$ has
    at most one neighbor among $v_1,v_7,v_8$.
  \end{enumerate}
  
  We thus have
  $|\hat{i}|=|\hat{j}|=|\hat{k}|=|\hat{\ell}|=|\hat{m}|=|\hat{n}|=|\hat{o}|=|\hat{p}|=5$,
  $|\hat{v_2}|=6$, $|\hat{a}|=|\hat{c}|=|\hat{e}|=|\hat{g}|=7$,
  $|\hat{b}|=|\hat{d}|=|\hat{f}|=|\hat{h}|=8$ and $|\hat{u}|=10$.
  Moreover, $\hat{v_1},\hat{v_3},\hat{v_5}$ and $\hat{v_7}$ have size
  at least $4$, and $\hat{v_6},\hat{v_8}$ at least $6$.

  Due to the previous observations, we may also assume that
  $|\hat{v_4}|$ is $6$ or $7$. We separate three cases:
  
  \begin{enumerate}
  \item Assume that $\hat{n}\not\subset\hat{v_6}$. Then we color $n$
    with a color not in $\hat{v_6}$, $d$ with a color not in
    $\hat{v_4}$, $g$ with a color not in $\hat{o}$ and $h$ with a
    color not in $\hat{p}$. We then color
    $a,c,e,f,b,u,v_7,m,v_5,\ell$, and forget $v_4,v_6$. We color $v_8$
    with a color not in $\hat{o}$, then $v_1$ and forget $o,p$. We
    finally apply Lemma~\ref{lem:diam} on $\mathcal{T}(G)$ with the path $iv_2jv_3k$.
  \item If $\hat{n}\subset\hat{v_6}$ (and by symmetry
    $\hat{i}\subset\hat{v_2}$) and $v_6$ is not a neighbor from both
    $v_1,v_3$. Then $|\hat{v_6}|<8$ and we can color $f$ with a color
    not in $\hat{v_6}$ (hence not in $\hat{n}$), and $b$ with a color
    not in $\hat{v_2}$ (hence not in $\hat{i}$). Then we color
    $a,c,e,g,h,d,u$ and forget $v_2,i,j,v_6,m,n$, and use
    Theorem~\ref{thm:=deg} to color
    $\{v_1,v_3,v_4,v_5,v_7,v_8,k,\ell,o,p\}$.
  \item Otherwise, we color $b$ with a color not in $\hat{v_2}$, then
    color $a,c,e,g,d,f,h,u, v_3$ and forget $v_2,i,j$. We apply
    Lemma~\ref{lem:diam} on $\mathcal{T}(G)$ with the path $kv_4\ell v_5$, then
    color $m$.

    If $v_1v_7\not\in E(G)$, then $|\hat{v_7}|=2=|\hat{n}|$. If
    $\hat{n}=\hat{v_7}$, we color $v_6$ and $o$ with a color not in
    $\hat{n}$, then color $v_1,p,v_8,v_7,n$. Otherwise, we color $n$
    with a color not in $\hat{v_7}$, then color $v_1$ with a color not
    in $v_6$, forget $v_6$ and apply Lemma~\ref{lem:diam} on $\mathcal{T}(G)$
    with the path $pv_8ov_7$.

    Otherwise, $v_1v_7\in E(G)$. We then color $v_1$ with a color not
    in $\hat{v_6}$, forget $v_6$ and apply Lemma~\ref{lem:diam} on
    $\mathcal{T}(G)$ with the path $pv_8ov_7n$.
  \end{enumerate}
\end{proof}

\begin{lemma}
  \label{lem:C12bc}
The graph $G$ does not contain $C_{\ref{C12}b}$ nor $C_{\ref{C12}c}$.
\end{lemma}

\begin{proof}
  First note that $C_{\ref{C12}b}$ is a sub-configuration of
  $C_{\ref{C12}c}$. It is thus sufficient to prove that $G$ does not
  contain $C_{\ref{C12}c}$.  We use the notation depicted in
  Figure~\ref{fig:C12bc}. By minimality, we color
  $G\setminus\{a,\ldots,o\}$ and uncolor $u,v_1,\ldots,v_8$.

  \begin{figure}[!h]
    \centering
    \begin{tikzpicture}[v/.style={draw=black,minimum size = 10pt,ellipse,inner sep=1pt}]
      \node[v,label=85:{$u$}] (u) at (0,0)  {8};
      \node[v,label=right:{$v_7$}] (v0) at (45:1.5) {5};
      \node[v,label=above:{$v_6$}] (v1) at (90:1.5) {5};
      \node[v,label=left:{$v_5$}] (v2) at (135:1.5) {6};
      \node[v,label=left:{$v_4$}] (v3) at (180:1.5) {5};
      \node[v,label=left:{$v_3$}] (v4) at (225:1.5) {6};
      \node[v,label=below:{$v_2$}] (v5) at (270:1.5) {5};
      \node[v,label=right:{$v_1$}] (v6) at (315:1.5) {6};
      \node[v,label=right:{$v_8$}] (v7) at (0:1.5) {6};
      \draw (v6) -- (u) node[midway,right] {$a$};
      \draw (v5) -- (u) node[midway,right] {$b$};
      \draw (v4) -- (u) node[midway,below] {$c$};
      \draw (v3) -- (u) node[midway,below] {$d$};
      \draw (v2) -- (u) node[midway,left] {$e$};
      \draw (v1) -- (u) node[midway,left] {$f$};
      \draw (v0) -- (u) node[midway,above] {$g$};
      \draw (v7) -- (u) node[midway,above] {$h$};
      \draw (v5) -- (v6) node[midway,below] {$i$};
      \draw (v5) -- (v4) node[midway,below] {$j$};
      \draw (v4) -- (v3) node[midway,left] {$k$};
      \draw (v3) -- (v2) node[midway,left] {$\ell$};
      \draw (v2) -- (v1) node[midway,above] {$m$};
      \draw (v1) -- (v0) node[midway,above] {$n$};
      \draw (v0) -- (v7) node[midway,right] {$o$};
    \end{tikzpicture}
\caption{Notation for Lemma~\ref{lem:C12bc}}
    \label{fig:C12bc}
  \end{figure}
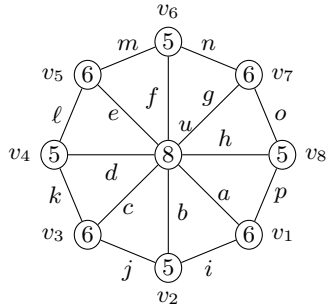

  We have $|\hat{i}|=|\hat{o}|=4$,
  $|\hat{j}|=|\hat{k}|=|\hat{\ell}|=|\hat{m}|=5$,
  $|\hat{a}|=|\hat{h}|=|\hat{n}|=6$, $|\hat{c}|=|\hat{e}|=7$,
  $|\hat{b}|=|\hat{d}|=|\hat{f}|=|\hat{g}|=8$ and
  $|\hat{u}|=10$. Moreover, $|\hat{v_1}|,|\hat{v_8}|$ are at least 2,
  $|\hat{v_3}|,|\hat{v_5}|$ are at least $4$ and
  $|\hat{v_2}|,|\hat{v_4}|,|\hat{v_6}|,|\hat{v_7}|$ are at least $6$.

  We color $c$ with a color not in $\hat{j}$ and $d$ with a color not
  in $\hat{k}$. Then, we color $v_1$, and $v_8$ such that
  $\hat{u}\neq\hat{f}$. We color $a,h,e,b,g$, then $u,f$ since
  $\hat{u}\neq\hat{f}$, and color $i,o$. We then color $v_3$ such that
  $\hat{v_4}\neq\hat{\ell}$. Then we color $v_2,j,k$ (if
  $v_2v_4\in E(G)$, when we color $v_2$, we ensure that we still have
  $\hat{v_4}\neq\hat{\ell}$). We then color $\hat{v_5}$ such that
  $\hat{v_6}\neq\hat{n}$, $v_4$ and $\ell$ (since $\hat{v_4}$ and
  $\hat{\ell}$ are different and of size at least one), then $m,v_7$
  arbitrarily, and finally $v_6$ and $n$ (since again $\hat{v_6}$ and
  $\hat{n}$ are different of size at least one).
\end{proof}

\begin{lemma}
  \label{lem:C12d}
The graph $G$ does not contain $C_{\ref{C12}d}$.
\end{lemma}

\begin{proof}
  We use the notation depicted in Figure~\ref{fig:C12d}.  By
  minimality, we color $G\setminus\{a,\ldots,o\}$ and uncolor
  $u,v_1,\ldots,v_8$.
  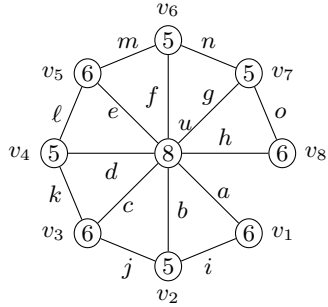
\begin{figure}[!h]
    \centering
    \begin{tikzpicture}[v/.style={draw=black,minimum size = 10pt,ellipse,inner sep=1pt}]
      \node[v,label=85:{$u$}] (u) at (0,0)  {8};
      \node[v,label=right:{$v_7$}] (v0) at (45:1.5) {5};
      \node[v,label=above:{$v_6$}] (v1) at (90:1.5) {6};
      \node[v,label=left:{$v_5$}] (v2) at (135:1.5) {5};
      \node[v,label=left:{$v_4$}] (v3) at (180:1.5) {5};
      \node[v,label=left:{$v_3$}] (v4) at (225:1.5) {6};
      \node[v,label=below:{$v_2$}] (v5) at (270:1.5) {5};
      \node[v,label=right:{$v_1$}] (v6) at (315:1.5) {6};
      \node[v,label=right:{$v_8$}] (v7) at (0:1.5) {6};
      \draw (v6) -- (u) node[midway,right] {$a$};
      \draw (v5) -- (u) node[midway,right] {$b$};
      \draw (v4) -- (u) node[midway,below] {$c$};
      \draw (v3) -- (u) node[midway,below] {$d$};
      \draw (v2) -- (u) node[midway,left] {$e$};
      \draw (v1) -- (u) node[midway,left] {$f$};
      \draw (v0) -- (u) node[midway,above] {$g$};
      \draw (v7) -- (u) node[midway,above] {$h$};
      \draw (v5) -- (v6) node[midway,below] {$i$};
      \draw (v5) -- (v4) node[midway,below] {$j$};
      \draw (v4) -- (v3) node[midway,left] {$k$};
      \draw (v3) -- (v2) node[midway,left] {$\ell$};
      \draw (v2) -- (v1) node[midway,above] {$m$};
      \draw (v1) -- (v0) node[midway,above] {$n$};
      \draw (v0) -- (v7) node[midway,right] {$o$};
    \end{tikzpicture}
\caption{Notation for Lemma~\ref{lem:C12d}}
    \label{fig:C12d}
  \end{figure}

  We have $|\hat{i}|=|\hat{o}|=4$,
  $|\hat{j}|=|\hat{k}|=|\hat{m}|=|\hat{n}|=5$,
  $|\hat{a}|=|\hat{h}|=|\hat{\ell}|=6$, $|\hat{c}|=|\hat{f}|=7$,
  $|\hat{b}|=|\hat{d}|=|\hat{e}|=|\hat{g}|=8$ and
  $|\hat{u}|=10$. Moreover, $|\hat{v_1}|,|\hat{v_8}|$ are at least 2,
  $|\hat{v_3}|,|\hat{v_6}|$ are at least $4$ and
  $|\hat{v_2}|,|\hat{v_4}|,|\hat{v_5}|,|\hat{v_7}|$ are at most $6$.

  We color $f$ with a color not in $\hat{n}$ and $b$ with a color not
  in $\hat{j}$, then we color $v_8$, and $v_1$ such that
  $\hat{u}\neq\hat{d}$. We color $a,h,c,e,g$, then $u$ and $d$ since
  $\hat{u}\neq\hat{d}$, then $i,o,v_6,v_7,n,m,v_5$. We color $v_3$
  such that $\hat{v_4}\neq\hat{\ell}$, then $v_2,j,k$ and finally
  $v_4,\ell$ since $\hat{v_4}\neq\hat{\ell}$.
\end{proof}

\subsection{Configuration $C_{\ref{C13}}$}
By definition, if $G$ contains $C_{\ref{C13}}$, then we are in one of
the following cases:
\begin{itemize}
\item[$\bullet$] $C_{\ref{C13}a}$: $v_1,\ldots,v_4$ are weak neighbors of $u$ of
  degree $4$ and $u$ has a neighbor $w$ of degree $7$.
\item[$\bullet$] $C_{\ref{C13}b}$: $v_1,\ldots,v_4$ are $(7,8)$-neighbors of $u$
  such that $v_1,v_2,v_3$ have degree $4$ and $v_4$ has degree at most
  $5$.
\item[$\bullet$] $C_{\ref{C13}c}$: $u$ has a $(7,7)$-neighbor $v_1$ of degree
  $4$, a weak neighbor $v_2$ of degree $4$ and two non-adjacent
  neighbors $v_3,v_4$ of degree $5$ such that
  $\dist_u(v_1,v_2)=\dist_u(v_1,v_3)=2$.
\end{itemize}

We dedicate a lemma to each of these configurations.

\begin{lemma}
  \label{lem:C13a}
  The graph $G$ does not contain $C_{\ref{C13}a}$.
\end{lemma}

\begin{proof}
  We use the notation depicted in Figure~\ref{fig:C13a}. By
  minimality, we color $G'=G\setminus\{a,\ldots,p\}$ and uncolor
  $u,v_1,\ldots,v_4$.
  \begin{figure}[!h]
    \centering
    \begin{tikzpicture}[v/.style={draw=black,minimum size = 10pt,ellipse,inner sep=1pt}]
      \node[v,label=85:{$u$}] (u) at (0,0)  {8};
      \node[v, very thick] (v0) at (45:1.5) {7};
      \node[v,label=above:{$v_1$}] (v1) at (90:1.5) {4};
      \node[v, very thick] (v2) at (135:1.5) {8};
      \node[v,label=left:{$v_4$}] (v3) at (180:1.5) {4};
      \node[v, very thick] (v4) at (225:1.5) {8};
      \node[v,label=below:{$v_3$}] (v5) at (270:1.5) {4};
      \node[v, very thick] (v6) at (315:1.5) {8};
      \node[v,label=right:{$v_2$}] (v7) at (0:1.5) {4};
      \draw (v6) -- (u) node[midway,right] {$d$};
      \draw (v5) -- (u) node[midway,right] {$e$};
      \draw (v4) -- (u) node[midway,below] {$f$};
      \draw (v3) -- (u) node[midway,below] {$g$};
      \draw (v2) -- (u) node[midway,left] {$h$};
      \draw (v1) -- (u) node[midway,left] {$a$};
      \draw (v0) -- (u) node[midway,above] {$b$};
      \draw (v7) -- (u) node[midway,above] {$c$};
      \draw (v5) -- (v6) node[midway,below] {$\ell$};
      \draw (v5) -- (v4) node[midway,below] {$m$};
      \draw (v4) -- (v3) node[midway,left] {$n$};
      \draw (v3) -- (v2) node[midway,left] {$o$};
      \draw (v2) -- (v1) node[midway,above] {$p$};
      \draw (v1) -- (v0) node[midway,above] {$i$};
      \draw (v0) -- (v7) node[midway,right] {$j$};
      \draw (v6) -- (v7) node[midway,right] {$k$};
    \end{tikzpicture}
\caption{Notation for Lemma~\ref{lem:C13a}}
    \label{fig:C13a}
  \end{figure}

  We have
  $|\hat{k}|=|\hat{\ell}|=|\hat{m}|=|\hat{n}|=|\hat{o}|=|\hat{p}|=3$,
  $|\hat{d}|=|\hat{f}|=|\hat{h}|=|\hat{i}|=|\hat{j}|=4$,
  $|\hat{b}|=5$,
  $|\hat{v_1}|=|\hat{v_2}|=|\hat{v_3}|=|\hat{v_4}|=|\hat{u}|=6$, and
  $|\hat{a}|=|\hat{c}|=|\hat{e}|=|\hat{g}|=9$.

  We forget $v_1,v_2,v_3,v_4$ and conclude using
  Theorem~\ref{thm:nss}.
\end{proof}

\begin{lemma}
  \label{lem:C13b}
  The graph $G$ does not contain $C_{\ref{C13}b}$.
\end{lemma}

\begin{proof}
  We use the notation depicted in Figure~\ref{fig:C13b}. By minimality,
  we color $G\setminus\{a,\ldots,p\}$ and uncolor
  $u,v_1,\ldots,v_4,w_1,w_2$.
  \begin{figure}[!h]
    \centering
    \begin{tikzpicture}[v/.style={draw=black,minimum size = 10pt,ellipse,inner sep=1pt}]
      \node[v,label=85:{$u$}] (u) at (0,0)  {8};
      \node[v,label=right:{$w_1$}] (v0) at (45:1.5) {7};
      \node[v,label=above:{$v_1$}] (v1) at (90:1.5) {4};
      \node[v, very thick] (v2) at (135:1.5) {8};
      \node[v,label=left:{$v_4$}] (v3) at (180:1.5) {5};
      \node[v,label=left:{$w_2$}] (v4) at (225:1.5) {7};
      \node[v,label=below:{$v_3$}] (v5) at (270:1.5) {4};
      \node[v, very thick] (v6) at (315:1.5) {8};
      \node[v,label=right:{$v_2$}] (v7) at (0:1.5) {4};
      \draw (v6) -- (u) node[midway,right] {$d$};
      \draw (v5) -- (u) node[midway,right] {$e$};
      \draw (v4) -- (u) node[midway,below] {$f$};
      \draw (v3) -- (u) node[midway,below] {$g$};
      \draw (v2) -- (u) node[midway,left] {$h$};
      \draw (v1) -- (u) node[midway,left] {$a$};
      \draw (v0) -- (u) node[midway,above] {$b$};
      \draw (v7) -- (u) node[midway,above] {$c$};
      \draw (v5) -- (v6) node[midway,below] {$\ell$};
      \draw (v5) -- (v4) node[midway,below] {$m$};
      \draw (v4) -- (v3) node[midway,left] {$n$};
      \draw (v3) -- (v2) node[midway,left] {$o$};
      \draw (v2) -- (v1) node[midway,above] {$p$};
      \draw (v1) -- (v0) node[midway,above] {$i$};
      \draw (v0) -- (v7) node[midway,right] {$j$};
      \draw (v6) -- (v7) node[midway,right] {$k$};
    \end{tikzpicture}
\caption{Notation for Lemma~\ref{lem:C13b}}
    \label{fig:C13b}
  \end{figure}

  We have $|\hat{o}|=2$, $|\hat{k}|=|\hat{\ell}|=|\hat{p}|=3$,
  $|\hat{d}|=|\hat{h}|=|\hat{n}|=4$,
  $|\hat{i}|=|\hat{j}|=|\hat{m}|=5$,
  $|\hat{b}|=|\hat{f}|=6$,
  $|\hat{u}|=|\hat{g}|=8$ and
  $|\hat{a}|=|\hat{c}|=|\hat{e}|=9$. Moreover, we may also assume that
  $|\hat{w_1}|,|\hat{w_2}|$ are at least $2$, $|\hat{v_4}|$ is at least $5$ and $|\hat{v_1}|=|\hat{v_2}|=|\hat{v_3}|$ are at least $7$ (depending on whether $w_1w_2$ is an edge).

  We forget $v_1,v_2,v_3$, color $h$ with a color not in $\hat{p}$,
  then color $o$. We remove from $\hat{w_2},\hat{f}$ and $\hat{n}$ a
  color $\alpha\in\hat{m}\setminus\hat{\ell}$. Then, we color
  $w_2,n,f,d$. We color $u,v_4,w_1,b,g$ applying
  Theorem~\ref{thm:=deg} on the subgraph of $\mathcal{T}(G)$ they
  induce. Due to the choice of $\alpha$, we have
  $\hat{m}\neq\hat{\ell}$ if $|\hat{m}|=|\hat{\ell}|=2$ thus we can
  color $\ell$ with a color not in $\hat{m}$ and forget $m$. We then
  color $k$, then $p$ such that $\hat{a}\neq\hat{e}$ and apply
  Lemma~\ref{lem:fryingpan} on $\mathcal{T}(G)$ with the cycle $aijc$
  and the element $e$.
\end{proof}

\begin{lemma}
  \label{lem:C13c}
  The graph $G$ does not contain $C_{\ref{C13}c}$.
\end{lemma}

\begin{proof}
  We use the notation depicted in Figure~\ref{fig:C13c}. By
  minimality, we color $G'=G\setminus\{a,\ldots,\ell\}$ and uncolor
  $u,v_1,v_2,v_4,w_1,w_2$.
  \begin{figure}[!h]
    \centering
    \begin{tikzpicture}[v/.style={draw=black,minimum size = 10pt,ellipse,inner sep=1pt}]
      \node[v,label=85:{$u$}] (u) at (0,0)  {8};
      \node[v,label=right:{$w_2$}] (v0) at (45:1.5) {7};
      \node[v,label=above:{$v_1$}] (v1) at (90:1.5) {4};
      \node[v,label=left:{$w_1$}] (v2) at (135:1.5) {7};
      \node[v,label=left:{$v_4$}] (v3) at (180:1.5) {5};
      \node[v, very thick, label=right:{$v_3$}] (v5) at (270:1.5) {5};
      \node[v, very thick] (v6) at (315:1.5) {8};
      \node[v,label=right:{$v_2$}] (v7) at (0:1.5) {4};
      \draw (v6) -- (u) node[midway,right] {$f$};
      \draw (v5) -- (u) node[midway,right] {$g$};
      \draw (v3) -- (u) node[midway,below] {$a$};
      \draw (v2) -- (u) node[midway,left] {$b$};
      \draw (v1) -- (u) node[midway,left] {$c$};
      \draw (v0) -- (u) node[midway,above] {$d$};
      \draw (v7) -- (u) node[midway,above] {$e$};
      \draw (v3) -- (v2) node[midway,left] {$h$};
      \draw (v2) -- (v1) node[midway,above] {$i$};
      \draw (v1) -- (v0) node[midway,above] {$j$};
      \draw (v0) -- (v7) node[midway,right] {$k$};
      \draw (v6) -- (v7) node[midway,right] {$\ell$};
    \end{tikzpicture}
\caption{Notation for Lemma~\ref{lem:C13c}}
    \label{fig:C13c}
  \end{figure}

  We have $|\hat{f}|=|\hat{\ell}|=2$, $|\hat{h}|=3$,
  $|\hat{v_4}|=|\hat{g}|=4$,
  $|\hat{b}|=|\hat{d}|=|\hat{i}|=|\hat{j}|=|\hat{k}|=5$,
  $|\hat{a}|=|\hat{u}|=6$, $|\hat{v_2}|=7$, and
  $|\hat{v_1}|=|\hat{c}|=|\hat{e}|=8$. Moreover,
  $|\hat{w_1}|,|\hat{w_2}|$ depend on whether
  $w_1w_2,w_2v_4\in E(G)$. In each case, we forget $v_1, v_2$ and we conclude using Theorem~\ref{thm:nss}.
\end{proof}

\subsection{Configuration $C_{\ref{C15}}$}
Due to $C_{\ref{C3a}}$ and to the definition of $C_{\ref{C15}}$, if
$G$ contains $C_{\ref{C15}}$ then $G$ contains a subconfiguration of
one of the three following cases, as shown in Figure~\ref{fig:C15init}:
\begin{itemize}
\item[$\bullet$] $C_{\ref{C15}a}$: $u$ has a weak neighbor of degree 3 and two $(6,6)$-neighbors of degree $5$.
\item[$\bullet$] $C_{\ref{C15}b}$: $u$ has a weak neighbor of degree 3 and three weak neighbors of degree $5$
  and two neighbors of degree $6$, such that there is a triangular
  face containing $u$ and two vertices of degree $5$.
\item[$\bullet$] $C_{\ref{C15}c}$: $u$ has a weak neighbor of degree 3 and three weak neighbors of degree $5$
  and two neighbors of degree $6$, such that there is no triangular
  face containing $u$ and two vertices of degree $5$.
\end{itemize}
We dedicate a lemma to each of these configurations. In each of them,
we did not succeed in finding a suitable monomial for the
Nullstellensatz approach, hence we only present case analysis proofs.

\begin{lemma}
  \label{lem:C15a}
The graph $G$ does not contain $C_{\ref{C15}a}$.
\end{lemma}

\begin{proof}
  We use the notation depicted in Figure~\ref{fig:C15a}. By minimality,
  we color $G\setminus\{a,\ldots,n\}$ and uncolor
  $u,v_1,v_2,v_3,w_2$.
  \begin{figure}[!h]
    \centering
    \begin{tikzpicture}[v/.style={draw=black,minimum size = 10pt,ellipse,inner sep=1pt}]
      \node[v,label=85:{$u$}] (u) at (0,0)  {8};
      \node[v, very thick] (v0) at (45:1.5) {8};
      \node[v,label=above:{$v_1$}] (v1) at (90:1.5) {3};
      \node[v, very thick] (v2) at (135:1.5) {8};
      \node[v,label=left:{$w_3$}, very thick] (v3) at (180:1.5) {6};
      \node[v,label=left:{$v_3$}] (v4) at (225:1.5) {5};
      \node[v,label=below:{$w_2$}] (v5) at (270:1.5) {6};
      \node[v,label=right:{$v_2$}] (v6) at (315:1.5) {5};
      \node[v,label=right:{$w_1$}, very thick] (v7) at (0:1.5) {6};
      \draw (v6) -- (u) node[midway,right] {$e$};
      \draw (v5) -- (u) node[midway,right] {$f$};
      \draw (v4) -- (u) node[midway,below] {$g$};
      \draw (v3) -- (u) node[midway,below] {$h$};
      \draw (v2) -- (u) node[midway,left] {$a$};
      \draw (v1) -- (u) node[midway,left] {$b$};
      \draw (v0) -- (u) node[midway,above] {$c$};
      \draw (v7) -- (u) node[midway,above] {$d$};
      \draw (v5) -- (v6) node[midway,below] {$\ell$};
      \draw (v5) -- (v4) node[midway,below] {$m$};
      \draw (v4) -- (v3) node[midway,left] {$n$};
      \draw (v2) -- (v1) node[midway,above] {$i$};
      \draw (v1) -- (v0) node[midway,above] {$j$};
      \draw (v6) -- (v7) node[midway,right] {$k$};
    \end{tikzpicture}
\caption{Notation for Lemma~\ref{lem:C15a}}
    \label{fig:C15a}
  \end{figure}

  We have
  $|\hat{a}|=|\hat{c}|=|\hat{i}|=|\hat{j}|=|\hat{k}|=|\hat{n}|=3$,
  $|\hat{w_2}|=4$,
  $|\hat{v_2}|=|\hat{v_3}|=|\hat{d}|=|\hat{h}|=|\hat{\ell}|=|\hat{m}|=5$,
  $|\hat{u}|=6$, $|\hat{f}|=7$, $|\hat{v_1}|=|\hat{e}|=|\hat{g}|=8$
  and $|\hat{b}|=10$. We forget $v_1$ and consider two cases:

  \begin{enumerate}
  \item Assume that $\hat{g}\neq\hat{v_3}\cup\hat{n}$, and color $g$
    with a color not in $\hat{v_3}\cup\hat{n}$. Then forget
    $v_3,n,m$. We then color $a,c$ such that afterwards we have
    $\hat{i}\neq\hat{j}$ if $|\hat{i}|=|\hat{j}|=2$. We can thus
    forget $i$ and $j$ (since after coloring every other element,
    either one of them has $2$ choices, or both have one but not the
    same), then $b$.
    \begin{enumerate}
    \item If $\hat{d}=\hat{h}$, we color $u,f,e$ with colors not in
      $\hat{d}$, forget $h,d$ and apply Lemma~\ref{lem:diam} on $\mathcal{T}(G)$
      with the path $kv_2\ell w_2$.
    \item Otherwise, if $|\hat{u}\cup\hat{d}\cup\hat{h}|=3$, we color
      $f,e$ with a color not in this union, then color $d$ with a
      color not in $\hat{h}$, forget $h$, color $k$, and apply
      Lemma~\ref{lem:diam} on $\mathcal{T}(G)$ with the path $uv_2w_2\ell$.
    \item Otherwise, if
      $|\hat{u}\cup\hat{d}\cup\hat{f}\cup\hat{h}|=4$, we color $e$
      with a color not in this union, then $d$ with a color not in
      $\hat{h}$ and color $k$. If $\hat{h}=\hat{u}$, we color $f$ with a color not
      in $\hat{u}$, forget $h$ and apply Lemma~\ref{lem:diam} on
      $\mathcal{T}(G)$ with the path $uv_2w_2\ell$. Otherwise, we color $u$ with
      a color not in $\hat{h}$, forget $h$ and apply
      Lemma~\ref{lem:diam} on $\mathcal{T}(G)$ with the path $v_2\ell w_2f$.
    \item Otherwise, we color $e$ with a color not in $\hat{k}$ and
      color $\{u,d,f,h\}$ using Theorem~\ref{thm:clique}. Then we
      apply Lemma~\ref{lem:diam} on $\mathcal{T}(G)$ with the path
      $kv_2\ell w_2$.
    \end{enumerate}
  \item Otherwise, we can assume by symmetry that
    $\hat{v_3}\cap\hat{n}=\varnothing=\hat{v_2}\cap\hat{k}$. Then we
    can forget $v_2,v_3$, color $g$ with a color not in $\hat{m}$ and
    color $a,c$ such that afterwards we have $\hat{i}\neq\hat{j}$ if
    $|\hat{i}|=|\hat{j}|=2$. Then, we again forget $i$ and $j$ and we
    color $h,d,u,f,e,k,\ell,w_2,n,m,b$.

\end{enumerate}
\end{proof}

\begin{lemma}
  \label{lem:C15b}
The graph $G$ does not contain $C_{\ref{C15}b}$.
\end{lemma}

\begin{proof}
  We use the notation depicted in Figure~\ref{fig:C15b}. By minimality,
  we color $G\setminus\{a,\ldots,o\}$ and uncolor
  $u,v_1,\ldots,v_4,w_1,w_2$.
  \begin{figure}[!h]
    \centering
    \begin{tikzpicture}[v/.style={draw=black,minimum size = 10pt,ellipse,inner sep=1pt}]
      \node[v,label=85:{$u$}] (u) at (0,0)  {8};
      \node[v, very thick] (v0) at (45:1.5) {8};
      \node[v,label=above:{$v_1$}] (v1) at (90:1.5) {3};
      \node[v, very thick] (v2) at (135:1.5) {8};
      \node[v,label=left:{$v_4$}] (v3) at (180:1.5) {5};
      \node[v,label=left:{$w_2$}] (v4) at (225:1.5) {6};
      \node[v,label=below:{$v_3$}] (v5) at (270:1.5) {5};
      \node[v,label=right:{$v_2$}] (v6) at (315:1.5) {5};
      \node[v,label=right:{$w_1$}] (v7) at (0:1.5) {6};
      \draw (v6) -- (u) node[midway,right] {$d$};
      \draw (v5) -- (u) node[midway,right] {$e$};
      \draw (v4) -- (u) node[midway,below] {$f$};
      \draw (v3) -- (u) node[midway,below] {$g$};
      \draw (v2) -- (u) node[midway,left] {$h$};
      \draw (v1) -- (u) node[midway,left] {$a$};
      \draw (v0) -- (u) node[midway,above] {$b$};
      \draw (v7) -- (u) node[midway,above] {$c$};
      \draw (v1) -- (v0) node[midway,above] {$i$};
      \draw (v6) -- (v7) node[midway,right] {$j$};
      \draw (v5) -- (v6) node[midway,below] {$k$};
      \draw (v5) -- (v4) node[midway,below] {$\ell$};
      \draw (v4) -- (v3) node[midway,left] {$m$};
      \draw (v3) -- (v2) node[midway,left] {$n$};
      \draw (v2) -- (v1) node[midway,above] {$o$};
    \end{tikzpicture}
\caption{Notation for Lemma~\ref{lem:C15b}}
    \label{fig:C15b}
  \end{figure}

  We have $|\hat{n}|=2$, $|\hat{b}|=|\hat{i}|=3$,
  $|\hat{h}|=|\hat{j}|=|\hat{o}|=4$, $|\hat{\ell}|=|\hat{m}|=5$,
  $|\hat{c}|=|\hat{k}|=6$, $|\hat{f}|=7$,
  $|\hat{u}|=|\hat{v_1}|=|\hat{d}|=|\hat{e}|=|\hat{g}|=8$ and
  $|\hat{a}|=10$. We may moreover assume that
  $|\hat{w_1}|\geqslant 2$, $|\hat{w_2}|\geqslant 4$,
  $|\hat{v_4}|\geqslant 5$ and $|\hat{v_2}|,|\hat{v_3}|\geqslant 6$.

  We forget $v_1$, then we remove from $\hat{h}$ and $\hat{n}$ a color
  $\alpha\in\hat{o}\setminus\hat{i}$. We then color $n$. Due to the
  choice of $\alpha$, we may forget $i,o$ since any coloring of the
  other elements gives either $|\hat{o}|>1$ or $\hat{o}\neq\hat{i}$,
  hence we can always color $i$ then $o$. We may also forget $a$.

  Note that $v_4$ has degree $5$ hence it is adjacent (in $G$) to at
  most four uncolored vertices, hence we may assume that
  $|\hat{v_4}|<7$. We color $g$ with a color not in $\hat{v_4}$, then
  $h,b$. We then color $f$ with a color not in $\hat{m}$, then
  $w_1,c,u,j,d,e$. We forget $v_4,m$ and color $v_3,v_2,w_2,k,\ell$
  using Theorem~\ref{thm:=deg} on the subgraph of $\mathcal{T}(G)$ they induce.
\end{proof}

\begin{lemma}
  \label{lem:C15c}
The graph $G$ does not contain $C_{\ref{C15}c}$.
\end{lemma}

\begin{proof}
  We use the notation depicted in Figure~\ref{fig:C15c}. By minimality,
  we color $G\setminus\{a,\ldots,p\}$ and uncolor
  $u,v_1,\ldots,v_4,w_1,w_2$.
  \begin{figure}[!h]
    \centering
     \begin{tikzpicture}[v/.style={draw=black,minimum size = 10pt,ellipse,inner sep=1pt}]
      \node[v,label=85:{$u$}] (u) at (0,0)  {8};
      \node[v, very thick] (v0) at (45:1.5) {8};
      \node[v,label=above:{$v_1$}] (v1) at (90:1.5) {3};
      \node[v, very thick] (v2) at (135:1.5) {8};
      \node[v,label=left:{$v_4$}] (v3) at (180:1.5) {5};
      \node[v,label=left:{$w_2$}] (v4) at (225:1.5) {6};
      \node[v,label=below:{$v_3$}] (v5) at (270:1.5) {5};
      \node[v,label=right:{$w_1$}] (v6) at (315:1.5) {6};
      \node[v,label=right:{$v_2$}] (v7) at (0:1.5) {5};
      \draw (v6) -- (u) node[midway,right] {$d$};
      \draw (v5) -- (u) node[midway,right] {$e$};
      \draw (v4) -- (u) node[midway,below] {$f$};
      \draw (v3) -- (u) node[midway,below] {$g$};
      \draw (v2) -- (u) node[midway,left] {$h$};
      \draw (v1) -- (u) node[midway,left] {$a$};
      \draw (v0) -- (u) node[midway,above] {$b$};
      \draw (v7) -- (u) node[midway,above] {$c$};
      \draw (v5) -- (v6) node[midway,below] {$\ell$};
      \draw (v5) -- (v4) node[midway,below] {$m$};
      \draw (v4) -- (v3) node[midway,left] {$n$};
      \draw (v3) -- (v2) node[midway,left] {$o$};
      \draw (v2) -- (v1) node[midway,above] {$p$};
      \draw (v1) -- (v0) node[midway,above] {$i$};
      \draw (v0) -- (v7) node[midway,right] {$j$};
      \draw (v6) -- (v7) node[midway,right] {$k$};
    \end{tikzpicture}
\caption{Notation for Lemma~\ref{lem:C15c}}
    \label{fig:C15c}
  \end{figure}

  We have $|\hat{j}|=|\hat{o}|=2$,
  $|\hat{b}|=|\hat{h}|=|\hat{i}|=|\hat{p}|=4$,
  $|\hat{k}|=|\hat{\ell}|=|\hat{m}|=|\hat{n}|=5$,
  $|\hat{d}|=|\hat{f}|=7$,
  $|\hat{v_1}|=|\hat{c}|=|\hat{e}|=|\hat{g}|=|\hat{u}|=8$ and
  $|\hat{a}|=10$. We may also assume that $|\hat{v_2}|,|\hat{v_4}|$
  are at least $5$, $|\hat{w_1}|,|\hat{w_2}|$ are at least $4$ and
  $|\hat{v_3}|$ is at least $6$.

  We forget $v_1$, color $j$ and remove from $\hat{h}$ and $\hat{o}$ a
  color $\alpha\in\hat{p}\setminus\hat{i}$. Then we may forget
  $i,p$. Indeed, if we can color every element except $i,p$, then
  due to the choice of $\alpha$, if $|\hat{p}|=|\hat{i}|=1$,
  $\hat{p}\neq\hat{i}$, hence we can color them. We may then forget
  $a$. We also color $o$. 

  Note that there are only six uncolored vertices, hence $w_1$ has at
  most $5$ uncolored neighbors in $G$. We thus have
  $|\hat{w_1}|\leqslant 6$, hence we can color $d$ with a color not in
  $\hat{w_1}$. We then color $h,b$ arbitrarily, and $c$ with a color
  not in $\hat{k}$. We color $f,g$, then $u,v_4,w_2,n$ applying
  Theorem~\ref{thm:=deg} on the subgraph of $\mathcal{T}(G)$ they induce, and
  then color $e,m$. We finally apply Lemma~\ref{lem:diam} on $\mathcal{T}(G)$
  with the path $v_2kw_1\ell v_3$.
\end{proof}

\subsection{Configuration $C_{\ref{C15.5a}}$}
To prove that $G$ does not contain the configuration $C_{\ref{C15.5a}}$,
we prove that it does not contain any of the configuration below.

\begin{itemize}
\item[$\bullet$] $C_{\ref{C15.5a}a}$ is a $8$-vertex $u$ with a weak
  neighbor $v$ of degree $3$, a $(7,8)$-neighbor of degree 4 at
  triangle distance 4 from $v$, and two weak neighbors of degree 5.
\item[$\bullet$] $C_{\ref{C15.5a}b}$ is a $8$-vertex $u$ with a weak
  neighbor $v$ of degree $3$ and a weak neighbor of degree 5 at
  triangle distance 2 from $v$, having two neighbors of degree 6.
\item[$\bullet$] $C_{\ref{C15.5a}c}$ is a $8$-vertex $u$ with a weak
  neighbor $v$ of degree $3$, a weak neighbor of degree 4 at triangle
  distance 4 from $v$, and two weak neighbors of degree 5, one of them
  having a neighbor of degree 5.
\end{itemize}

We dedicate a lemma to each of these configurations.

\begin{lemma}
  \label{lem:C20ac}
The graph $G$ does not contain $C_{\ref{C15.5a}a}$.
\end{lemma}

\begin{proof}
  We use the notation depicted in Figure~\ref{fig:C20ac}.  By
  minimality, we take a coloring $\gamma$ of $G\setminus\{p\}$.

  \begin{figure}[!h]
    \centering
    \begin{tikzpicture}[v/.style={draw=black,minimum size = 10pt,ellipse,inner sep=1pt}]
      \node[v,label=85:{$u$}] (u) at (0,0)  {8};
      \node[v, very thick] (v0) at (45:1.5) {8};
      \node[v,label=above:{$v_1$}] (v1) at (90:1.5) {3};
      \node[v, very thick] (v2) at (135:1.5) {8};
      \node[v,label=left:{$v_4$}] (v3) at (180:1.5) {5};
      \node[v, very thick] (v4) at (225:1.5) {7};
      \node[v,label=below left:{$v_3$}] (v5) at (270:1.5) {4};
      \node[v, very thick] (v6) at (315:1.5) {8};
      \node[v,label=right:{$v_2$}] (v7) at (0:1.5) {5};
      \node[v, very thick] (w) at (270:3) {8};
      \draw (v6) -- (u) node[midway,right] {$d$};
      \draw (v5) -- (u) node[midway,right] {$e$};
      \draw (v4) -- (u) node[midway,below] {$f$};
      \draw (v3) -- (u) node[midway,below] {$g$};
      \draw (v2) -- (u) node[midway,left] {$h$};
      \draw (v1) -- (u) node[midway,left] {$a$};
      \draw (v0) -- (u) node[midway,above] {$b$};
      \draw (v7) -- (u) node[midway,above] {$c$};
      \draw (v5) -- (v6) node[midway,below] {$\ell$};
      \draw (v5) -- (v4) node[midway,below] {$m$};
      \draw (v4) -- (v3) node[midway,left] {$n$};
      \draw (v3) -- (v2) node[midway,left] {$o$};
      \draw (v2) -- (v1) node[midway,above] {$p$};
      \draw (v1) -- (v0) node[midway,above] {$i$};
      \draw (v0) -- (v7) node[midway,right] {$j$};
      \draw (v6) -- (v7) node[midway,right] {$k$};
      \draw (v5) -- (w) node[midway,right] {$q$};
    \end{tikzpicture}
\caption{Notation for Lemma~\ref{lem:C20ac}}
    \label{fig:C20ac}
  \end{figure}

Assume first that both endpoints of $q$ are neighbors of $u$. That means $q$ is incident to $i$ or $p$. The situation will be symmetric, hence assume that $q$ is incident to $i$. We uncolor $a,e,i,\ell,m,p,q,v_1,v_3$ and forget $v_1,v_3$. Now we have $|\hat{\ell}|=|\hat{p}|=2$ and $|\hat{a}|=|\hat{e}|=|\hat{i}|=|\hat{m}|=|\hat{q}|=3$. 
\begin{enumerate}
\item If $\hat{e}\neq\hat{q}$, we color $m$ with a color not in $\ell$, and color $p$ arbitrarily. We then conclude using Lemma~\ref{lem:fryingpan} on $\{\ell,e,a,i,q\}$.
\item Otherwise, if $\hat{m}\neq\hat{e}$, we color $m$ with a color outside of $\hat{e}$ (hence of $\hat{q}$), then $\ell$, and apply Lemma~\ref{lem:fryingpan} on $\{p,i,q,a,e\}$.
\item We may thus assume that $\hat{e}=\hat{m}=\hat{q}=\{\gamma(e),\gamma(m),\gamma(q)\}$. But in that case, we may color $\ell$ with $\gamma(\ell)$, and this does not change $\hat{e},\hat{m},\hat{q}$. And we can color $p,a,i,e,q,m$ in order.
\end{enumerate}
  
  We may thus assume that the only endpoint of $q$ in the neighborhood of $u$ is $v_3$. In that case, we uncolor $a,\ldots,q,v_1,v_2,v_3,v_4$ and forget $v_1,v_3$. We have $|\hat{j}|=|\hat{k}|=|\hat{o}|=|\hat{q}|=2$, $|\hat{n}|=3$,
  $|\hat{b}|=|\hat{d}|=|\hat{h}|=|\hat{i}|=|\hat{\ell}|=|\hat{p}|=4$,
  $|\hat{f}|=|\hat{m}|=5$, $|\hat{u}|=6$, $|\hat{c}|=|\hat{g}|=8$ and
  $|\hat{a}|=|\hat{e}|=10$ and moreover, $\hat{v_2}$ and $\hat{v_4}$
  have size 4 or 5 depending on whether $v_2v_4\in E(G)$. 

  \begin{enumerate} \item Assume that $\hat{h}\cap\hat{n}\neq\varnothing$. Then we color
    $h$ and $n$ with the same color, then color $o,p,q$.  We then color
    the remaining graph using Theorem~\ref{thm:nss}. Therefore, we may
    assume that $\hat{h}$ and $\hat{n}$ are disjoint.
    \item Assume that $\hat{b}\cap\hat{k}\neq\varnothing$. Then we color
    $b$ and $k$ with the same color, then color $j,q$.  We then color
    the remaining graph using Theorem~\ref{thm:nss}. Therefore, we may
    assume that $\hat{b}$ and $\hat{k}$ are disjoint.
    
\item Assume that $\hat{d}\cap \hat{j}\neq\varnothing$. We color $d$
    and $j$ with the same color, then $k$ and $\ell,q$ arbitrarily. We
    then color the remaining graph using
    Theorem~\ref{thm:nss}. Therefore, we may assume that $\hat{j}$ and
    $\hat{d}$ are disjoint.
  \item Assume that $\hat{d}\cap \hat{v_2}\neq\varnothing$. We color
    $d$ and $v_2$ with the same color (which hence does not lie in
    $\hat{j}$). Then we color $k,j,\ell,q$. We then color the
    remaining graph using Theorem~\ref{thm:nss}, hence we may assume
    that $\hat{d}$ and $\hat{v_2}$ are disjoint.
  \item Assume that $\hat{b}\cap \hat{v_2}\neq\varnothing$. We color $b$
    and $v_2$ with the same color (which hence does not lie in
    $\hat{k}$ nor in $\hat{d}$). Then we color $j,k$.

    \begin{itemize}
    \item If $\hat{o}\not\subset\hat{p}$, we color $o$ with a color
      not in $\hat{p}$, then forget $p,i,a$. If
      $\hat{i}\not\subset\hat{p}$, we color $i$ with a color not in
      $\hat{p}$, then $o$, and we forget $p,a$. Finally, if
      $\hat{o}\cap\hat{i}\neq\varnothing$, we color $o$ and $i$ with
      the same color, then forget $p,a$. In the three cases, we end up
      with the same configuration, that we reduce using
      Theorem~\ref{thm:nss}. Therefore, we may assume that $\hat{o}$
      and $\hat{i}$ are disjoint, and that their union is $\hat{p}$.
    \item Since $|\hat{h}|=3,|\hat{i}|=2$ and
      $|\hat{p}|=4$, we have either $\hat{h}\not \subset\hat{p}$ or $\hat{h}\cap\hat{i}\neq\varnothing$. In the former, we color $h$ with a color
      not in $\hat{p}$ (hence not in $\hat{o}$), then forget
      $p,i,a$. In the latter, we color $h$ and $i$ with the same color (hence not in $\hat{o}$), then
      forget $p,a$.

      In both cases, we end up with the same configuration, that we
      reduce using Theorem~\ref{thm:nss}.
    \end{itemize}
    Therefore, we may assume that $\hat{b}$ and $\hat{v_2}$ are
    disjoint.
    
  \item If $\hat{j}\not\subset \hat{v_2}$
    (resp. $\hat{k}\not\subset\hat{v_2}$), we color $j$ (resp. $k$) with
    a color not in $\hat{v_2}$. We forget $v_2$, color $k$ (resp. $j$),
    then $q,\ell,d$ arbitrarily. We then end up with the same
    configuration as in 4., which is reducible. Therefore, we may
    assume that $\hat{v_2}$ contains $\hat{j}$ and $\hat{k}$.
  \item If $\hat{b}=\hat{i}$, then $\hat{j}\cap \hat{i}=\varnothing$, hence we can forget $i,p,a$, and color all the remaining elements with their color in $\gamma$. Therefore, we may assume that $\hat{b}\neq\hat{i}$.
  \item Assume that $\hat{j}=\hat{k}$. In particular $\hat{j}=\{\gamma(j),\gamma(k)\}$, and $\hat{j}$ cannot contain $\gamma(v_2)$ nor $\gamma(c)$. 
  
  Observe that $\gamma(b)\in\hat{b}$. Since $\hat{b}$ is disjoint from $\hat{v_2}$, which contains $\hat{j}$, we have $\gamma(b)\notin\hat{j}$. Similarly, $\gamma(d)\notin\hat{k}$. We now color $b,c,d,f,g,h,n,o,u,v_2,v_4$ with their color in $\gamma$, and afterwards, $\hat{j}$ and $\hat{k}$ remain unchanged.

  Denote by $\tilde{x}$ the new list of available colors for the element $x$. Without loss of generality, we may assume that $|\tilde{j}|=|\tilde{k}|=|\tilde{p}|=|\tilde{q}|=2$ and $|\tilde{a}|=|\tilde{e}|=|\tilde{\ell}|=|\tilde{m}|=3$. 
    \begin{itemize} 
    \item If $\tilde{e}\neq\tilde{\ell}$, we color $m$ with a color not in $\hat{q}$, and color $p$ arbitrarily. We then conclude using Lemma~\ref{lem:fryingpan} on $\{q,e,a,i,j,k,\ell\}$.
    \item If $\tilde{m}\neq\tilde{e}$, we color $m$ with a color not in $\tilde{e}=\tilde{\ell}$, then $p,q$. We then color $\{a,i,j,k,\ell,e\}$ using Corollary~\ref{cor:evencycle}.
    \item Therefore, $\tilde{e}=\tilde{\ell}=\tilde{m}=\{\gamma(e),\gamma(\ell),\gamma(m)\}$ are the same lists of size $3$. We may thus color $q$ with $\gamma(q)$ without changing these lists, and color arbitrarily $e,a,p,i,j,k,\ell,m$ in order.
    \end{itemize}
Therefore, we may assume that $\hat{j}\neq\hat{k}$.
    \item We now color $b$ with a color not in $\hat{i}$, then forget $i,p,a$. 
\begin{itemize}
    \item If $\hat{q}\cap\hat{h}\neq\varnothing$, color $q$ and $h$ with the same color, color $o$ and forget $e,m,\ell$. Denote by $(\star)$ the current configuration.
    \begin{itemize} 
    \item If $\hat{f}$ and $\hat{v_4}$ share a color, we color them with it, then color $n$ and $d$ (recall that this does not affect $\hat{k}$). Denote the current configuration $(\star\star)$. If $\hat{g}=\hat{u}$, we color $c$ with a color outside of $\hat{g}$, then $j,k$ (recall that $\hat{j}\neq\hat{k}$), then $v_2,u,g$. Otherwise, we color $g,c$ with colors not in $\hat{u}$, then $j,k,v_2,u$. 

     Therefore, we may assume that $\hat{f}$ is disjoint from $\hat{v_4}$.
     \item If $\hat{n}\not\subset\hat{v_4}$, we color $n$ with a color not in $\hat{v_4}$, forget $v_4$ and color $d,f$ arbitrarily. We get again configuration $(\star\star)$.

 Therefore, we may assume that $\hat{n}\subset\hat{v_4}$. In particular $\hat{n}\cap\hat{f}=\varnothing$.
 \item  We color $v_4$ with a color outside of $\hat{n}$, then $d$ and forget $n$. If $\hat{f}=\hat{u}$ afterwards, we color $g$ then $c$ with colors outside of $\hat{f}$, then $j,k$ (recall that $\hat{j}\neq\hat{k}$), then $v_2,u$ and $f$. Otherwise, $\hat{f}\neq\hat{u}$, and coloring $f$ with a color not in $u$ yields again the configuration $(\star\star)$.

    Therefore, we may assume that $\hat{q}$ and $\hat{h}$ are disjoint.
    \end{itemize}
    \item If $\hat{q}\not\subset\hat{m}$, we color $q$ with a color not in $\hat{m}$, then forget $m,e,\ell$, and color $o,h$ arbitrarily. We now obtain the configuration $(\star)$ from the previous item. Therefore we can assume that $\hat{q}\subset \hat{m}$. 
    
    \item If $\hat{q}\cap\hat{f}\neq\varnothing$, color $q$ and $f$ with the same color, and forget $m,e,\ell$. Afterwards, if $\hat{n}\neq\hat{o}$, we color $n$ with a color not in $\hat{o}$ and obtain the configuration $(\star)$.

    Therefore, we have $\hat{n}=\hat{o}$. We remove the colors of $\hat{n}$ from $\hat{v_2}$ and $\hat{g}$, so we can forget $n$ and $o$. We then color $d$. Denote by $(\star\star\star)$ the current configuration.
    \begin{itemize}
    \item If $\hat{h}\not\subset\hat{g}$, we color $h$ with a color not in $\hat{g}$. Now either $\hat{v_4}\neq\hat{u}$ and coloring $v_4$ with a color not in $\hat{u}$ yields the configuration $(\star\star)$, or $\hat{v_4}=\hat{u}$, and we color $g$ and $c$ with colors outside of $\hat{u}$, then $j,k,v_2,u,v_4$. 

    Therefore, we have $\hat{h}\subset\hat{g}$.
    \item If $\hat{h}\cap\hat{v_4}\neq\varnothing$, we color $h$ and $v_4$ with the same color and obtain again configuration $(\star\star)$.
    \item Therefore, we can color $v_4$ with a color not in $\hat{g}$. Afterwards, either $\hat{h}\neq\hat{u}$ and coloring $h$ with a color not in $\hat{u}$ yields configuration $(\star\star)$, or $\hat{h}=\hat{u}$ and we can color $g,c$ with a color not in $\hat{u}$, then $j,k,v_2,u,h$.
    \end{itemize}
    We may thus assume that $\hat{q}$ and $\hat{f}$ are disjoint.

\item Assume that $\hat{h}\neq\hat{d}$.   Since $\hat{q}$ and $\hat{f}$ are disjoint and $\hat{q}\subset\hat{m}$, there must exist a color $\alpha\in\hat{f}\setminus\hat{m}$. If $\alpha\notin\hat{h}$, we are left with the configuration of the previous item hence we may assume $\alpha\in \hat{h}$ (hence $\alpha\notin\hat{n}$ since $\hat{h}\cap\hat{n}=\varnothing$). 

Note that after coloring $f$ with $\alpha$ (and possibly removing a color from $\hat{d}$), $\hat{h}$ and $\hat{d}$ are distinct lists of size $2$. Color now $h$ with a color not in $\hat{d}$, then $o,n$, and observe that we get configuration $(\star\star\star)$.

Therefore, we can assume that $\hat{h}=\hat{d}$.
    
    \item Similarly, assume that $\hat{o}\not\subset\hat{n}$. Color again $f$ with $\alpha$, and observe that it does not impact $\hat{n}$ nor $\hat{o}$ so we can still color $o$ with a color not in $\hat{n}$, then $h$ and $g$. Now forget $n,v_4,g$ and observe that we are left with configuration $(\star\star)$.

    Therefore, we can assume that $\hat{o}\subset\hat{n}$. In particular, we get that $\hat{o}\cap\hat{h}=\varnothing$.
    \item If $\hat{\ell}\cap\hat{h}\neq\varnothing$, color $h$ and $\ell$ with the same color (which lies in $\hat{d}$, hence not in $\hat{k}$), forget $e,m,\ell$ and color $o$. Now we are left with the configuration $(\star)$.

    Therefore, we may assume that $\hat{\ell}$ is disjoint from $\hat{h}$, hence from $\hat{d}$.
    \item We may thus assume that $h$ is not incident with $o$ anymore, and that $d$ is not incident to $k,\ell$ anymore. Now we color $f$ with a color not in $\hat{h}$ (nor $\hat{d}$). If afterwards $\hat{n}=\hat{o}$, we remove their colors from $\hat{v_4}$ and $\hat{g}$ and forget $o,n,m,e,q,\ell$. We can now color $h$ arbitrarily, and obtain the configuration $(\star\star\star)$.

   Otherwise, we color $n$ and $v_4$ with colors not in $\hat{o}$ and forget $o$. 
   \begin{itemize}
   \item If $\hat{u}=\hat{h}$ afterwards, then remove the colors of $\hat{u}$ from $\hat{c},\hat{e}$ and $\hat{g}$, so we can forget $h,d,u,v_2$, and color $g$. Now color $q$ and $m$ and apply Lemma~\ref{lem:fryingpan} to color $\{j,k,\ell,e,c\}$. 
   \item Otherwise, we color $u$ with a color not in $\hat{h}$. If afterwards, $\hat{g}=\hat{h}$, we apply the same argument as in the previous item. 
   \item Therefore, we may color $u$ and $g$ with colors outside of $\hat{h}$. Now color $k$ such that $\hat{\ell}\neq\hat{m}$ afterwards, then $j,v_2,c,d,h$. Now two colors are remaining on $q$, and three on $e,\ell,m$, but $\hat{\ell}\neq\hat{m}$ so we can color $\{e,\ell,m,q\}$.
\end{itemize} 
\end{itemize}
\end{enumerate}
\end{proof}

\begin{lemma}
  \label{lem:C20ab}
The graph $G$ does not contain $C_{\ref{C15.5a}b}$.
\end{lemma}

\begin{proof}
  We use the notation depicted in Figure~\ref{fig:C20ab}.  By
  minimality, we take a coloring $\gamma$ of
  $G\setminus\{a,b,c,v_1\}$, uncolor $d,e,f,g,h,v_2$ and forget $v_1$.

  \begin{figure}[!h]
    \centering
    \begin{tikzpicture}[v/.style={draw=black,minimum size = 10pt,ellipse,inner sep=1pt}]
      \node[v, very thick, label=right:{$u$}] (u) at (0,0)  {8};
      \node[v, very thick] (v0) at (45:1.5) {8};
      \node[v,label=above:{$v_1$}] (v1) at (90:1.5) {3};
      \node[v, very thick] (v2) at (135:1.5) {8};
      \node[v,label=left:{$v_2$}] (v3) at (180:1.5) {5};
      \node[xshift=-1.5cm, very thick,v] (w1) at (143:1) {6};
      \node[xshift=-1.5cm, very thick,v] (w2) at (215:1) {6};
      \node[xshift=-1.5cm, very thick,v] (w3) at (287:1) {8};
      
      \draw (v3) -- (w1) node[midway,above] {$e$};
      \draw (v3) -- (w2) node[midway,below] {$f$};
      \draw (v3) -- (w3) node[midway,right] {$g$};
      \draw (v3) -- (u) node[midway,above] {$h$};
      \draw[ very thick] (v2) -- (u);
      \draw[ very thick] (v0) -- (u);
      \draw (v1) -- (u) node[midway,left] {$c$};
      \draw (v3) -- (v2) node[midway,left] {$d$};
      \draw (v2) -- (v1) node[midway,above] {$b$};
      \draw (v1) -- (v0) node[midway,above] {$a$};
    \end{tikzpicture}
\caption{Notation for Lemma~\ref{lem:C20ab}}
    \label{fig:C20ab}
  \end{figure}

  We have $|\hat{a}|=|\hat{g}|=2$,
  $|\hat{b}|=|\hat{c}|=|\hat{d}|=|\hat{h}|=3$, $|\hat{e}|=|\hat{f}|=4$
  and $|\hat{v_2}|=5$.

  If we color $d,e,f,g,h,v_2$ with their colors in $\gamma$, then the
  only way for the coloring not to extend to $G$ is to have
  $\hat{a},\hat{b}$ and $\hat{c}$ to be the same list of size two.  To
  avoid this, our goal is to find another coloring of $d,e,f,g,h,v_2$
  which differs from $\gamma$ on either $d$ or $h$. We consider the
  color shifting graph $H$ of $\{d,e,f,g,h,v_2\}$. By
  Lemma~\ref{lem:SCC}, there exists a strongly connected component $C$
  of $H$ such that $|C|>\max_{x\in C} d^-(x)$. By
  Lemma~\ref{lem:degmin}, this inequality ensures that $|C|>1$. We show
  that $C$ contains either $d$ or $h$ by distinguishing some cases:
  \begin{enumerate}
  \item If $C$ contains a vertex $s_\alpha$, then
    $|C|>d^-(s_\alpha)=|V(H)|-1$. Then $C=V(H)$ and contains $d$ and
    $h$.
  \item Otherwise, if $C$ contains $v_2$, then $|C|\geqslant 5$ and
    $C$ contains either $d$ or $h$.
  \item Otherwise, if $C$ contains $e$ or $f$, then $|C|\geqslant 4$ and
    $C$ contains either $d$ or $h$.
  \end{enumerate}
 Otherwise, $C$ has size at least $2$ and is contained in
  $\{d,g,h\}$, hence it contains $d$ or $h$. Thus, we can apply
  Lemma~\ref{lem:recolor} to ensure that we can recolor $d$ or $h$,
  hence we can extend the coloring to $G$.
\end{proof}

\begin{lemma}
  \label{lem:C20b}
  The graph $G$ does not contain $C_{\ref{C15.5a}c}$.
\end{lemma}

\begin{proof}
  We use the notation depicted in Figure~\ref{fig:C20b}.  By
  minimality, we color $G\setminus\{a,p,i\}$, and uncolor
  $a,\ldots,s,v_1,\ldots,v_5$. We forget $v_1,v_3$.

  \begin{figure}[!h]
    \centering
    \begin{tikzpicture}[v/.style={draw=black,minimum size = 10pt,ellipse,inner sep=1pt}]
      \node[v,label=85:{$u$}] (u) at (0,0)  {8};
      \node[v, very thick] (v0) at (45:1.5) {8};
      \node[v,label=above:{$v_1$}] (v1) at (90:1.5) {3};
      \node[v, very thick] (v2) at (135:1.5) {8};
      \node[v,label=left:{$v_4$}] (v3) at (180:1.5) {5};
      \node[v, very thick] (v4) at (225:1.5) {8};
      \node[v,label=below left:{$v_3$}] (v5) at (270:1.5) {4};
      \node[v, very thick] (v6) at (315:1.5) {8};
      \node[v,label=right:{$v_2$}] (v7) at (0:1.5) {5};
      \node[v, very thick] (w) at (270:2.5) {8};
      \node[v,xshift=1.5cm,label=right:{$v_5$}] (w1) at (36:1) {5};
      \node[v, very thick,xshift=1.5cm] (w2) at (-36:1) {8};
      \draw (v6) -- (u) node[midway,right] {$d$};
      \draw (v5) -- (u) node[midway,right] {$e$};
      \draw (v4) -- (u) node[midway,below] {$f$};
      \draw (v3) -- (u) node[midway,below] {$g$};
      \draw (v2) -- (u) node[midway,left] {$h$};
      \draw (v1) -- (u) node[midway,left] {$a$};
      \draw (v0) -- (u) node[midway,above] {$b$};
      \draw (v7) -- (u) node[midway,above] {$c$};
      \draw (v5) -- (v6) node[midway,below] {$\ell$};
      \draw (v5) -- (v4) node[midway,below] {$m$};
      \draw (v4) -- (v3) node[midway,left] {$n$};
      \draw (v3) -- (v2) node[midway,left] {$o$};
      \draw (v2) -- (v1) node[midway,above] {$p$};
      \draw (v1) -- (v0) node[midway,above] {$i$};
      \draw (v0) -- (v7) node[midway,right] {$j$};
      \draw (v6) -- (v7) node[midway,right] {$k$};
      \draw (v5) -- (w) node[midway,right] {$s$};
      \draw (v7) -- (w1) node[midway,above] {$q$};
      \draw (v7) -- (w2) node[midway,below] {$r$};
    \end{tikzpicture}
\caption{Notation for Lemma~\ref{lem:C20b}}
    \label{fig:C20b}
  \end{figure}

    If $v_4=v_5$, the sizes of the color lists may differ depending on whether $r,s$ are incident with neighbors of $u$ distinct from $v_2,v_3$, but we have at least the following: $|\hat{r}|=|\hat{s}|=2$,
    $|\hat{n}|=|\hat{o}|=3$, $|\hat{b}|=|\hat{d}|=|\hat{f}|=|\hat{h}|=|\hat{i}|=|\hat{j}|=|\hat{k}|=|\hat{\ell}|=|\hat{m}|=|\hat{p}|=4$, 
  $|\hat{v_4}|=|\hat{u}|=6$, $|\hat{v_2}|=7$, $|\hat{g}|=8$, $|\hat{q}|=9$ and
  $|\hat{a}|=|\hat{c}|=|\hat{e}|=10$. We color $v_2$ with a color not in $\hat{v_4}$, then $r,s$ and finally forget $v_4$ and $q$. We color the remaining elements using Theorem~\ref{thm:nss}.

Assume now that $v_4\neq v_5$. Then by $C_{\ref{C3a}}$, $v_5$ cannot be a neighbor of $u$ and $r\neq s$. The sizes of the color lists may differ depending on whether $r,s$ are incident with neighbors of $u$ distinct from $v_2,v_3$, but we have at least the following: $|\hat{n}|=|\hat{o}|=|\hat{r}|=|\hat{s}|=|\hat{v_5}|=2$,
  $|\hat{v_4}|=|\hat{b}|=|\hat{d}|=|\hat{f}|=|\hat{h}|=|\hat{i}|=|\hat{j}|=|\hat{k}|=|\hat{\ell}|=|\hat{m}|=|\hat{p}|=4$,
  $|\hat{q}|=|\hat{u}|=6$, $|\hat{v_2}|=7$, $|\hat{g}|=8$ and
  $|\hat{a}|=|\hat{c}|=|\hat{e}|=10$.

  For all items except the last one, we remove from $\hat{v_2}$ the
  colors from $\hat{v_5}$, so that $\hat{v_5}$ becomes disjoint from
  $\hat{v_2}$, so we can forget $v_5$ then $q$.
  
  \begin{itemize}
  \item If $\hat{h}\cap \hat{n}\neq \varnothing$, we color $h$ and $n$
    with the same color, then $o,r,s$.  We color the remaining
    elements using Theorem~\ref{thm:nss}. Therefore, we may assume
    that $\hat{h}$ and $\hat{n}$ are disjoint.
  \item If $\hat{h}\cap\hat{v_4}\neq\varnothing$, we color $h$ and
    $v_4$ with the same color (hence not in $\hat{n}$), then
    $o,n,r,s$. We color the remaining elements using
    Theorem~\ref{thm:nss}.  Therefore, we may assume that $\hat{h}$
    and $\hat{v_4}$ are disjoint.

  \item If $\hat{f}\cap \hat{o}\neq \varnothing$, we color $f$ and $o$
    with the same color, then $n,m,s,r$. We color the remaining
    elements using Theorem~\ref{thm:nss}. Therefore, we may assume
    that $\hat{f}$ and $\hat{o}$ are disjoint.
  \item If $\hat{f}\cap\hat{v_4}\neq\varnothing$, we color $f$ and
    $v_4$ with the same color (hence not in $\hat{o}$), then
    $n,o,m,s,r$. We color the remaining elements using
    Theorem~\ref{thm:nss}. Therefore, we may assume that $\hat{f}$ and
    $\hat{v_4}$ are disjoint.
  \item If $\hat{o}\cup\hat{n}\not\subset\hat{v_4}$, then we color $n$
    or $o$ with a color not in $\hat{v_4}$, then $o$ or $n$, then
    $r,s$, and we forget $v_4$. We color the remaining elements using
    Theorem~\ref{thm:nss}. Therefore, we may assume that $\hat{n}$ and
    $\hat{o}$ are contained in $v_4$.
  \item If $\hat{n}\neq\hat{o}$, we color $n$ with a color not in
    $\hat{o}$, then $f,m,s,r$. We remove $\hat{v_5}$ from $\hat{q}$,
    so that we can forget $v_5$ and $v_2$. We color the remaining
    elements using Theorem~\ref{thm:nss}. Therefore, we may assume
    that $\hat{n}=\hat{o}$.
  \end{itemize}
  Now we have $\gamma(h)\in \hat{h}$, hence not in $\hat{o}$ since
  $\hat{o}\subset \hat{v_4}$ which is disjoint from
  $\hat{h}$. Therefore, $\gamma(h)\notin \hat{o}$, and similarly,
  $\gamma(f)\notin\hat{n}$. We color $h,b,f$ and $d$ with their color in $\gamma$.  Since
  $\hat{n}=\hat{o}$ and $\{g,n,o,v_4\}$ is colorable, coloring $g$ and
  $v_4$ with their color in $\gamma$ does not affect $\hat{n}$ and
  $\hat{o}$. We also color $u$ with its color in $\gamma$. We remove $\hat{v_5}$ from $\hat{q}$, so that $\hat{v_5}$ becomes disjoint from $\hat{q}$, hence we can forget $v_5$ and $v_2$ and finally we color $r$ with $\gamma(r)$.

  Therefore, we obtain $|\hat{i}|=|\hat{j}|=|\hat{n}|=|\hat{o}|=|\hat{s}|=2$,
  $|\hat{c}|=|\hat{k}|=|\hat{\ell}|=|\hat{m}|=|\hat{p}|=|\hat{q}|=3$,
  $|\hat{a}|=|\hat{e}|=4$ (with maybe more available colors if $s$ is incident to two neighbors of $u$).
  

  Let $\alpha\in \hat{e}\setminus\hat{m}$. We distinguish two cases.
  \begin{itemize}
  \item Assume that $\{c,j,k,q\}$ stays colorable when we remove
    $\alpha$ from $\hat{c}$. Let $H$ be the color shifting graph of $\{c,j,k,q\}$. By
    Lemma~\ref{lem:SCC}, there exists a strongly connected component
    $C$ of $H$ such that $C$ contains all the in-neighbors of its vertices. By Lemma~\ref{lem:degmin}, this
    ensures that $|C|>d^-(i)=1$. 
    \begin{itemize}
    \item If $C$ contains $k$, then we can recolor $k$ by
      Lemma~\ref{lem:recolor}. In either the old or the new coloring, we have $\hat{e}\neq\hat{m}$ (since $c$ cannot be colored $\alpha$), and $\hat{\ell}\neq\hat{s}$ (since $k$ is recolored). Hence we can color
    $s$ with a color not in $\hat{\ell}$, then $i$ arbitrarily and finally apply Lemma~\ref{lem:fryingpan} to $\{\ell,m,n,o,p,a,e\}$. Thus we may assume that $C$ does not contain $k$. In particular, $C$ cannot contain some $s_\beta$.

    \item If $q\in C$, we have $|C|>2$, hence $C$ contains $j$. Otherwise, $C\subset\{c,j\}$ and if $c\in C$, then $|C|>1$ and $C$ contains $j$ again.
    \end{itemize}
    Therefore, $C$ contains $j$ but not $k$, so we can recolor $j$ without changing $k$. Denote by $\gamma'$ the new coloring of $\{c,j,k,q\}$, and observe that the lists of available colors of $i$ changed, and those of $\ell,m,n,o,p,s$ did not change.
    
    If $\alpha\notin\hat{\ell}$, we color $e$ with $\alpha$ and $\ell,m,n,o,s$ with their color in $\gamma$ arbitrarily. Now the list $\hat{p}$ did not change but $\hat{i}$ did, hence we can color $a,i,p$.

    Therefore, assume that $\alpha\in\hat{\ell}=\hat{s}$. We color $\ell$ and $s$ with their color in $\gamma$. In particular, one of them gets color $\alpha$ so it remains at least two colors in $\hat{m}$ afterwards. If $\hat{e}$ has size 2, then we color $i$ arbitrarily and color $\{a,e,m,n,o,p\}$ using Corollary~\ref{cor:evencycle}. This implies that $\hat{e}$ must have size 1 when $c$ is colored with $\gamma$ or $\gamma'$ and $\ell,s$ are colored. Therefore, we must have $\gamma'(c)=\gamma(e)$, and that after recoloring (and coloring $\ell,s$), $\hat{e}=\{\gamma(c)\}$. Now color $e,m,n,o$ and observe that the list $\hat{a}$ did not change (the recoloring only swapped the colors of $c$ and $e$), while $\hat{i}$ did, hence we can color $a,i,p$. 
    
  \item Assume that $\{c,j,k,q\}$ is not colorable when we remove
    $\alpha$ from $\hat{c}$. This means that $\gamma(c)=\alpha$. In
    particular, when coloring $\{c,j,k,q\}$ with their color in
    $\gamma$, we obtain that $\hat{e}$ and $\hat{m}$ are the same list
    of size 3 (otherwise there would have been two choices for $\alpha$, one of them satisfying the hypothesis of the previous item). Since $\{e,\ell,m,s\}$ is colorable, there must exist a
    color in $\hat{\ell}\cup\hat{s}$ not in $\hat{m}$. We color $\ell$
    or $s$ (say $\ell$, by symmetry) with this color, then $s$. We
    then apply Lemma~\ref{lem:fryingpan} to
    $\{i,p,o,n,m,e,a\}$.
  \end{itemize}
\end{proof}
 
\subsection{Configuration $C_{\ref{C15.5b}}$}
By definition of $E_3$-neighbor, and since $G$ does not contain
$C_{\ref{C15.5a}b}$, we know that if $G$ contains $C_{\ref{C15.5b}}$,
then it contains one of the following configurations.

\begin{itemize}
\item[$\bullet$] $C_{\ref{C15.5b}a}$ is a $8$-vertex $u$ with a weak
  neighbor $v$ of degree $3$, and a $(7,8)$-neighbor of degree 4 at
  triangle distance 2 from $v$.
\item[$\bullet$] $C_{\ref{C15.5b}b}$ is a $8$-vertex $u$ with two weak
  neighbors of degree $3$ and $4$ at triangle distance 2, and a
  $(6,6)$-neighbor of degree 5.\item[$\bullet$] $C_{\ref{C15.5b}c}$ is a $8$-vertex $u$ with a weak
  neighbor $v$ of degree $3$, a weak neighbor of degree 4 at triangle
  distance 2 from $v$, and two $(6,8)$-neighbors of degree 5.\item[$\bullet$] $C_{\ref{C15.5b}d}$ is a $8$-vertex $u$ with a weak
  neighbor $v$ of degree $3$, a weak neighbor of degree 4 at triangle
  distance 2 from $v$, and two $(7^+,8)$-neighbors of degree 5 such
  that one of them is a triangle-distance 2 from $v$ and has a
  neighbor of degree 5.
\item[$\bullet$] $C_{\ref{C15.5b}e}$ is a $8$-vertex $u$ with a weak
  neighbor $v$ of degree $3$, a weak neighbor of degree 4 at triangle
  distance 2 from $v$, and two $(7^+,8)$-neighbors of degree 5 such
  that one of them is a triangle-distance 4 from $v$ and has a
  neighbor of degree 5.
\item[$\bullet$] $C_{\ref{C15.5b}f}$ is a $8$-vertex $u$ with a weak
  neighbor $v$ of degree $3$, a weak neighbor of degree 4 at triangle
  distance 2 from $v$, and two $(7^+,8)$-neighbors of degree 5 such
  that one of them is a triangle-distance 4 from $v$ and has two
  neighbors of degree 6.
\end{itemize}

We dedicate a lemma to each of these configurations. 
\begin{lemma}
  \label{lem:C20aa}
The graph $G$ does not contain $C_{\ref{C15.5b}a}$.
\end{lemma}

\begin{proof}
  We use the notation depicted in Figure~\ref{fig:C20aa}.  By
  minimality, we take a coloring $\gamma$ of $G\setminus\{a,b,c\}$,
  and uncolor $d,e,f,g,v_1,v_2$. We forget $v_1,v_2$.
  \begin{figure}[!h]
    \centering
    \begin{tikzpicture}[v/.style={draw=black,minimum size = 10pt,ellipse,inner sep=1pt}]
      \node[v,label=right:{$u$}] (u) at (0,0)  {8};
      \node[v, very thick] (v0) at (45:1.5) {8};
      \node[v,label=above:{$v_1$}] (v1) at (90:1.5) {3};
      \node[v, very thick] (v2) at (135:1.5) {8};
      \node[v,label=above left:{$v_2$}] (v3) at (180:1.5) {4};
      \node[v, very thick] (v4) at (225:1.5) {7};
      \node[v, very thick] (w) at (180:3) {8};
      \draw (v3) -- (w) node[midway,below] {$d$};
      \draw [ very thick] (v4) -- (u);
      \draw (v3) -- (u) node[midway,above] {$f$};
      \draw[ very thick] (v2) -- (u);
      \draw (v1) -- (u) node[midway,left] {$c$};
      \draw [ very thick] (v0) -- (u);
      \draw (v4) -- (v3) node[midway,left] {$g$};
      \draw (v3) -- (v2) node[midway,left] {$e$};
      \draw (v2) -- (v1) node[midway,above] {$a$};
      \draw (v1) -- (v0) node[midway,above] {$b$};
    \end{tikzpicture}
\caption{Notation for Lemma~\ref{lem:C20aa}}
    \label{fig:C20aa}
  \end{figure}

  We have $|\hat{b}|=|\hat{d}|=2$ and
  $|\hat{a}|=|\hat{c}|=|\hat{e}|=|\hat{f}|=|\hat{g}|=3$.
  \begin{itemize}
  \item If $\hat{d}\not\subset\hat{g}$, then we can color $d$ with a
    color not in $\hat{g}$, forget $g$, and apply
    Lemma~\ref{lem:fryingpan} to $\{b,a,e,f,c\}$. We may thus assume
    that $\hat{d}\subset\hat{g}$.
  \item If $\hat{d}\not\subset\hat{e}$, then we color $d$ with a color
    not in $\hat{e}$, color $b$, and apply Lemma~\ref{lem:fryingpan}
    to $\{g,e,a,c,f\}$. We may thus assume (by symmetry) that
    $\hat{d}\subset\hat{e}$ and $\hat{d}\subset\hat{f}$.
  \item If $\hat{f}\neq\hat{g}$, we color $g$ with a color not in
    $\hat{f}$ (hence not in $\hat{d}$). We color $b$ an apply
    Lemma~\ref{lem:fryingpan} to color $\{d,f,c,a,e\}$. Therefore, we
    may assume that $\hat{f}=\hat{g}$. By symmetry, we also have
    $\hat{e}=\hat{g}$.
  \end{itemize}

  Therefore, we have $|\hat{d}\cup\hat{e}\cup\hat{f}\cup\hat{g}|=3$,
  hence $\{d,e,f,g\}$ is not colorable. This is impossible since
  $\gamma$ is a proper coloring. This means that one of the previous
  cases should happen, hence that we can color $G$.
\end{proof}

\begin{lemma}
  \label{lem:C20a}
  The graph $G$ does not contain $C_{\ref{C15.5b}b}$.
\end{lemma}

\begin{proof}
  We use the notation depicted in Figure~\ref{fig:C20a}.  By
  minimality, we color $G\setminus\{a,\ldots,n,v_1,\ldots,v_5\}$. We
  forget $v_1,v_2$.

  \begin{figure}[!h]
    \centering
    \begin{tikzpicture}[v/.style={draw=black,minimum size = 10pt,ellipse,inner sep=1pt}]
      \node[v,label=85:{$u$}] (u) at (0,0)  {8};
      \node[v, very thick] (v0) at (45:1.5) {8};
      \node[v,label=above:{$v_1$}] (v1) at (90:1.5) {3};
      \node[v, very thick] (v2) at (135:1.5) {8};
      \node[v,label=above:{$v_5$}] (v3) at (180:1.5) {6};
      \node[v,label=left:{$v_4$}] (v4) at (225:1.5) {5};
      \node[v,label=right:{$v_3$}] (v5) at (270:1.5) {6};
      \node[v, very thick] (v6) at (315:1.5) {8};
      \node[v,label=right:{$v_2$}] (v7) at (0:1.5) {4};
      \draw (v6) -- (u) node[midway,right] {$d$};
      \draw (v5) -- (u) node[midway,right] {$e$};
      \draw (v4) -- (u) node[midway,below] {$f$};
      \draw (v3) -- (u) node[midway,below] {$g$};
      \draw (v2) -- (u) node[midway,left] {$h$};
      \draw (v1) -- (u) node[midway,left] {$a$};
      \draw (v0) -- (u) node[midway,above] {$b$};
      \draw (v7) -- (u) node[midway,above] {$c$};
      \draw (v5) -- (v4) node[midway,below] {$\ell$};
      \draw (v4) -- (v3) node[midway,left] {$m$};
      \draw (v2) -- (v1) node[midway,above] {$n$};
      \draw (v1) -- (v0) node[midway,above] {$i$};
      \draw (v0) -- (v7) node[midway,right] {$j$};
      \draw (v6) -- (v7) node[midway,right] {$k$};
    \end{tikzpicture}
\caption{Notation for Lemma~\ref{lem:C20a}}
    \label{fig:C20a}
  \end{figure}

  We have $|\hat{v_3}|=|\hat{v_5}|=|\hat{k}|=2$,
  $|\hat{d}|=|\hat{j}|=|\hat{h}|=|\hat{n}|=3$,
  $|\hat{b}|=|\hat{i}|=|\hat{\ell}|=|\hat{m}|=4$,
  $|\hat{v_4}|=|\hat{e}|=|\hat{g}|=6$, $|\hat{u}|=7$, $|\hat{f}|=8$,
  $|\hat{c}|=9$ and $|\hat{a}|=10$.

  We color $j$ and $d$ with colors not in $\hat{k}$ and forget $k$.

  \begin{itemize}
  \item If $\hat{b}=\hat{h}$, we remove the colors of $\hat{b}$ from
    $\hat{a},\hat{c}, \hat{e},\hat{f},\hat{g}$ and $\hat{u}$.

    We then color $\{u,v_3,v_4,v_5,e,f,g,\ell,m\}$ as done in
    Lemma~\ref{lem:C5a}. We then color $c$ and $a$. The remaining
    elements $\{h,b,i,n\}$ induce an even cycle, which is 2-choosable.
  \item Otherwise, we color $h$ with a color not in $\hat{b}$, then
    $i$ with a color not in $\hat{n}$ and forget $n,a,c$. We color $b$
    and we again come back to the case of Lemma~\ref{lem:C5a}.
  \end{itemize}
\end{proof}

\begin{lemma}
  \label{lem:C20c}
  The graph $G$ does not contain $C_{\ref{C15.5b}c}$.
\end{lemma}

\begin{proof}
  We use the notation depicted in Figure~\ref{fig:C20c}.  By
  minimality, we color $G\setminus\{a,\ldots,q,v_1,\ldots,v_5\}$. We
  forget $v_1,v_5$.

  \begin{figure}[!h]
    \centering
    \begin{tikzpicture}[v/.style={draw=black,minimum size = 10pt,ellipse,inner sep=1pt}]
      \node[v,label=85:{$u$}] (u) at (0,0)  {8};
      \node[v, very thick] (v0) at (45:1.5) {8};
      \node[v,label=above:{$v_1$}] (v1) at (90:1.5) {3};
      \node[v, very thick] (v2) at (135:1.5) {8};
      \node[v,label=below left:{$v_4$}] (v3) at (180:1.5) {4};
      \node[v, very thick] (v4) at (225:1.5) {8};
      \node[v,label=below left:{$v_3$}] (v5) at (270:1.5) {5};
      \node[v, very thick] (v6) at (315:1.5) {6};
      \node[v,label=right:{$v_2$}] (v7) at (0:1.5) {5};
      \node[v, very thick] (w) at (180:3) {8};
      \draw (v6) -- (u) node[midway,right] {$d$};
      \draw (v5) -- (u) node[midway,right] {$e$};
      \draw (v4) -- (u) node[midway,below] {$f$};
      \draw (v3) -- (u) node[midway,below] {$g$};
      \draw (v2) -- (u) node[midway,left] {$h$};
      \draw (v1) -- (u) node[midway,left] {$a$};
      \draw (v0) -- (u) node[midway,above] {$b$};
      \draw (v7) -- (u) node[midway,above] {$c$};
      \draw (v5) -- (v6) node[midway,below] {$\ell$};
      \draw (v5) -- (v4) node[midway,below] {$m$};
      \draw (v4) -- (v3) node[midway,left] {$n$};
      \draw (v3) -- (v2) node[midway,left] {$o$};
      \draw (v2) -- (v1) node[midway,above] {$p$};
      \draw (v1) -- (v0) node[midway,above] {$i$};
      \draw (v0) -- (v7) node[midway,right] {$j$};
      \draw (v6) -- (v7) node[midway,right] {$k$};
      \draw (v3) -- (w) node[midway,above] {$q$};
    \end{tikzpicture}
\caption{Notation for Lemma~\ref{lem:C20c}}
    \label{fig:C20c}
  \end{figure}

  We have $|\hat{j}|=|\hat{m}|=|\hat{q}|=2$,
  $|\hat{v_3}|=|\hat{b}|=|\hat{f}|=|\hat{h}|=|\hat{i}|=|\hat{n}|=|\hat{o}|=|\hat{p}|=4$,
  $|\hat{k}|=|\hat{\ell}|=5$, $|\hat{u}|=|\hat{d}|=7$,
  $|\hat{c}|=|\hat{e}|=8$ and $|\hat{a}|=|\hat{g}|=10$.  Moreover,
  $\hat{v_2}$ and $\hat{v_4}$ have size 5 or 6 depending on whether
  $v_2v_4\in E(G)$.

  We color $i$ with a color not in $\hat{j}$, then $o$ with a color
  not in $\hat{p}$, then $q$, then $b$ with a color not in $\hat{j}$,
  then $m,n,f,h,v_4$, and forget $p,a$. We then color the remaining
  elements using Theorem~\ref{thm:nss}.
\end{proof}

\begin{lemma}
  \label{lem:C20d}
  The graph $G$ does not contain $C_{\ref{C15.5b}d}$.
\end{lemma}

\begin{proof}
  We follow here the same approach as for $C_{\ref{C15.5a}c}$.  We use
  the notation depicted in Figure~\ref{fig:C20d}.  By minimality, we
  color $G\setminus\{a,p,i\}$, and uncolor
  $a,\ldots,s,v_1,\ldots,v_5$. We forget $v_1,v_4$.

  \begin{figure}[!h]
    \centering
    \begin{tikzpicture}[v/.style={draw=black,minimum size = 10pt,ellipse,inner sep=1pt}]
      \node[v,label=85:{$u$}] (u) at (0,0)  {8};
      \node[v, very thick] (v0) at (45:1.5) {8};
      \node[v,label=above:{$v_1$}] (v1) at (90:1.5) {3};
      \node[v, very thick] (v2) at (135:1.5) {8};
      \node[v,label=below left:{$v_4$}] (v3) at (180:1.5) {4};
      \node[v, very thick] (v4) at (225:1.5) {8};
      \node[v,label=below:{$v_3$}] (v5) at (270:1.5) {5};
      \node[v, very thick] (v6) at (315:1.5) {8};
      \node[v,label=right:{$v_2$}] (v7) at (0:1.5) {5};
      \node[v, very thick] (w) at (180:2.5) {8};
      \node[v,xshift=1.5cm,label=right:{$v_5$}] (w1) at (36:1) {5};
      \node[v, very thick,xshift=1.5cm] (w2) at (-36:1) {8};
      \draw (v6) -- (u) node[midway,right] {$d$};
      \draw (v5) -- (u) node[midway,right] {$e$};
      \draw (v4) -- (u) node[midway,below] {$f$};
      \draw (v3) -- (u) node[midway,below] {$g$};
      \draw (v2) -- (u) node[midway,left] {$h$};
      \draw (v1) -- (u) node[midway,left] {$a$};
      \draw (v0) -- (u) node[midway,above] {$b$};
      \draw (v7) -- (u) node[midway,above] {$c$};
      \draw (v5) -- (v6) node[midway,below] {$\ell$};
      \draw (v5) -- (v4) node[midway,below] {$m$};
      \draw (v4) -- (v3) node[midway,left] {$n$};
      \draw (v3) -- (v2) node[midway,left] {$o$};
      \draw (v2) -- (v1) node[midway,above] {$p$};
      \draw (v1) -- (v0) node[midway,above] {$i$};
      \draw (v0) -- (v7) node[midway,right] {$j$};
      \draw (v6) -- (v7) node[midway,right] {$k$};
      \draw (v3) -- (w) node[midway,above] {$s$};
      \draw (v7) -- (w1) node[midway,above] {$q$};
      \draw (v7) -- (w2) node[midway,below] {$r$};
    \end{tikzpicture}
\caption{Notation for Lemma~\ref{lem:C20d}}
    \label{fig:C20d}
  \end{figure}

    If $v_3=v_5$, the sizes of the color lists may differ depending on whether $r,s$ are incident with neighbors of $u$ distinct from $v_2,v_4$, but we have at least the following: $|\hat{r}|=|\hat{s}|=2$,
    $|\hat{\ell}|=|\hat{m}|=3$, $|\hat{b}|=|\hat{d}|=|\hat{f}|=|\hat{h}|=|\hat{i}|=|\hat{j}|=|\hat{k}|=|\hat{n}|=|\hat{o}|=|\hat{p}|=4$, 
  $|\hat{v_3}|=|\hat{u}|=6$, $|\hat{v_2}|=7$, $|\hat{e}|=|\hat{q}|=9$ and
  $|\hat{a}|=|\hat{c}|=|\hat{g}|=10$. We color $v_2$ with a color not in $\hat{v_3}$, then $r,s$ and finally forget $v_3$ and $q$. Then we color $\ell$ with a color not in $\hat{k}$, then $n$ with a color not in $\hat{m}$. We color the remaining elements using Theorem~\ref{thm:nss}.

 Assume now that $v_3\neq v_5$. The number of available colors for each element depends on whether $r,s$ are incident with more neighbors of $u$, but we have at least $|\hat{\ell}|=|\hat{m}|=|\hat{r}|=|\hat{s}|=|\hat{v_5}|=2$,
  $|\hat{v_3}|=|\hat{b}|=|\hat{d}|=|\hat{f}|=|\hat{h}|=|\hat{i}|=|\hat{j}|=|\hat{k}|=|\hat{n}|=|\hat{o}|=|\hat{p}|=4$,
  $|\hat{q}|=|\hat{u}|=6$, $|\hat{v_2}|=7$, $|\hat{e}|=8$ and
  $|\hat{a}|=|\hat{c}|=|\hat{g}|=10$.

  For all items except the last two, we remove from $\hat{v_2}$ the
  colors from $\hat{v_5}$, so that $\hat{v_5}$ becomes disjoint from
  $\hat{v_2}$, so we can forget $v_5$ then $q$.

  \begin{itemize}
  \item If $\hat{f}\cap \hat{\ell}\neq \varnothing$, we color $f$ and
    $\ell$ with the same color, then $m,n,s,r$.  We color the
    remaining elements using Theorem~\ref{thm:nss}. Therefore, we may
    assume that $\hat{f}$ and $\hat{\ell}$ are disjoint.
  \item If $\hat{d}\cap \hat{m}\neq \varnothing$, we color $d$ and $m$
    with the same color, then $\ell,k,r,s$. We color the remaining
    elements using Theorem~\ref{thm:nss}. Therefore, we may assume
    that $\hat{d}$ and $\hat{m}$ are disjoint.
  \item If $\hat{f}\cap \hat{v_3}\neq \varnothing$, we color $f$ and
    $v_3$ with the same color (hence not in $\hat{\ell}$), then
    $m,\ell,n,s,r$. We color the remaining elements using
    Theorem~\ref{thm:nss}. Therefore, we may assume that $\hat{f}$ and
    $\hat{v_3}$ are disjoint.
  \item If $\hat{d}\cap \hat{v_3}\neq \varnothing$, we color $d$ and
    $v_3$ with the same color (hence not in $\hat{m}$), then
    $\ell,m,k,r,s$. We color the remaining elements using
    Theorem~\ref{thm:nss}. Therefore, we may assume that $\hat{d}$ and
    $\hat{v_3}$ are disjoint.
  \item If $\hat{m}\neq\hat{\ell}$, we remove the colors of
    $\hat{v_5}$ from $\hat{q}$, so that we can forget $v_2$ and
    $v_5$. We then color $m$ with a color not in $\hat{\ell}$, then
    $f,n,s,r$. We color the remaining elements using
    Theorem~\ref{thm:nss}. Therefore, we may assume that
    $\hat{m}=\hat{\ell}$.
  \item We remove a color $\alpha\in \hat{v_2}\setminus\hat{q}$ from
    $\hat{u}$, so that if everything is colored except $v_2,v_5,q$, we
    obtain $\hat{v_2}\neq\hat{q}$. Therefore, we can forget
    $v_2,v_5,q$. The configuration is now symmetric (vertically).

    Assume that $\hat{m}\not\subset\hat{v_3}$ so that there exists
    $\alpha\in\hat{m}\setminus\hat{v_3}$. We color $m$ with $\alpha$
    and forget $v_3$. Since $\hat{m}=\hat{\ell}$, we have
    $\alpha\in \hat{\ell}$, hence $\alpha\notin \hat{f}$. In this
    case, we color $\ell$, then $f$ with a color not in $\hat{d}$,
    then $n,s,r$. We color the remaining elements using
    Theorem~\ref{thm:nss}. Therefore, we may assume that
    $\hat{m}\subset\hat{v_3}$ and then $\hat{\ell}\subset\hat{v_3}$
    since $\hat{m}=\hat{\ell}$.
  \end{itemize}
  Now we have $\gamma(f)\in \hat{f}$, hence not in $\hat{m}$ since
  $\hat{m}\subset \hat{v_3}$ which is disjoint from
  $\hat{f}$. Therefore, $\gamma(f)\notin \hat{m}$, and similarly,
  $\gamma(d)\notin\hat{\ell}$.

  We now color $h,b,f$ and $d$ with their color in $\gamma$.  Since
  $\hat{\ell}=\hat{m}$ and $\{e,\ell,m,v_3\}$ is colorable, coloring
  $e$ and $v_3$ with their color in $\gamma$ does not affect
  $\hat{\ell}$ and $\hat{m}$. We also color $u$ with its color in
  $\gamma$.

  We remove $\hat{v_5}$ from $\hat{q}$, so that $\hat{v_5}$ becomes
  disjoint from $\hat{q}$, hence we can forget $v_5$ and $v_3$.

  Therefore, we obtain $|\hat{\ell}|=|\hat{m}|=|\hat{r}|=|\hat{s}|=2$,
  $|\hat{i}|=|\hat{j}|=|\hat{k}|=|\hat{n}|=|\hat{o}|=|\hat{p}|=3$,
  $|\hat{a}|=|\hat{c}|=|\hat{g}|=|\hat{q}|=4$. Moreover,
  $\hat{\ell}=\hat{m}$.

  Observe that if we color everything but $a,i,p$ with their color in
  $\gamma$, the only problematic case is when $\hat{a},\hat{i}$ and
  $\hat{p}$ are the same list of size 2. Then, any recoloring of
  $j$ or $o$ can break this condition.

  Note that since $\hat{\ell}=\hat{m}$, it is sufficient to color the
  graph obtained by identifying the endpoints of $\ell$ and $m$ that
  are not $v_3$ (so that $k$ and $n$ become incident), and by removing
  $\ell$ and $m$.

  Let $\alpha\in \hat{g}\neq\hat{o}$. We distinguish two cases.
  \begin{itemize}
  \item Assume that $\{c,j,k,q,r\}$ stays colorable when we remove
    $\alpha$ from $\hat{c}$. If $\hat{n}\neq\hat{s}\cup\{\gamma(k)\}$,
    then we color $c,j,k,q,r,\ell,m$, so that $\hat{n}\neq\hat{s}$ and
    $\hat{o}\neq\hat{g}$. We then color $n$ with a color not in
    $\hat{s}$, then color $i$ arbitrarily. We then apply
    Lemma~\ref{lem:fryingpan} to $\{s,o,p,a,g\}$ since
    $\hat{o}\neq\hat{g}$.
    
    Let $H$ be the color shifting graph of $\{c,j,k,q,r\}$. By
    Lemma~\ref{lem:SCC}, there exists a strongly connected component
    $C$ of $H$ such that $C$ contains all the in-neighbors of its vertices. By Lemma~\ref{lem:degmin}, this
    ensures that $|C|>d^-(r)=1$.
    \begin{itemize}
    \item If $C$ contains $j$, then we can recolor $j$ by
      Lemma~\ref{lem:recolor}, which now breaks
      $\hat{a}=\hat{i}=\hat{p}$ after having colored every other
      element. Thus we may assume that $C$ does not contain $j$.
    
    \item If $C$ contains $k$, then we can recolor $k$ by
      Lemma~\ref{lem:recolor}, and the condition
      $\hat{n}=\hat{s}\cup\{\gamma(k)\}$ does not hold anymore with
      the new coloring. Thus we may assume that $C$ does not contain
      $k$.
    
    \item If $C$ contains some $s_\beta$, then it contains $j$ and
      $k$. 
    \item Otherwise, $C\subset\{c,q,r\}$. If $q\in C$, then $|C|>3$,
      which is not possible. 
    \item Otherwise $C\subset\{c,r\}$, hence $c\in C$ and $|C|>2$,
      which is again impossible. 
    \end{itemize}
    Therefore, we may always recolor either $j$ or $k$, and then
    extend the coloring to $G$.
  \item Assume that $\{c,j,k,q,r\}$ is not colorable when we remove
    $\alpha$ from $\hat{c}$. This means that $\gamma(c)=\alpha$. In
    particular, when coloring $\{c,j,k,q,r\}$ with their color in
    $\gamma$, we obtain that $\hat{o}$ and $\hat{g}$ are the same list
    of size 3. Since $\{g,n,o,s\}$ is colorable, there must exist a
    color in $\hat{n}\cup\hat{s}$ not in $\hat{o}$. We color $n$ or
    $s$ (say $n$, by symmetry) with this color, then $s$. We then
    apply Lemma~\ref{lem:fryingpan} to $\{i,p,o,g,a\}$.
  \end{itemize}
\end{proof}

\begin{lemma}
  \label{lem:C20e}
  The graph $G$ does not contain $C_{\ref{C15.5b}e}$.
\end{lemma}

\begin{proof}
  We follow here the same approach as for $C_{\ref{C15.5a}c}$.  We use
  the notation depicted in Figure~\ref{fig:C20e}.  By minimality,
  there exists a coloring $\gamma$ of $G\setminus\{a,p,i,v_1\}$. We
  uncolor $a,\ldots,s,v_1,\ldots,v_5$ and forget $v_1,v_4$.

  \begin{figure}[!h]
    \centering
    \begin{tikzpicture}[v/.style={draw=black,minimum size = 10pt,ellipse,inner sep=1pt}]
      \node[v,label=85:{$u$}] (u) at (0,0)  {8};
      \node[v, very thick] (v0) at (45:1.5) {8};
      \node[v,label=above:{$v_1$}] (v1) at (90:1.5) {3};
      \node[v, very thick] (v2) at (135:1.5) {8};
      \node[v,label=below:{$v_4$}] (v3) at (180:1.5) {4};
      \node[v, very thick] (v4) at (225:1.5) {8};
      \node[v,label=below:{$v_3$}] (v5) at (270:1.5) {5};
      \node[v, very thick] (v6) at (315:1.5) {8};
      \node[v,label=right:{$v_2$}] (v7) at (0:1.5) {5};
      \node[v,yshift=-1.5cm, very thick] (w1) at (234:1.5) {8};
      \node[v,yshift=-1.5cm,label=right:{$v_5$}] (w2) at (306:1.5) {5};
      \node[v,xshift=-1.5cm, very thick] (x) at (180:1.5) {8};
      \draw (v6) -- (u) node[midway,right] {$d$};
      \draw (v5) -- (u) node[midway,right] {$e$};
      \draw (v4) -- (u) node[midway,below] {$f$};
      \draw (v3) -- (u) node[midway,below] {$g$};
      \draw (v2) -- (u) node[midway,left] {$h$};
      \draw (v1) -- (u) node[midway,left] {$a$};
      \draw (v0) -- (u) node[midway,above] {$b$};
      \draw (v7) -- (u) node[midway,above] {$c$};
      \draw (v5) -- (v4) node[midway,below] {$m$};
      \draw (v4) -- (v3) node[midway,right] {$n$};
      \draw (v3) -- (v2) node[midway,left] {$o$};
      \draw (v2) -- (v1) node[midway,above] {$p$};
      \draw (v1) -- (v0) node[midway,above] {$i$};
      \draw (v0) -- (v7) node[midway,right] {$j$};
      \draw (v6) -- (v7) node[midway,right] {$k$};
      \draw (v6) -- (v5) node[midway,above] {$\ell$};
      \draw (v5) -- (w1) node[midway,left] {$r$};
      \draw (v5) -- (w2) node[midway,right] {$q$};
      \draw (v3) -- (x) node[midway,above] {$s$};
    \end{tikzpicture}
\caption{Notation for Lemma~\ref{lem:C20e}}
    \label{fig:C20e}
  \end{figure}
  
  Note that $v_2\neq v_5$ otherwise we obtain $C_{\ref{C15.5b}d}$. Hence we have $|\hat{j}|=|\hat{k}|=|\hat{r}|=|\hat{s}|=|\hat{v_5}|=2$,
  $|\hat{b}|=|\hat{d}|=|\hat{f}|=|\hat{h}|=|\hat{i}|=|\hat{\ell}|=|\hat{m}|=|\hat{n}|=|\hat{o}|=|\hat{p}|=|\hat{v_2}|=4$,
  $|\hat{q}|=|\hat{u}|=6$, $|\hat{v_3}|=7$, $|\hat{c}|=8$ and
  $|\hat{a}|=|\hat{e}|=|\hat{g}|=10$.

  For all items except the last one, we remove from $\hat{u}$ a color
  in $\hat{v_3}\setminus \hat{q}$, so that if we color everything but
  $q,v_3,v_5$, then $\hat{q}\neq\hat{v_3}$ if they are lists of size
  2. This means that we can forget about $q,v_3,v_5$.
  
  \begin{itemize}
  \item If $\hat{b}\cap \hat{k}\neq \varnothing$, we color $b$ and $k$
    with the same color, then $j,r,s$. We color the remaining elements
    using Theorem~\ref{thm:nss}. Therefore, we may assume that
    $\hat{b}$ and $\hat{k}$ are disjoint.
  \item If $\hat{d}\cap \hat{j}\neq \varnothing$, we color $d$ and $j$
    with the same color, then $k,\ell,r,s$. We color the remaining
    elements using Theorem~\ref{thm:nss}. Therefore, we may assume
    that $\hat{d}$ and $\hat{j}$ are disjoint.
  \item If $\hat{d}\cap \hat{v_2}\neq \varnothing$, we color $d$ and
    $v_2$ with the same color (hence not in $\hat{j}$), then
    $k,j,\ell,r,s$. We color the remaining elements using
    Theorem~\ref{thm:nss}. Therefore, we may assume that $\hat{d}$ and
    $\hat{v_2}$ are disjoint.
  \item If $\hat{b}\cap \hat{v_2}\neq \varnothing$, we color $b$ and
    $v_2$ with the same color (hence not in $\hat{k}$), then
    $j,k,r,s$.  We color the remaining elements using
    Theorem~\ref{thm:nss}. Therefore, we may assume that $\hat{b}$ and
    $\hat{v_2}$ are disjoint.
  \item If $\hat{k}\not\subset\hat{v_2}$ or
    $\hat{j}\not\subset\hat{v_2}$, we color $k$ (or $j$) with a color
    not in $\hat{v_2}$, then $j$ (or $k$), $r,s$, then forget
    $v_2$. We color the remaining elements using
    Theorem~\ref{thm:nss}. Therefore, we may assume that $\hat{k}$ and
    $\hat{j}$ are included in $\hat{v_2}$.
  \item If $\hat{j}\neq\hat{k}$, we color $j$ with a color not in
    $\hat{k}$, then $i,b$ and $s,r$. We remove $\hat{v_5}$ from
    $\hat{q}$, so that $\hat{v_5}$ becomes disjoint from $\hat{q}$,
    hence we can forget $v_5$ and $v_3$. We color the remaining
    elements using Theorem~\ref{thm:nss}. Therefore, we may assume
    that $\hat{k}=\hat{j}$.
  \end{itemize}    
  Now we have $\gamma(b)\in \hat{b}$, hence not in $\hat{j}$ since
  $\hat{j}\subset \hat{v_2}$ which is disjoint from
  $\hat{b}$. Therefore, $\gamma(b)\notin \hat{j}$, and similarly,
  $\gamma(d)\notin\hat{k}$.

  We now color $h,b,f$ and $d$ with their color in $\gamma$.  Since
  $\hat{j}=\hat{k}$ and $\{j,k,c,v_2\}$ is colorable, coloring $c$ and
  $v_2$ with their color in $\gamma$ does not affect $\hat{j}$ and
  $\hat{k}$. We also color $u$ with its color in $\gamma$.

  We remove $\hat{v_5}$ from $\hat{q}$, so that $\hat{v_5}$ becomes
  disjoint from $\hat{q}$, hence we can forget $v_5$ and $v_3$.

  Therefore, we obtain $|\hat{j}|=|\hat{k}|=|\hat{r}|=|\hat{s}|=2$,
  $|\hat{o}|=|\hat{p}|=|\hat{i}|=|\hat{\ell}|=|\hat{m}|=|\hat{n}|=3$,
  $|\hat{a}|=|\hat{e}|=|\hat{g}|=|\hat{q}|=4$. Moreover,
  $\hat{j}=\hat{k}$.
   
  Observe that if we color everything but $a,i,p$ with their color in
  $\gamma$, the only problematic case is when $\hat{a},\hat{i}$ and
  $\hat{p}$ are the same list of size 2. Then, any recoloring of
  $\ell$ or $o$ can break this condition.

  Let $\alpha\in \hat{g}\neq\hat{o}$. We distinguish two cases.
  \begin{itemize}
  \item Assume that $\{e,\ell,m,q,r\}$ stays colorable when we remove
    $\alpha$ from $\hat{e}$. If $\hat{n}\neq\hat{s}\cup\{\gamma(m)\}$,
    then we color $\ell,m,q,r,\ell,k$, so that $\hat{n}\neq\hat{s}$
    and $\hat{o}\neq\hat{g}$. We then color $n$ with a color not in
    $\hat{s}$, then color $i$ arbitrarily. We then apply
    Lemma~\ref{lem:fryingpan} to $\{s,o,p,a,g\}$ since
    $\hat{o}\neq\hat{g}$.
    
    Let $H$ be the color shifting graph of $\{e,\ell,m,q,r\}$. By
    Lemma~\ref{lem:SCC}, there exists a strongly connected component
    $C$ of $H$ such that $C$ contains all the in-neighbors of its vertices. By Lemma~\ref{lem:degmin}, this
    ensures that $|C|>d^-(r)=1$. 
    \begin{itemize}
    \item If $C$ contains $\ell$, then we can recolor $\ell$ by
      Lemma~\ref{lem:recolor}, which now breaks
      $\hat{a}=\hat{i}=\hat{p}$ after having colored every other
      element. Thus we may assume that $C$ does not contain $\ell$.
    
    \item If $C$ contains $m$, then we can recolor $m$ by
      Lemma~\ref{lem:recolor}, and the condition
      $\hat{n}=\hat{s}\cup\{\gamma(m)\}$ does not hold anymore with
      the new coloring. Thus we may assume that $C$ does not contain
      $m$.
    
    \item If $C$ contains some $s_\beta$, then it contains $m$ and
      $\ell$. 
    \item Otherwise, $C\subset\{e,q,r\}$. If $q\in C$, then $|C|>3$,
      which is not possible. 
    \item Otherwise $C\subset\{e,r\}$, hence $e\in C$ and $|C|>2$,
      which is again impossible. 
    \end{itemize}
    Therefore, we may always recolor either $\ell$ or $m$, and then
    extend the coloring to $G$.
  \item Assume that $\{e,\ell,m,q,r\}$ is not colorable when we remove
    $\alpha$ from $\hat{e}$. This means that $\gamma(e)=\alpha$. In
    particular, when coloring $\{e,\ell,m,q,r\}$ with their color in
    $\gamma$, we obtain that $\hat{o}$ and $\hat{g}$ are the same list
    of size 3. Since $\{g,n,o,s\}$ is colorable, there must exist a
    color in $\hat{n}\cup\hat{s}$ not in $\hat{o}$. We color $n$ or
    $s$ (say $n$, by symmetry) with this color, then $s$. We then
    apply Lemma~\ref{lem:fryingpan} to $\{i,p,o,g,a\}$.
  \end{itemize}
\end{proof}

\begin{lemma}
  \label{lem:C20f}
  The graph $G$ does not contain $C_{\ref{C15.5b}f}$.
\end{lemma}

\begin{proof}
  We use the notation depicted in Figure~\ref{fig:C20f}.  By
  minimality, there exists a coloring $\gamma$ of
  $G\setminus\{a,i,p,v_1\}$.

  \begin{figure}[!h]
    \centering
    \begin{tikzpicture}[v/.style={draw=black,minimum size = 10pt,ellipse,inner sep=1pt}]
      \node[v,label=85:{$u$}] (u) at (0,0)  {8};
      \node[v, very thick] (v0) at (45:1.5) {8};
      \node[v,label=above:{$v_1$}] (v1) at (90:1.5) {3};
      \node[v, very thick] (v2) at (135:1.5) {8};
      \node[v,label=below:{$v_4$}] (v3) at (180:1.5) {4};
      \node[v, very thick] (v4) at (225:1.5) {8};
      \node[v,label=below:{$v_3$}] (v5) at (270:1.5) {5};
      \node[v, very thick] (v6) at (315:1.5) {8};
      \node[v,label=right:{$v_2$}] (v7) at (0:1.5) {5};
      \node[v,yshift=-1.5cm, very thick] (w1) at (234:1.5) {6};
      \node[v,yshift=-1.5cm, very thick] (w2) at (306:1.5) {6};
      \node[v,xshift=-1.5cm, very thick] (x) at (180:1.5) {8};
      \draw (v6) -- (u) node[midway,right] {$d$};
      \draw (v5) -- (u) node[midway,right] {$e$};
      \draw (v4) -- (u) node[midway,below] {$f$};
      \draw (v3) -- (u) node[midway,below] {$g$};
      \draw (v2) -- (u) node[midway,left] {$h$};
      \draw (v1) -- (u) node[midway,left] {$a$};
      \draw (v0) -- (u) node[midway,above] {$b$};
      \draw (v7) -- (u) node[midway,above] {$c$};
      \draw (v5) -- (v4) node[midway,below] {$m$};
      \draw (v4) -- (v3) node[midway,right] {$n$};
      \draw (v3) -- (v2) node[midway,left] {$o$};
      \draw (v2) -- (v1) node[midway,above] {$p$};
      \draw (v1) -- (v0) node[midway,above] {$i$};
      \draw (v0) -- (v7) node[midway,right] {$j$};
      \draw (v6) -- (v7) node[midway,right] {$k$};
      \draw (v6) -- (v5) node[midway,above] {$\ell$};
      \draw (v5) -- (w1) node[midway,left] {$r$};
      \draw (v5) -- (w2) node[midway,right] {$q$};
      \draw (v3) -- (x) node[midway,above] {$s$};
    \end{tikzpicture}
\caption{Notation for Lemma~\ref{lem:C20f}}
    \label{fig:C20f}
  \end{figure}
  
  We uncolor $a,\ldots,s,v_1,v_2,v_3,v_4$ and forget $v_1,v_4$. We
  have $|\hat{s}|=|\hat{j}|=|\hat{k}|=2$,
  $|\hat{h}|=|\hat{b}|=|\hat{d}|=|\hat{f}|=|\hat{o}|=|\hat{p}|=|\hat{i}|=|\hat{\ell}|=|\hat{m}|=|\hat{n}|=|\hat{q}|=|\hat{r}|=|\hat{v_2}|=4$,
  $|\hat{v_3}|=|\hat{u}|=6$, $|\hat{c}|=8$,
  $|\hat{a}|=|\hat{g}|=|\hat{e}|=10$.

  We first prove that we can color $G$ unless
  $\hat{n}=\hat{s}\cup\{\gamma(f),\gamma(m)\}$. Indeed, otherwise, we
  color every element but $\{o,p,i,n,a,g,s\}$, and we obtain
  $|\hat{s}|=|\hat{i}|=2$,
  $|\hat{o}|=|\hat{p}|=|\hat{a}|=|\hat{g}|=3$, and either
  $|\hat{n}|=3$ or $|\hat{n}|=2$ and $\hat{n}\neq \hat{s}$. We focus
  on the last case (since we may always remove a color from $\hat{n}$
  in the first case to obtain the second one).

  \begin{itemize}
  \item If $\hat{o}\neq\hat{g}$, we color $n$ with a color not in
    $\hat{s}$, then $i$, and apply Lemma~\ref{lem:fryingpan} to color
    $\{s,o,p,a,g\}$. Thus we may assume that $\hat{o}=\hat{g}$.
  \item Since $\{o,g,n,s\}$ is colorable, we cannot have both
    $\hat{s}\subset\hat{o}$ and $\hat{n}\subset\hat{o}$. By symmetry,
    assume that we can color $s$ with a color not in $\hat{o}$ (hence
    not in $\hat{g}$). Then we color $n$ and apply
    Lemma~\ref{lem:fryingpan} to color $\{i,p,o,g,a\}$.
  \end{itemize}

  We uncolor the elements of $S=\{m,e,\ell,q,r,v_3\}$. Let $H$ be the
  color shifting graph of $S$. By Lemma~\ref{lem:SCC}, there exists a
  strongly connected component $C$ of $H$ such that
  $|C|>\max_{x\in C} d^-(x)$. By Lemma~\ref{lem:degmin}, this
  inequality ensures that $|C|>1$.
  \begin{itemize}
  \item If $C$ contains $m$, then we can recolor $m$ by
    Lemma~\ref{lem:recolor}, and the condition
    $\hat{n}=\hat{s}\cup\{\gamma(f),\gamma(m)\}$ does not hold anymore
    with the new coloring. Thus we may assume that $C$ does not contain $m$.
  \item If $C$ contains some $s_\alpha$, then $C=V(H)$, hence $C$
    contains $m$.
  \item Otherwise, $C\subset\{e,\ell,q,r,v_3\}$. If $C$ contains $v_3$, it
    has size at least 5, hence $C=\{e,\ell,q,r,v_3\}$. Since $C$ contains all in-neighbors of its vertices, all the colors of
    $\hat{e},\hat{\ell},\hat{q},\hat{r}$ and $\hat{v_3}$ are actually in
    $\gamma(\{e,\ell,q,r,v_3\})$. In particular, we get that the union of
    these lists has size $5$. Since $\{m,e,\ell,q,r,v_3\}$ is colorable,
    this means that we can color $m$ with a color not in
    $\hat{e}\cup\hat{\ell}\cup\hat{q}\cup\hat{r}\cup\hat{v_3}$. We may
    then color $n,s,o,g,p,i,a,e,\ell,q,r$.
  \item Otherwise, $C\subset\{e,\ell,q,r\}$. Since $C$ has size at least
    two, it contains one element among $e,r,q$, hence it has size four
    and $C=\{e,\ell,q,r\}$.

    Similarly to the previous item, this means that the union
    $\hat{e}\cup\hat{\ell}\cup\hat{q}\cup\hat{r}$ has size $4$. Since
    $\{m,e,\ell,q,r,v_3\}$ is colorable, this means that we can color $v_3$
    then $m$ with a color not in
    $\hat{e}\cup\hat{\ell}\cup\hat{q}\cup\hat{r}$. We may then color
    $n,s,o,g,p,i,a,e,\ell,q,r$.
  \end{itemize}
\end{proof}
 
\subsection{Configuration $C_{\ref{C16}}$}
For reducing the remaining configurations, we use the recoloration
technique.

\begin{lemma}
  \label{lem:C16}
  The graph $G$ does not contain $C_{\ref{C16}}$.
\end{lemma}

\begin{proof}
  First, we consider the notation depicted in
  Figure~\ref{fig:C16}. By minimality, we color
  $G\setminus\{a,\ldots,f\}$ and uncolor $v_1,v_2$.
  \begin{figure}[!h]
    \centering
     \begin{tikzpicture}[v/.style={draw=black,minimum size = 10pt,ellipse,inner sep=1pt}]
      \node[v,label=85:{$u$}] (u) at (0,0)  {8};
      \node[v, very thick] (v0) at (45:1.5) {8};
      \node[v,label=above:{$v_1$}] (v1) at (90:1.5) {3};
      \node[v, very thick] (v2) at (135:1.5) {8};
      \node[v, very thick] (v4) at (225:1.5) {8};
      \node[v,label=below:{$v_2$}] (v5) at (270:1.5) {3};
      \node[v, very thick] (v6) at (315:1.5) {8};
      \node[v,very thick, label=right:{$v_3$}] (v7) at (0:1.5) {5};
      \draw (v6) -- (u) node[midway,below] {$j$};
      \draw (v5) -- (u) node[midway,left] {$d$};
      \draw (v4) -- (u) node[midway,above] {$k$};
      \draw (v2) -- (u) node[midway,left] {$g$};
      \draw (v1) -- (u) node[midway,left] {$c$};
      \draw (v0) -- (u) node[midway,above] {$h$};
      \draw (v7) -- (u) node[midway,above] {$i$};
      \draw (v5) -- (v6) node[midway,below] {$f$};
      \draw (v5) -- (v4) node[midway,below] {$e$};
      \draw (v2) -- (v1) node[midway,above] {$a$};
      \draw (v1) -- (v0) node[midway,above] {$b$};
    \end{tikzpicture}
\caption{Notation for Lemma~\ref{lem:C16}}
    \label{fig:C16}
  \end{figure}
  We have $|\hat{a}|=|\hat{b}|=|\hat{e}|=|\hat{f}|=2$,
  $|\hat{c}|=|\hat{d}|=3$ and $|\hat{v_1}|=|\hat{v_2}|=8$. We forget
  $v_1,v_2$. In this situation, note that we can extend the coloring
  to $G$ if and only if one of the following conditions is satisfied:
  \begin{enumerate}
  \item $\hat{a}\neq\hat{b}$
  \item $\hat{e}\neq\hat{f}$
  \item $\hat{c}\setminus\hat{a}\neq\hat{d}\setminus\hat{e}$ 
  \item $|\hat{c}\setminus\hat{a}|\neq 1$ or $|\hat{d}\setminus\hat{e}|\neq 1$ 
  \end{enumerate}

  Indeed, if $\hat{a}\neq\hat{b}$ (or similarly $\hat{e}\neq\hat{f}$),
  we color $a$ with a color not in $\hat{b}$, then color
  $e,f,d,c,b$. Otherwise, we color $a,b,e,f$ arbitrarily. If one of
  the last two conditions holds, then we can color $c$ and
  $d$. Therefore, we can extend the coloring to $G$. Conversely, if
  none of these conditions holds, then however we color $a,b,e,f$, we
  obtain $\hat{c}=\hat{d}=\{\alpha\}$ so we cannot produce a coloring
  of $G$.

  Assume now that none of these conditions holds. We prove that we can
  recolor some elements among $g,h,j,k$. This ensures that one of
  these conditions will hold. If we uncolor $g,h,i,j,k,u$, we may
  assume that $|\hat{g}|=|\hat{h}|=|\hat{j}|=|\hat{k}|=2$,
  $|\hat{u}|=3$ and $|\hat{i}|=4$.

  Let $H$ be the color shifting graph of $\{g,h,i,j,k,u\}$. Recall
  that Lemma~\ref{lem:degmin} implies that the in-degree of any vertex
  $x\neq s_\alpha$ of $H$ is at least $|\hat{x}|-1$. By
  Lemma~\ref{lem:SCC}, there is a strong component $C$ of $H$ such
  that $|C|>\max_{x\in C} d^-(x)$. Note that this inequality ensures
  that $|C|>1$. We show that $C$ contains $g,h,j$ or $k$ by
  distinguishing three cases:
  \begin{enumerate}
  \item If $C$ contains a vertex $s_\alpha$, then we have
    $|C|>d^-(s_\alpha)=|V(H)|-1$. Therefore, we have $C=V(H)$, hence
    $C$ contains $g,h,j$ and $k$.
  \item Otherwise, if $C$ contains $i$, then we have
    $|C|>|\hat{i}|-1$, so $|C|\geqslant 4$. Hence $C$ also contains
    $g,h,j$ or $k$.
  \item Otherwise, if $C$ contains $u$, then $|C|\geqslant 3$ and $C$
    contains $g,h,j$ or $k$.
  \end{enumerate}
  We thus obtain that $C$ is a strong component of size at least $2$
  that contains $g,h,j$ or $k$. Therefore, there is a directed cycle
  containing at least one of these vertices. Thus, applying
  Lemma~\ref{lem:recolor} gives a valid coloring of $\{g,h,i,j,k,u\}$
  where the color of $g,h,j$ or $k$ is different from its color in the
  previous coloring. 

  With the initial coloring, we had $\hat{a}=\hat{b}$ and
  $\hat{d}=\hat{e}$. Since we recolored at least one element among
  $g,h,j,k$, we necessarily have $\hat{a}\neq\hat{b}$ or
  $\hat{d}\neq\hat{e}$ with the new coloring. Thus we can extend it to
  $a,b,c,d,e,f$.
\end{proof}

\subsection{Configuration $C_{\ref{C17}}$}
\begin{lemma}
  \label{lem:C17}
  The graph $G$ does not contain $C_{\ref{C17}}$.
\end{lemma}

\begin{proof}
  We use the notation depicted in Figure~\ref{fig:C17}. By minimality, we
  color $G\setminus\{a,\ldots,f\}$ and uncolor $v_1,v_2$.
  \begin{figure}[!h]
    \centering
    \begin{tikzpicture}[v/.style={draw=black,minimum size = 10pt,ellipse,inner sep=1pt}]
      \node[v,label=below:{$u$}] (u) at (0,0)  {8};
      \node[v, very thick] (v0) at (45:1.5) {8};
      \node[v,label=above:{$v_1$}] (v1) at (90:1.5) {3};
      \node[v, very thick] (v2) at (135:1.5) {8};
      \node[v,label=left:{$w_1$}] (v4) at (225:1.5) {};
      \node[v,label=left:{$w_2$}] (v3) at (180:1.5) {};
      \node[v,label=left:{$v_2$}] (v6) at (315:1.5) {3};
      \node[v,very thick, xshift=1.5cm] (v5) at (315:1.5) {8};
      \node[v, very thick] (v7) at (0:1.5) {8};
      \draw (v6) -- (u) node[midway,right] {$d$};
      \draw (v4) -- (u) node[midway,below] {$j$};
      \draw (v2) -- (u) node[midway,left] {$g$};
      \draw (v1) -- (u) node[midway,left] {$c$};
      \draw (v0) -- (u) node[midway,above] {$h$};
      \draw (v7) -- (u) node[midway,above] {$i$};
      \draw (v5) -- (v6) node[midway,above] {$f$};
      \draw (v6) -- (v7) node[midway,right] {$e$};
      \draw (v2) -- (v1) node[midway,above] {$a$};
      \draw (v1) -- (v0) node[midway,above] {$b$};
      \draw (u) -- (v3) node[midway,below] {$k$};
    \end{tikzpicture}
\caption{Notation for Lemma~\ref{lem:C17}}
    \label{fig:C17}
  \end{figure}
  We forget $v_1,v_2$ and we have
  $|\hat{a}|=|\hat{b}|=|\hat{e}|=|\hat{f}|=2$ and
  $|\hat{c}|=|\hat{d}|=3$. As in Lemma~\ref{lem:C16}, our goal is to
  obtain one of the following conditions:
  \begin{enumerate}
  \item $\hat{a}\neq\hat{b}$
  \item $\hat{e}\neq\hat{f}$
  \item $\hat{c}\setminus\hat{a}\neq\hat{d}\setminus\hat{e}$ 
  \item $|\hat{c}\setminus\hat{a}|\neq 1$ or
    $|\hat{d}\setminus\hat{e}|\neq 1$
  \end{enumerate}
  Assume that none of them holds. In this case, note that any
  recoloring of $g,h$ or $i$ is sufficient to ensure that one of these
  conditions holds. We uncolor $u,g,h,i,j,k$. We have two cases:
  \begin{enumerate}
  \item If $(d(w_1),d(w_2))=(4^-,7^-)$, we uncolor and forget $w_1$
    and we may assume that
    $|\hat{g}|=|\hat{h}|=|\hat{i}|=|\hat{k}|=2$, $|\hat{u}|=4$ and
    $|\hat{j}|=6$.
  \item If $(d(w_1),d(w_2))=(5^-,6^-)$, we may assume that 
    $|\hat{g}|=|\hat{h}|=|\hat{i}|=2$, $|\hat{u}|=|\hat{k}|=3$ and
    $|\hat{j}|=4$.
  \end{enumerate} 
  Denote by $H$ the color shifting graph of $\{g,h,i,j,k,u\}$. By
  Lemma~\ref{lem:SCC}, there exists a strongly connected component $C$
  of $H$ such that $|C|>\max_{x\in C} d^-(x)$. By
  Lemma~\ref{lem:degmin}, this inequality ensures that $|C|>1$. We
  show that $C$ contains $g,h$ or $i$ by distinguishing four cases:
  \begin{enumerate}
  \item If $C$ contains a vertex $s_\alpha$, then we have
    $|C|>d^-(s_\alpha)=|V(H)|-1$. Therefore, $C=V(H)$, and $C$
    contains $g,h,i$.
  \item Otherwise, if $C$ contains $j$, then it has size at least $4$,
    hence it also contains $g,h$ or $i$.
  \item Otherwise, if $C$ contains $u$, then it has size at least $3$,
    hence contains $g,h$ or $i$.
  \item Otherwise, $C\subset\{g,h,i,k\}$. If $C$ contains $k$, then
    its size is at least $2$, hence it also contains $g,h$ or $i$.
  \end{enumerate}
  We thus obtain that $C$ is a strong component of size at least $2$
  that contains $g,h$ or $i$. Therefore, there is a directed cycle
  containing one of these vertices. Thus, we can apply
  Lemma~\ref{lem:recolor} to ensure that one the conditions is now
  satisfied, hence we can extend the coloring to $G$.
\end{proof}

\subsection{Configuration $C_{\ref{C21}}$}
\begin{lemma}
  \label{lem:C21}
  The graph $G$ does not contain $C_{\ref{C21}}$.
\end{lemma}

\begin{proof}
  We use the notation depicted in Figure~\ref{fig:C21}. By minimality,
  we color $G\setminus\{a,b,c\}$ and uncolor $v_1,v_2,v_3$.
  \begin{figure}[!h]
    \centering
    \begin{tikzpicture}[v/.style={draw=black,minimum size = 10pt,ellipse,inner sep=1pt}]
      \node[v,label=85:{$u$}] (u) at (0,0)  {8};
      \node[v, very thick] (v0) at (45:1.5) {8};
      \node[v,label=above:{$v_1$}] (v1) at (90:1.5) {3};
      \node[v, very thick] (v2) at (135:1.5) {8};
      \node[v,label=left:{$v_2$}] (v4) at (225:1.5) {4};
      \node[v,label=left:{$v_3$}] (v3) at (180:1.5) {4};
      \node[v, very thick] (v6) at (315:1.5) {7};
      \node[v, very thick] (v5) at (270:1.5) {5};
      \node[v, very thick] (v7) at (0:1.5) {8};
      \draw (v2) -- (u) node[midway,left] {$d$};
      \draw (v3) -- (u) node[midway,below] {$j$};
      \draw (v6) -- (u) node[midway,right] {$g$};
      \draw (v1) -- (u) node[midway,left] {$c$};
      \draw (v5) -- (u) node[midway,right] {$h$};
      \draw (v4) -- (u) node[midway,below] {$i$};
      \draw (u) -- (v7) node[midway,above] {$f$};
      \draw (u) -- (v0) node[midway,above] {$e$};
      \draw (v2) -- (v1) node[midway,above] {$a$};
      \draw (v1) -- (v0) node[midway,above] {$b$};
    \end{tikzpicture}
\caption{Notation for Lemma~\ref{lem:C21}}
    \label{fig:C21}
  \end{figure}
  We forget $v_1,v_2,v_3$ and we have
  $|\hat{a}|=|\hat{b}|=|\hat{c}|=2$.

  If $\hat{a}\neq\hat{b}$, we can color $c$ arbitrarily, then $a$ and
  $b$. Therefore, assume $\hat{a}=\hat{b}$. In this case, note that
  any recoloring of $d$ or $e$ is sufficient to ensure that
  $\hat{a}\neq\hat{b}$. Denote by $H$ the color shifting graph of
  $S=\{d,e,f,g,h,i,j,u\}$. 

  We uncolor the elements of $S$. Note that we can assume that
  $|\hat{f}|=2$, $|\hat{d}|=|\hat{e}|=|\hat{g}|=3$,
  $|\hat{h}|=|\hat{u}|=5$ and $|\hat{i}|=|\hat{j}|=7$.

  By Lemma~\ref{lem:SCC}, there exists a strongly connected component
  $C$ of $H$ such that $|C|>\max_{x\in C} d^-(x)$. By
  Lemma~\ref{lem:degmin}, this inequality ensures that $|C|>1$. We
  show that $C$ contains $d$ or $e$ by distinguishing five cases:
  \begin{enumerate}
  \item If $C$ contains a vertex $s_\alpha$, then we have
    $|C|>d^-(s_\alpha)=|V(H)|-1$. Therefore, $C=V(H)$, and $C$
    contains $d$ and $e$.
  \item Otherwise, if $C$ contains $i$ or $j$, then it has size at least $7$,
    hence it also contains $d$ or $e$.
  \item Otherwise, if $C$ contains $u$ or $h$, then it has size at least $5$,
    hence contains $d$ or $e$.
  \item Otherwise, $C\subset\{d,e,f,g\}$. If $C$ contains $g$, then
    its size is at least $3$, hence it also contains $d$ or $e$.
  \item Otherwise, $C\subset\{d,e,f\}$. If $C$ contains $f$, then its
    size is at least $2$, hence it also contains $d$ or $e$.
  \end{enumerate}
  We thus obtain that $C$ is a strong component of size at least $2$
  that contains $d$ or $e$. Therefore, there is a directed cycle
  containing one of these vertices. Thus, we can apply
  Lemma~\ref{lem:recolor} to ensure that now $\hat{a}\neq\hat{b}$, so
  that we can extend the coloring to $G$.
\end{proof}

\subsection{Configuration $C_{\ref{C18}}$}
To prove that $G$ does not contain $C_{\ref{C18}}$, we prove that it
does not contain the three following configurations:
\begin{itemize}
\item[$\bullet$] $C_{\ref{C18}a}$: $u$ has a weak neighbor $v_1$ of
  degree 3, a $(7,8)$-neighbor $v_2$ of degree 4 such that
  $\dist_u(v_1,v_2)=2$, and neighbor $v_3$ of degree 4.
\item[$\bullet$] $C_{\ref{C18}b}$: $u$ has a weak neighbor $v_1$ of
  degree 3, a $(8,8)$-neighbor $v_2$ of degree 4 such that
  $\dist_u(v_1,v_2)=2$, a neighbor $v_3$ of degree 4 and a neighbor
  $v_4$ of degree 7.
\item[$\bullet$] $C_{\ref{C18}c}$: $u$ has a weak neighbor $v_1$ of
  degree 3, a $(7,8)$-neighbor $v_2$ of degree 4 such that
  $\dist_u(v_1,v_2)\geqslant 3$, and two neighbors of degree 4 and 7.
\end{itemize}

\begin{lemma}
  \label{lem:C18a}
  The graph $G$ does not contain $C_{\ref{C18}a}$.
\end{lemma}

\begin{proof}
  We use the notation depicted in Figure~\ref{fig:C18a}. By
  minimality, we color $G\setminus\{a,b,c\}$ and uncolor and forget
  $v_1,v_2,v_3$. We then uncolor $d,e,f,g$.
  \begin{figure}[!h]
    \centering
    \begin{tikzpicture}[v/.style={draw=black,minimum size = 10pt,ellipse,inner sep=1pt}]
      \node[v,label=right:{$u$}] (u) at (0,0)  {8};
      \node[v, very thick] (v0) at (45:1.5) {8};
      \node[v,label=above:{$v_1$}] (v1) at (90:1.5) {3};
      \node[v, very thick] (v2) at (135:1.5) {8};
      \node[v,label=below left:{$v_2$}] (v3) at (180:1.5) {4};
      \node[v,xshift=-1.5cm] (v3') at (180:1.5) {8};
      \node[v, very thick] (v4) at (225:1.5) {7};
      \node[v,label=right:{$v_3$}] (v7) at (315:1.5) {4};

      \draw (v2) -- (v1) node[midway,above] {$a$};
      \draw (v1) -- (v0) node[midway,above] {$b$};
      \draw (v1) -- (u) node[midway,left] {$c$};
      \draw (v3) -- (v3') node[midway,above] {$g$};
      \draw (v2) -- (v3) node[midway,left] {$d$};
      \draw (v3) -- (u) node[midway,above] {$e$};
      \draw (v3) -- (v4) node[midway,left] {$f$};
      \draw[ very thick] (u) -- (v0);
      \draw[ very thick] (u) -- (v2);
      \draw[ very thick] (u) -- (v4);
      \draw[ very thick] (u) -- (v7);
    \end{tikzpicture}
\caption{Notation for Lemma~\ref{lem:C18a}}
    \label{fig:C18a}
  \end{figure}

  We have $|\hat{b}|=|\hat{g}|=2$ or $3$ (depending on whether $b$ and $g$ share an endpoint) and
  $|\hat{a}|=|\hat{c}|=|\hat{d}|=|\hat{e}|=|\hat{f}|=3$. Moreover, we
  have $|\hat{d}\cup\hat{e}\cup\hat{f}\cup\hat{g}|\geqslant 4$ since
  $d,e,f,g$ were properly colored.

  If $\hat{g}$ is not included in $\hat{d}$, we color $g$ with a color
  not in $\hat{d}$, then $b$, and apply Lemma~\ref{lem:fryingpan} to
  color $\{f,d,a,c,e\}$. Therefore, we may assume that
  $\hat{g}\subset\hat{d}$, and similarly, $\hat{g}\subset\hat{e}$.

  We may also assume that $\hat{g}\subset \hat{f}$. Indeed, otherwise,
  we color $g$ with a color not in $\hat{f}$, then forget $f$ and
  apply Lemma~\ref{lem:fryingpan} to $\{b,a,d,e,c\}$.

  Now, if $\hat{f}\not\subset\hat{d}$, we color $f$ with a color not
  in $\hat{d}$ (thus not in $\hat{g}$), then color $b$ and apply
  Lemma~\ref{lem:fryingpan} to $\{g,d,a,c,e\}$. If
  $\hat{d}\not\subset\hat{f}$, we color $d$ with a color not in
  $\hat{f}$, then color $a,b,c,e,g,f$. Therefore, we may assume that
  $\hat{f}=\hat{d}$ and similarly $\hat{f}=\hat{e}$.

  This ensures that 
  \[|\hat{f}|=|\hat{d}\cup\hat{e}\cup\hat{f}\cup\hat{g}|\geqslant 4.\]
  Now we forget $f$ and apply Theorem~\ref{thm:=deg} to color $\{a,b,c,d,e,g\}$.
\end{proof}

\begin{lemma}
  \label{lem:C18b}
  The graph $G$ does not contain $C_{\ref{C18}b}$.
\end{lemma}

\begin{proof}
  We use the notation depicted in Figure~\ref{fig:C18b}. By
  minimality, we color $G\setminus\{a,b,c\}$. We uncolor and forget
  $v_1,v_2,v_3$. Denote by $\alpha$ the color of $d$ and $\beta$ the
  color of $e$. We then uncolor $d,e,f,g$.

  \begin{figure}[!h]
    \centering
    \begin{tikzpicture}[v/.style={draw=black,minimum size = 10pt,ellipse,inner sep=1pt}]
      \node[v,label=85:{$u$}] (u) at (0,0)  {8};
      \node[v, very thick] (v0) at (45:1.5) {8};
      \node[v,label=above:{$v_1$}] (v1) at (90:1.5) {3};
      \node[v, very thick] (v2) at (135:1.5) {8};
      \node[v,label=below left:{$v_2$}] (v3) at (180:1.5) {4};
      \node[v, very thick,xshift=-1.5cm] (v3') at (180:1.5) {8};
      \node[v, very thick] (v4) at (225:1.5) {8};
      \node[v,label=right:{$v_3$}] (v6) at (315:1.5) {4};
      \node[v, very thick] (v5) at (270:1.5) {7};
      \node[v, very thick] (v7) at (0:1.5) {8};
      \draw (v2) -- (v1) node[midway,above] {$a$};
      \draw (v1) -- (v0) node[midway,above] {$b$};
      \draw (v1) -- (u) node[midway,left] {$c$};
      \draw (v3) -- (v3') node[midway,above] {$g$};
      \draw (v2) -- (v3) node[midway,left] {$d$};
      \draw (v3) -- (u) node[midway,below] {$e$};
      \draw (v3) -- (v4) node[midway,left] {$f$};
      \draw (u) -- (v0) node[midway,above] {$i$};
      \draw (u) -- (v2) node[midway,left] {$h$};
      \draw (u) -- (v4) node[midway,below] {$m$};
      \draw (u) -- (v7) node[midway,above] {$j$};
      \draw (u) -- (v6) node[midway,right] {$k$};
      \draw (u) -- (v5) node[midway,right] {$\ell$};
    \end{tikzpicture}
\caption{Notation for Lemma~\ref{lem:C18b}}
    \label{fig:C18b}
  \end{figure}

  We have $|\hat{f}|=2$, $|\hat{a}|=|\hat{c}|=|\hat{d}|=|\hat{e}|=3$ and $|\hat{b}|=|\hat{g}|\in\{2,3\}$ depending on whether they share an endpoint. Moreover, since $d,e,f,g$ are
  pairwise incident, they were colored with at least 4 colors before
  being uncolored. Therefore, we have
  $|\hat{d}\cup\hat{e}\cup\hat{f}\cup\hat{g}|\geqslant 4$.
  \begin{itemize}
  \item If $\hat{d}=\hat{e}$, then this means there exists
    $\gamma\in (\hat{f}\cup \hat{g})\setminus(\hat{d}\cup
    \hat{e})$. We color $f$ or $g$ with $\gamma$, then the remaining edge among $\{f,g\}$, and $b$. The
    elements $\{a,c,d,e\}$ induce an even cycle, which is
    $2$-choosable. We may thus assume that $\hat{d}\neq\hat{e}$.
  \item If $\hat{f}\neq \hat{g}$ (which happens when $g$ and $b$ share an endpoint), we color $f$ with
    $\gamma\notin \hat{g}$, and we observe that $\hat{d}\neq \hat{e}$ if
    $|\hat{d}|=|\hat{e}|=2$. We may then color $b$ and apply
    Lemma~\ref{lem:fryingpan} to $\{g,d,a,c,e\}$.
  \item If $|\hat{e}\setminus \hat{f}|$ and $|\hat{d}\setminus \hat{f}|$ are both at least $2$,
    then we color $f$ and $g$ and apply Lemma~\ref{lem:fryingpan} to
    $\{b,a,d,e,c\}$. 
  \item Observe that since $\hat{f}=\hat{g}$ and $d,e,f,g$ were
    colored, we have $\alpha,\beta\notin \hat{f}$. Therefore, by the previous item, either $\hat{d}=\{\alpha\}\cup \hat{f}$ and $d$ is forced to be colored $\alpha$ or $\hat{e}=\{\beta\}\cup\hat{f}$ and $e$ is forced to be colored $\beta$. 
    Moreover, if afterwards the remaining element among $\{d,e\}$ has two available colors, then for at least one choice we can extend the coloring to $a,b,c$. Therefore, $G$ is not colorable only if $d$ is forced to be colored $\alpha$, $e$ is forced to be colored $\beta$ and $\hat{a}\setminus\{\alpha\}=\hat{c}\setminus\{\beta\}=\hat{b}$.
  \item Let $\gamma$ be the color of $h$. Observe that
    $\gamma\in L(a)$. Indeed, since $h$ is adjacent to every colored
    element incident to $a$, we could otherwise assume that
    $|\hat{a}|=4$, hence forget $a,b,c$ and color
    $d,e,f,g$. Similarly, we have $\gamma\in L(c)$.
  \end{itemize}
  We show that we can recolor $h,i$ or $m$ and then extend the
  coloring with the new available colors of $a,b,c,d,e,f,g$.

  Let $H$ be the color shifting graph of $S=\{h,i,j,k,\ell, m,u\}$. If
  $\alpha$ does not appear on $S$, we remove the arc $s_\alpha\to h$ from $H$,
  otherwise $\alpha$ appears on $x\in S$, and we remove the arc
  $x\to h$ in $H$ (in both cases, this assumes that the arc is contained in the color shifting graph, that is if $\alpha\in \hat{h}$). We uncolor the vertices of $S$, and we have
  $|\hat{j}|=2$, $|\hat{i}|=|\hat{\ell}|=|\hat{m}|=3$, $|\hat{h}|=4$,
  $|\hat{u}|=5$ and $|\hat{k}|=7$. Observe that $d^-(h)=2$.

  By Lemma~\ref{lem:SCC}, there is a strong component $C$ of $H$ such
  that $|C|>\max_{x\in C} d^-(x)$. We consider two cases, depending on
  whether $C$ contains a vertex $s_\delta$.

  \begin{enumerate}
  \item Assume first that $C$ does not contain any vertex $s_\delta$.
    \begin{itemize}
    \item If $C$ contains $k$, then $|C|\geqslant 7$, hence $C$ contains
      $h,i$ or $m$.
    \item Otherwise, if $C$ contains $u$, then $|C|\geqslant 7$ hence $C$
      contains $h,i$ or $m$.
    \item Otherwise, $C\subset\{h,i,j,\ell,m\}$. If $C$ contains $\ell$,
      then $|C|\geqslant 3$ hence $C$ contains $h,i$ or $m$.
    \item Otherwise, $C\subset\{h,i,j,m\}$. If $C$ contains $j$, then
      $|C|\geqslant 2$ hence $C$ contains $h,i$ or $m$.
    \end{itemize}

    We may thus find a directed cycle in $H$ containing $h,i$ or $m$
    and only vertices of $S$. Shifting the colors along this cycle as
    done in Lemma~\ref{lem:recolor} yields another coloring of $S$
    obtained by permuting the colors. Denote by $\tilde{x}$ the new
    list of available colors for the element $x$ after the recoloring
    process.

    Observe that since we removed an ingoing arc to $h$, the edge $h$
    cannot be colored with $\alpha$ in the new coloring. This implies
    that $\alpha\in\tilde{d}$. Moreover, $\hat{e}=\tilde{e}$, hence
    $\beta\in \tilde{e}$. Note also that the recoloring process may have changed the lists of $f$ and $g$ (if $g$ was incident to $i$ or $j$). But we still have $\beta\notin \tilde{f}\cup\tilde{g}$ since $\beta$ did not appear on the recolored elements, and $\alpha$ cannot appear on both $\tilde{f}$ and $\tilde{g}$ since it would mean that $m$ and the edge among $\{i,j,\ell\}$ incident with $g$ were both colored with $\alpha$. 
    
    We consider three cases, depending on which elements among $\{h,i,m\}$ were recolored. 

    \begin{itemize}
    \item Assume that $h$ was recolored. We color $d$ with $\alpha$
      and $e$ with $\beta$. Since $h$ was recolored, then its former
      color $\gamma$ does not appear anymore on a colored incident
      element of $a$. Since $\gamma\in L(a)$, we can color $a$ with
      $\gamma$. Observe that this does not remove a color from $\tilde{c}$ since $\gamma\notin \hat{c}$. Therefore, after this, we may assume that $|\tilde{b}|=1$ and $|\tilde{c}|=2$ so we can color $b$ and $c$. Finally, by the above remark, we know that $f$ and $g$ can be colored.
      
    \item If $h$ was not recolored, then assume that $i$ was. In this
      case, we may still color $d$ with $\alpha$ and $e$ with $\beta$,
      then color again $f,g$. Since we recolored $i$, we now have
      $\tilde{b}\neq\tilde{a}$, hence we can color $a,b,c$.
    \item Finally, if we only recolored $m$, then let $\varepsilon$ be the former color of $m$. 
      
      Note that $\varepsilon\in L(f)$, otherwise, we could have assumed that $|\hat{f}|=3$, and obtained the same situation as in Lemma~\ref{lem:C18a}. In particular, $\varepsilon\in\tilde{f}\setminus\hat{f}$, and $\varepsilon\notin \hat{g}=\hat{f}$. Moreover, if $\varepsilon\in \tilde{g}$, it means that $g$ was incident with an element colored with $\varepsilon$ before the recoloring process, which is not possible. Therefore, we can color $f$ with $\varepsilon$, then $b$, and afterwards observe that all remaining elements have two available colors, and moreover $\tilde{d}\neq\tilde{e}$ (since these elements were not affected by the recoloring). We may thus apply Lemma~\ref{lem:fryingpan} to $\{g,d,a,c,e\}$.
      

    
    \end{itemize}
  \item Assume now that $C$ contains a vertex $s_\delta$, ensuring
    that $|C|>|V(H)|-1$, i.e. $H$ is strongly connected.

    Assume that $\gamma\notin L(e)$. Note that if $\gamma\notin L(d)$, then $|\hat{e}|=|\hat{d}|=4$, hence we may just color $f,g,b$ and conclude since $\{a,c,e,d\}$ induce an even cycle and every element has two available colors. Hence we assume that $\gamma\in L(d)$. We consider a directed path $s_\delta\to \cdots\to h$ in $H$ where
    each internal vertex lies in $S$. Since $h$ is colored with
    $\gamma$, $s_\gamma$ has no successor in $H$, hence we have
    $\delta\neq \gamma$. We shift the colors of $S$ along this path,
    as done in Lemma~\ref{lem:recolor}. We may then color $c$ and $d$
    with $\gamma$ (this is possible since $\gamma\in L(d)$ and
    $\gamma\neq\delta$). Note that $\gamma\notin\hat{d}\cup\hat{e}$,
    hence $\gamma\notin\hat{g}$, and $\gamma\notin\tilde{g}$ (since that would mean that $g$ was incident with an edge from $S$ colored with $\gamma$). Then we may assume that
    $|\tilde{f}|=1$, $|\tilde{a}|=|\tilde{b}|=|\tilde{g}|=2$ and
    $|\tilde{e}|=3$. We may thus color $f,g,e,b,a$.

    If $\gamma\in L(e)$, we again shift the colors along the directed
    path $s_\delta\to \cdots \to h$, then color $a,e$ with $\gamma$
    and $d$ with $\alpha$ (since the new color of $h$ is not
    $\alpha$). Again, we have $\gamma\notin\tilde{g}$, hence we may
    assume that $|\tilde{b}|=|\tilde{f}|=1$ and
    $|\tilde{c}|=|\tilde{g}|=2$. We may thus color
    $f,g,b,c$.
  \end{enumerate}
\end{proof}

\begin{lemma}
  \label{lem:C18c}
  The graph $G$ does not contain $C_{\ref{C18}c}$.
\end{lemma}

\begin{proof}
  We use the notation depicted in Figure~\ref{fig:C18c}. By
  minimality, we color $G\setminus\{a,b,c\}$. We uncolor and forget
  $v_1,v_2,v_3$. Let $\alpha,\beta$ be the colors of $e$ and $f$. We
  then uncolor $d,e,f,g$.
  \begin{figure}[!h]
    \centering
    \begin{tikzpicture}[v/.style={draw=black,minimum size = 10pt,ellipse,inner sep=1pt}]
      \node[v,label=85:{$u$}] (u) at (0,0)  {8};
      \node[v, very thick] (v0) at (45:1.5) {8};
      \node[v,label=above:{$v_1$}] (v1) at (90:1.5) {3};
      \node[v, very thick] (v2) at (135:1.5) {8};
      \node[v, very thick] (v3) at (180:1.5) {7};
      \node[v,xshift=-1.5cm, very thick] (v4') at (225:1.5) {8};
      \node[v,label=below:{$v_2$}] (v4) at (225:1.5) {4};
      \node[v,label=right:{$v_3$}] (v6) at (315:1.5) {4};
      \node[v, very thick] (v5) at (270:1.5) {8};
      \node[v, very thick] (v7) at (0:1.5) {8};
      \draw (v2) -- (v1) node[midway,above] {$a$};
      \draw (v1) -- (v0) node[midway,above] {$b$};
      \draw (v1) -- (u) node[midway,left] {$c$};

      \draw (v4) -- (v4') node[midway,below] {$e$};
      \draw (v4) -- (v3) node[midway,left] {$d$};
      \draw (v4) -- (v5) node[midway,below] {$f$};
      \draw (u) -- (v4) node[midway,below] {$g$};
      \draw (u) -- (v2) node[midway,left] {$h$};
      \draw (u) -- (v0) node[midway,above] {$i$};
      \draw (u) -- (v7) node[midway,above] {$j$};
      \draw (u) -- (v6) node[midway,right] {$k$};
      \draw (u) -- (v5) node[midway,right] {$\ell$};

      \draw (u) -- (v3) node[midway,below] {$m$};
    \end{tikzpicture}
\caption{Notation for Lemma~\ref{lem:C18c}}
    \label{fig:C18c}
  \end{figure}
  We have $|\hat{a}|=|\hat{f}|=2$,
  $|\hat{c}|=|\hat{d}|=|\hat{g}|=3$ and $|\hat{b}|=|\hat{e}|\in\{2,3\}$ depending on whether $b$ and $e$ share an endpoint. Moreover, since $d,e,f,g$ are
  pairwise incident, they were colored with at least 4 colors before
  being uncolored. Therefore, we have
  $|\hat{d}\cup\hat{e}\cup\hat{f}\cup\hat{g}|\geqslant 4$.

  \begin{itemize}
  \item If $\hat{d}=\hat{g}$, then this means there exists
    $\gamma\in (\hat{e}\cup \hat{f})\setminus(\hat{d}\cup
    \hat{g})$. We color $e$ or $f$ with $\gamma$, then
    $f,a,b,c,g,d$. We may thus assume that $\hat{d}\neq\hat{g}$.

  \item We also assume that $\hat{a}=\hat{b}$. Indeed, otherwise, we
    color $a$ with a color not in $\hat{b}$, then forget $b$, $c$, and
    put back the initial colors on $d,e,f,g$.

  \item Let $\gamma$ be the color of $m$. Observe that
    $\gamma\in L(c)$. Indeed, since $m$ is incident to $c$, we could otherwise assume that
    $|\hat{c}|=4$, hence forget $c,b,a$ and color $d,e,f,g$.

  \item We now show that we can recolor $h$ or $i$. Let $H$ be the
    color shifting graph of $S=\{h,i,j,k,\ell, m,u\}$. We uncolor the
    vertices of $S$, and we have $|\hat{j}|\geqslant 2$,
    $|\hat{h}|,|\hat{i}|,|\hat{\ell}|\geqslant 3$, $|\hat{m}|= 4$, $|\hat{u}|=5$
    and $|\hat{k}|=7$ (the exact values depend on whether $e$ is incident with $h,i$ or $j$).  By Lemma~\ref{lem:SCC}, there is a strong
    component $C$ of $H$ such that $|C|>\max_{x\in C} d^-(x)$.
    \begin{itemize}
    \item If $C$ contains a vertex $s_\delta$, then $|C|=|V(H)|$, hence
      $C$ contains $h$ and $i$.
    \item If $C$ contains $k$, then $|C|\geqslant 7$, hence $C$
      contains $h$ or $i$.
    \item Otherwise, if $C$ contains $u$, then $|C|\geqslant 5$ hence
      $C$ contains $h$ or $i$.
    \item Otherwise, $C\subset\{h,i,j,\ell,m\}$. If $C$ contains $m$,
      then $|C|\geqslant 4$ hence $C$ contains $h$ or $i$.
    \item Otherwise, $C\subset\{h,i,j,\ell\}$. If $C$ contains $\ell$,
      then $|C|\geqslant 3$ hence $C$ contains $h$ or $i$.
    \item Otherwise, $C\subset\{h,i,j\}$. If $C$ contains $j$, then
      $|C|\geqslant 2$ hence $C$ contains $h$ or $i$.
    \end{itemize} In every case, we can recolor $h$ or $i$ by
    Lemma~\ref{lem:recolor}. This allows to obtain new lists
    $\tilde{x}$ of available colors for the element $x$, and we have
    $\tilde{a}\neq\tilde{b}$. However, this may break colorability of
    $d,e,f,g$. Proving that this colorability is actually preserved
    requires a more careful analysis we give in the rest of the proof.  Let $\delta,\varepsilon$ be the colors of $h$ and $i$ before recoloring.

  \item Observe that $\gamma\in L(d)$ and $L(g)$ contains either $\delta$ or $\varepsilon$.
  Indeed, otherwise, we have $|\hat{d}|=4$ or $|\hat{g}|=5$ 
  at the beginning, and we can color $a,b,c,g,e,f,d$.
  

    We may thus assume that $\gamma\in L(d)$ and that (by symmetry) $\delta\in L(g)$.  
  \item Consider the strong component $C$ of $H$ given by
    Lemma~\ref{lem:SCC}. We consider two cases, depending on whether
    $C$ contains a vertex $s_\zeta$.
    \begin{itemize}
    \item Assume first that $C$ does not contain any vertex
      $s_\zeta$. As we saw, $C$ contains $h$ or $i$, hence we may
      find a directed cycle in $H$ containing $h$ or $i$ and only
      vertices of $S$. 

      The coloring given by applying Lemma~\ref{lem:recolor} to this
      directed cycle uses the same set of colors (the colors are only
      permuted). Therefore, we have $\tilde{c}=\hat{c}$ and
      $\tilde{g}=\hat{g}$, together with $\tilde{a}\neq\tilde{b}$.
      
      If $m$ was not recolored, then we also have
      $\tilde{d}=\hat{d}$. Hence $\tilde{d}$ and $\tilde{g}$ are
      different lists of size at least 3.  Otherwise, if $m$ was
      recolored, then its former color $\gamma$ lies now in
      $\tilde{d}$, but not in $\tilde{g}$ (since color $\gamma$ is
      still present on $S$).

      Therefore, in both cases, we can hence we can color $e,f$
      arbitrarily (if $e$ is incident with $b$, we instead color $e$ so that $a$ and $b$ get different lists of colors afterwards), then $d,g$ since $\tilde{d}\neq\tilde{g}$, then
      $c$, and $a,b$ since $\tilde{a}\neq\tilde{b}$.

    \item Assume now that $C$ contains a vertex $s_\zeta$, ensuring
      that $|C|>|V(H)|-1$, i.e. $H$ is strongly connected. 

      We consider a directed path $s_\zeta\to \cdots\to h$ in $H$
      where each internal vertex lies in $S$. Since $m$ is colored
      with $\gamma$, $s_\gamma$ has no successor in $H$, hence we have
      $\zeta\neq \gamma$. We shift the colors of $S$ along this path,
      as done in Lemma~\ref{lem:recolor}, so that
      $\tilde{a}\neq\tilde{b}$.

      Assume that this path goes through $m$. Then since $m$ is not
      the last vertex of the path, the color $\gamma$ is still present
      on some element of $S$, hence $\gamma\notin\tilde{g}$. However,
      we have $\gamma\in \tilde{d}$ since $\gamma\in L(d)$ and $m$ is
      adjacent to every colored element incident to $d$. Therefore, we
      have $\tilde{d}\neq\tilde{g}$. We can then color
      $e,f,d,g,c,a,b$.

    \item Assume that the path does not go through $m$. Therefore, the
      color of $m$ is still $\gamma$ after the recoloring. Since the
      initial color $\delta$ of $h$ lies in $L(g)$, we now have
      $\tilde{g}=(\hat{g}\setminus\{\zeta\})\cup \{\delta\}$. If
      $\tilde{g}\neq\tilde{d}$, then we can color $a$ with a color not
      in $\tilde{b}$, then forget $b,c$ and color
      $e,f,d,g$. Otherwise, we have
      \[(\hat{g}\setminus\{\zeta\})\cup
        \{\delta\}=\tilde{g}=\tilde{d}=\hat{d}.\] 
    In particular $\delta\in\hat{d}$. 

      We then apply the same argument to a path
      $s_\eta\to\cdots \to i$. Such a path cannot contain $m$ by the
      previous item, and we can extend the new coloring to $G$ unless $\hat{d}=\hat{g}\setminus\{\eta\}$ or $\hat{d}=(\hat{g}\setminus\{\eta\})\cup \{\varepsilon\}$ (depending on whether $\varepsilon\in L(g)$). However, none of these can happen since $\delta\in \hat{d}$ but $\delta\neq\varepsilon$ and $\delta\notin \hat{g}$.  Therefore, we can extend the new coloring to $G$.
    \end{itemize}
  \end{itemize}
\end{proof}

\subsection{Configuration $C_{\ref{C19}}$}
By definition of $C_{\ref{C19}}$, if $G$ contains $C_{\ref{C19}}$,
then $G$ contains one of the following:
\begin{itemize}
\item[$\bullet$] $C_{\ref{C19}a}$: $u$ has three semi-weak neighbors
  $v_1,v_2,v_3$ of degree $3$ and a neighbor $v_4$ of degree 7.
\item[$\bullet$] $C_{\ref{C19}b}$: $u$ has two semi-weak neighbors
  $v_1,v_2$ of degree $3$, two neighbors $w_1,w_2$ of degree $4$ and a
  neighbor $w_3$ of degree 7.
\end{itemize}
We dedicate a lemma to each of these configurations.

\begin{lemma}
  \label{lem:C19a}
  The graph $G$ does not contain $C_{\ref{C19}a}$.
\end{lemma}

\begin{proof}
  We use the notation depicted in Figure~\ref{fig:C19a}. By minimality,
  we color $G\setminus\{a,\ldots,i\}$ and uncolor $v_1,v_2,v_3$.
  \begin{figure}[!h]
    \centering
    \begin{tikzpicture}[v/.style={draw=black,minimum size = 10pt,ellipse,inner sep=1pt}]
      \node[v,label=10:{$u$}] (u) at (0,0)  {8};
      \node[v,label=above:{$v_1$}] (v1) at (45:1.5) {3};
      \node[v] (v0) at (90:1.5) {8};
      \node[v,xshift=1cm, very thick] (v2) at (45:1.5) {8};
      \node[v,label=left:{$v_4$}, very thick] (v4) at (135:1.5) {7};
      \node[v, very thick] (v3) at (180:1.5) {8};
      \node[v, very thick] (v6) at (0:1.5) {8};
      \node[v, very thick,xshift=1cm] (v5) at (315:1.5) {8};
      \node[v, very thick] (v8) at (270:1.5) {8};
      \node[v,label=below:{$v_3$}] (v9) at (225:1.5) {3};
      \node[v, very thick,xshift=-1cm] (v10) at (225:1.5) {8};
      \node[v,label=below:{$v_2$}] (v7) at (315:1.5) {3};
      \draw (v6) -- (u) node[midway,below] {$k$};
      \draw (v4) -- (u) node[midway,above] {$n$};
      \draw (v1) -- (u) node[midway,right] {$c$};
      \draw (v0) -- (u) node[midway,right] {$j$};
      \draw (v7) -- (u) node[midway,below] {$f$};
      \draw (v5) -- (v7) node[midway,below] {$d$};
      \draw (v6) -- (v7) node[midway,right] {$e$};
      \draw (v2) -- (v1) node[midway,above] {$a$};
      \draw (v1) -- (v0) node[midway,above] {$b$};
      \draw (u) -- (v3) node[midway,above] {$m$};
      \draw (u) -- (v9) node[midway,left] {$i$};
      \draw (u) -- (v8) node[midway,left] {$\ell$};
      \draw (v9) -- (v8) node[midway,below] {$h$};
      \draw (v9) -- (v10) node[midway,below] {$g$};
    \end{tikzpicture}
\caption{Notation for Lemma~\ref{lem:C19a}}
    \label{fig:C19a}
  \end{figure}
  Note that there may not be any adjacency between the uncolored elements other than the ones depicted in Figure~\ref{fig:C19a}, otherwise we could apply Theorem~\ref{thm:=deg} to color them. Hence we have $|\hat{a}|=|\hat{b}|=|\hat{d}|=|\hat{e}|=|\hat{g}|=|\hat{h}|=2$,
  $|\hat{c}|=|\hat{f}|=|\hat{i}|=4$ and
  $|\hat{v_1}|=|\hat{v_2}|=|\hat{v_3}|=7$. We forget $v_1,v_2,v_3$. We
  assume that $\hat{a}=\hat{b}$, $\hat{d}=\hat{e}$, $\hat{g}=\hat{h}$,
  and
  $\hat{c}\setminus\hat{a}=\hat{f}\setminus\hat{d}=\hat{i}\setminus\hat{g}$
  have size $2$ (otherwise, we can already extend the coloring to
  $G$). Note that any recoloring of $j,k$ or $\ell$ is sufficient to
  ensure that this hypothesis does no longer hold (even if $a,d,g$ are incident with $m$ or $n$). We uncolor
  $j,k,\ell,m,n,u$. We may assume that $|\hat{m}|=2$,
  $|\hat{n}|=|\hat{j}|=|\hat{k}|=|\hat{\ell}|=3$ and $|\hat{u}|=5$.

  Denote by $H$ the color shifting graph of $S=\{j,k,\ell,m,n,u\}$. By
  Lemma~\ref{lem:SCC}, there exists a strong component $C$ of $H$ such
  that $|C|>\max_{x\in C} d^-(x)$. Note that this inequality ensures
  that $|C|>1$. We show that $C$ contains $j,k$ or $\ell$ by
  distinguishing four cases:
  \begin{enumerate}
  \item If $C$ contains a vertex $s_\alpha$, then we have
    $|C|>d^-(s_\alpha)=|V(H)|-1$. Therefore, $C=V(H)$, and $C$
    contains $j,k$ or $\ell$.
  \item Otherwise, if $C$ contains $u$, then it has size at least $5$,
    hence it also contains $j,k$ or $\ell$.
  \item Otherwise, if $C$ contains $n$, then its size is at least $3$,
    hence it also contains $j,k$ or $\ell$.
  \item Otherwise, $C\subset\{j,k,\ell,m\}$. Then, if $C$ contains
    $m$, its size is at least $2$, hence it also contains $j,k$ or
    $\ell$.
  \end{enumerate}
  We thus obtain that $C$ is a strong component of size at least $2$
  that contains $j,k$ or $\ell$. Therefore, there is a directed cycle
  containing one of these vertices. Thus, we can apply
  Lemma~\ref{lem:recolor} to ensure that the starting hypothesis does
  not hold anymore, hence we can extend the coloring to $G$.
\end{proof}

\begin{lemma}
  \label{lem:C19b}
  The graph $G$ does not contain $C_{\ref{C19}b}$.
\end{lemma}

\begin{proof}
  We use the notation depicted in Figure~\ref{fig:C19b}. By
  minimality, we color $G\setminus\{a,\ldots,f\}$ and uncolor $v_1$
  and $v_2$.
  \begin{figure}[!h]
    \centering
   \begin{tikzpicture}[v/.style={draw=black,minimum size = 10pt,ellipse,inner sep=1pt}]
      \node[v,label=above left:{$u$}] (u) at (0,0)  {8};
      \node[v,label=above:{$v_1$}] (v1) at (45:1.5) {3};
      \node[v, very thick] (v0) at (90:1.5) {8};
      \node[v,xshift=1cm, very thick] (v2) at (45:1.5) {8};
      \node[v,label=above:{$w_3$},very thick] (v3) at (180:1.5) {7};
      \node[v,label=above:{$v_2$}] (v6) at (0:1.5) {3};
      \node[v, very thick] (v5) at (0:2.5) {8};
      \node[v,label=left:{$w_1$}] (v8) at (270:1.5) {4};
      \node[v,label=left:{$w_2$}] (v9) at (225:1.5) {4};
      \node[v, very thick] (v7) at (315:1.5) {8};
      \draw (v7) -- (u) node[midway,below] {$h$};
      \draw (v1) -- (u) node[midway,right] {$c$};
      \draw (v0) -- (u) node[midway,right] {$g$};
      \draw (v6) -- (u) node[midway,below] {$d$};
      \draw (v5) -- (v6) node[midway,below] {$e$};
      \draw (v6) -- (v7) node[midway,right] {$f$};
      \draw (v2) -- (v1) node[midway,above] {$a$};
      \draw (v1) -- (v0) node[midway,above] {$b$};
      \draw (u) -- (v3) node[midway,above] {$k$};
      \draw (u) -- (v9) node[midway,left] {$j$};
      \draw (u) -- (v8) node[midway,left] {$i$};
    \end{tikzpicture}
\caption{Notation for Lemma~\ref{lem:C19b}}
    \label{fig:C19b}
  \end{figure}
 As in Lemma~\ref{lem:C16}, we have $|\hat{a}|=|\hat{b}|=|\hat{e}|=|\hat{f}|=2$,
  $|\hat{c}|=|\hat{d}|=3$ and $|\hat{v_1}|=|\hat{v_2}|=7$, we forget
  $v_1,v_2$ and we assume that $\hat{a}=\hat{b},\hat{e}=\hat{f}$ and
  $\hat{c}\setminus\hat{a}=\hat{d}\setminus\hat{e}=\{\alpha\}$
  (otherwise, we can already extend the coloring to $G$). Note that
  any recoloring of $g$ or $h$ is sufficient to ensure that this
  hypothesis does no longer hold (even if $a$ or $e$ is incident with $g,h,i,j$ or $k$). We uncolor $g,h,i,j,k,u,w_1,w_2$,
  then forget $w_1,w_2$.

  We may assume that $|\hat{g}|=|\hat{h}|=|\hat{k}|=2$, $|\hat{u}|=5$,
  and $|\hat{i}|=|\hat{j}|=6$. Denote by $H$ the color shifting graph
  of $S=\{g,h,i,j,k,u\}$. By Lemma~\ref{lem:SCC}, there exists a
  strong component $C$ of $H$ such that $|C|>\max_{x\in C}
  d^-(x)$. Note that this inequality ensures that $|C|>1$. We show
  that $C$ contains $g$ or $h$ by distinguishing three cases:
  \begin{enumerate}
  \item If $C$ contains a vertex $s_\alpha$, then we have
    $|C|>d^-(s_\alpha)=|V(H)|-1$. Therefore, $C=V(H)$, and $C$
    contains $g,h$.
  \item Otherwise, if $C$ contains $u,i$ or $j$, then it has size at least $5$,
    hence it also contains $g$ or $h$.
  \item Otherwise, $C\subset\{g,h,k\}$. If $C$ contains $k$, then its
    size is at least $2$, hence it also contains $g$ or $h$.
  \end{enumerate}
  We thus obtain that $C$ is a strong component of size at least $2$
  that contains $g$ or $h$. Therefore, there is a directed cycle
  containing one of these vertices. Thus, we can apply
  Lemma~\ref{lem:recolor} to ensure that we can extend the coloring to
  $G$.
\end{proof}

\section{Discharging argument}
\label{sec:positive}
In this section, we present the rules and check that every element of
$G$ has non-negative final weight after the discharging procedure.

\subsection{Discharging rules}
\label{sec:rules}
We start with the definition of the initial weighting $\omega$: we set
$\omega(v)=d(v)-6$ and $\omega(f)=2\ell(f)-6$ for each vertex $v$ and
face $f$. Using Euler's formula, the total weight is $-12$.

We then introduce several discharging rules, see
Figure~\ref{fig:rules}:
\begin{itemize}
\item[$\bullet$] For any $4^+$-face $f$,
  \begin{itemize}
  \item[($\regle{R1}$)] If $f$ is incident to a $5^-$-vertex $u$, then $f$
    gives 1 to $u$.
  \item[($\regle{R2}$)] If $f$ has a vertex $v$ such that $d(v)=8$ and
    the neighbors $u,w$ of $v$ along $f$ satisfy $d(u)=3$ and
    $d(w)\geqslant 6$, then $f$ gives $\frac{5}{12}$ to $v$.
  \item[($\regle{R3}$)] If $f$ has a vertex $v$ such that $d(v)=7$ and
    the neighbors $u,w$ of $v$ along $f$ both have degree at least
    $6$, then $f$ gives $\frac{1}{3}$ to $v$ if $d(u)=6$ or $d(v)=6$,
    and $\frac{1}{12}$ otherwise.
  \item[($\regle{R3+}$)] If $f$ has a vertex $v$ such that $d(v)=7$
    and the neighbors $u,w$ of $v$ along $f$ have degree 5 and 6 respectively,
    then $f$ gives $\frac{1}{6}$ to $v$, except if $\ell(f)=4$ and the
    last vertex of $f$ has degree 5.

  \end{itemize}
\item[$\bullet$] For any $8$-vertex $u$,
  \begin{itemize}
  \item[($\regle{R4}$)] If $u$ has a weak neighbor $v$ of degree $3$,
    then $u$ gives 1 to $v$.
  \item[($\regle{R5}$)] If $u$ has a semi-weak neighbor $v$ of degree
    $3$, then $u$ gives $\frac{1}{2}$ to $v$.
  \item[($\regle{R6}$)] If $u$ has a $(p,q)$-neighbor $v$ of degree $4$,
    then $u$ gives $\omega$ to $v$ where:
    \[\omega=\begin{cases} \frac{2}{3}&\text{ if }p=q=7,\\
        \frac{7}{12}&\text{ if } p=7 \text{ and }q=8,\\
        \frac{1}{2}&\text{ if } p=q=8,\\
        0&\text{ otherwise.}
      \end{cases}
    \]
  \item[($\regle{R7}$)] If $u$ has a semi-weak neighbor $v$ of degree
    $4$ and a neighbor $w$ of degree $7$ such that $uvw$ is a
    triangular face, then $u$ gives $\frac{1}{12}$ to $v$.
  \item[($\regle{R8}$)] If $u$ has a $(p,q)$-neighbor $v$ of degree $5$
    such that $p,q\geqslant 5$, then $u$ gives $\omega$ to $v$ where
    \[\omega=\begin{cases} 
        \frac{1}{2}&\text{ if } p=5 \text{ and }q=6,\\
        \frac{1}{6}&\text{ if } p=5 \text{ and }q>6,\\
        \frac{2}{3}&\text{ if } p=q=6,\\
        \frac{1}{3}&\text{ if } v \text{ is an } E_3\text{-neighbor,}\\
\frac{1}{4}&\text{otherwise.}
      \end{cases}
    \]
  \end{itemize}
\item[$\bullet$] For any $7$-vertex $u$,
  \begin{itemize}
  \item[($\regle{R9}$)] If $u$ has a $(p,q)$-neighbor $v$ of degree
    $4$, then $u$ gives $\omega$ to $v$ where
    \[\omega=\begin{cases} 
        \frac{1}{2}&\text{ if } p=q=7,\\
        \frac{5}{12}&\text{ if } p=7 \text{ and }q=8,\\ 
        \frac{1}{3}&\text{ if } p=q=8,\\
        0&\text{ otherwise.}
      \end{cases}
    \]
  \item[($\regle{R10}$)] If $u$ has a weak neighbor $v$ of
  degree $5$, then $u$ gives $\omega$ to $v$ where:
  \[\omega=\begin{cases} \frac12&\text{ if $v$ is a $(5,6)$-neighbor of $u$,}\\
  \frac13&\text{ if $v$ is an $S_3$-neighbor of $u$,}\\
  \frac15&\text{ if $v$ is an $S_5$-neighbor of $u$,}\\
  \frac16&\text{ otherwise.}
  \end{cases}\]
  \end{itemize}
\end{itemize}

\begin{figure}[!ht]
  \centering
  \begin{tikzpicture}[every node/.style={draw=black,minimum size = 10pt,ellipse,inner sep=1pt}]
    \node (v1) at (270:1) {};
    \node (v2) at (0:1) {$5^-$};
    \node (v3) at (90:1) {};
    \node (v4) at (180:1) {};
    \node[draw=none] (v5) at (0,0){$f$};
    \draw (v1) -- (v2) -- (v3) -- (v4);
    \draw[dotted] (v4) -- (v1);
    \path[draw,->] (v5) -- node [draw=none,above]{1} (v2);
    \node[draw=none] at (-1,-1) {$R_{\ref{R1}}$};
    \tikzset{xshift=3cm}
    \node (v1) at (270:1) {3};
    \node (v2) at (0:1) {8};
    \node (v3) at (90:1) {$6^+$};
    \node (v4) at (180:1) {};
    \node[draw=none] (v5) at (0,0){$f$};
    \draw (v1) -- (v2) -- (v3) -- (v4);
    \draw[dotted] (v4) -- (v1);
    \path[draw,->] (v5) -- node [draw=none,above]{\tiny $\frac{5}{12}$} (v2);
    \node[draw=none] at (-1,-1) {$R_{\ref{R2}}$};
    \tikzset{xshift=3cm}
    \node (v1) at (270:1) {$6^+$};
    \node (v2) at (0:1) {7};
    \node (v3) at (90:1) {$6^+$};
    \node (v4) at (180:1) {};
    \node[draw=none] (v5) at (0,0){$f$};
    \draw (v1) -- (v2) -- (v3) -- (v4);
    \draw[dotted] (v4) -- (v1);
    \path[draw,->] (v5) -- node [draw=none,above]{$\omega$} (v2);
    \node[draw=none] at (-1,-1) {$R_{\ref{R3}}$};
    \tikzset{xshift=3cm}
    \node (v1) at (270:1) {$5$};
    \node (v2) at (0:1) {7};
    \node (v3) at (90:1) {$6$};
    \node (v4) at (180:1) {};
    \node[draw=none] (v5) at (0,0){$f$};
    \draw (v1) -- (v2) -- (v3) -- (v4);
    \draw[dotted] (v4) -- (v1);
    \path[draw,->] (v5) -- node [draw=none,above]{\tiny $\frac{1}{6}$} (v2);
    \node[draw=none] at (-1,-1) {$R_{\ref{R3+}}$};
    \tikzset{xshift=-10.5cm,yshift=-3cm}
    \node (v1) at (0,0) {8};
    \node (v2) at (1,0) {3};
    \node (v3) at (0.5,0.866) {};
    \node (v4) at (0.5,-0.866) {};
    \draw(v1) -- (v3) -- (v2) -- (v1) -- (v4) -- (v2);
    \path[draw,->, bend left] (v1) to node [draw=none,above]{1} (v2);
    \node[draw=none] (v5) at (0,-1) {$R_{\ref{R4}}$};
    \tikzset{xshift=1.75cm}
    \node (v1) at (0,0) {8};
    \node (v2) at (1,0) {3};
    \node (v3) at (0,1) {};
    \node (v33) at (1,1) {};
    \node (v4) at (0.5,-0.866) {};
    \draw(v1) -- (v3);
    \draw (v33) -- (v2) -- (v1) -- (v4) -- (v2);
    \path[draw,->, bend left] (v1) to node [draw=none,above]{\small $\frac{1}{2}$} (v2);
    \node[draw=none] (v5) at (0,-1) {$R_{\ref{R5}}$};
    \tikzset{xshift=1.75cm}
    \node (v1) at (0,0) {8};
    \node (v2) at (1,0) {4};
    \node (v3) at (0.5,0.866) {$p$};
    \node (v4) at (0.5,-0.866) {$q$};
    \draw(v1) -- (v3) -- (v2) -- (v1) -- (v4) -- (v2);
    \path[draw,->, bend left] (v1) to node [draw=none,above]{$\omega$} (v2);
    \node[draw=none] (v5) at (0,-1) {$R_{\ref{R6}}$};
    \tikzset{xshift=1.75cm}
    \node (v1) at (0,0) {8};
    \node (v2) at (1,0) {4};
    \node (v3) at (0,1) {};
    \node (v33) at (1,1) {};
    \node (v4) at (0.5,-0.866) {7};
    \draw(v1) -- (v3);
    \draw (v33) -- (v2) -- (v1) -- (v4) -- (v2);
    \path[draw,->, bend left] (v1) to node [draw=none,above]{\small $\frac{1}{12}$} (v2);
    \node[draw=none] (v5) at (0,-1) {$R_{\ref{R7}}$};
    \tikzset{xshift=1.75cm}
    \node (v1) at (0,0) {8};
    \node[label=right:{\scriptsize weak}] (v2) at (1,0) {5};
    \node (v3) at (0.5,0.866) {$p$};
    \node (v4) at (0.5,-0.866) {$q$};
    \draw(v1) -- (v3) -- (v2) -- (v1) -- (v4) -- (v2);
    \path[draw,->, bend left] (v1) to node [draw=none,above]{$\omega$} (v2);
    \node[draw=none] (v5) at (0,-1) {$R_{\ref{R8}}$};
    \tikzset{xshift=2.25cm}
    \node (v1) at (0,0) {7};
    \node (v2) at (1,0) {4};
    \node (v3) at (0.5,0.866) {$p$};
    \node (v4) at (0.5,-0.866) {$q$};
    \draw(v1) -- (v3) -- (v2) -- (v1) -- (v4) -- (v2);
    \path[draw,->, bend left] (v1) to node [draw=none,above]{$\omega$} (v2);
    \node[draw=none] (v5) at (0,-.5) {$R_{\ref{R9}}$};
    \tikzset{xshift=1.75cm}
    \node (v1) at (0,0) {7};
    \node[label=right:{\scriptsize weak}] (v2) at (1,0) {5};
    \draw(v1) -- (v2);
    \path[draw,->, bend left] (v1) to node [draw=none,above]{$\omega$} (v2);
    \node[draw=none] (v5) at (0.5,-1) {$R_{\ref{R10}}$};
  \end{tikzpicture}
  \caption{The discharging rules}
  \label{fig:rules}
\end{figure}

We now show that every element ends up with non-negative weight. Since
the sum of initial weights is -12, this yields the expected
contradiction and proves Theorem~\ref{thm:main}. We first handle the
faces, and then distinguish several cases for the vertices, depending
on their degree. First note that due to $C_{\ref{C1}}$, the minimum
degree of $G$ is $3$. Moreover, only vertices of degree $7$ or $8$
lose weight. For these vertices, we present here an extensive case
analysis argument. Since the proof is quite long, we also provide two
programs at
\href{https://github.com/tpierron/Delta8}{\url{https://github.com/tpierron/Delta8}}. These
programs generate all possible neighborhoods of a $7^+$-vertex and
check that either there is a reducible configuration or the final
weight is non-negative.

\subsection{Faces}
Note that only faces of length at least $4$ lose weight. Consider a
$4^+$-face $f$. We distinguish some cases, depending on its length
$\ell$ and the minimal degree $\delta$ of its incident vertices.

\begin{enumerate}
\item $\ell\geqslant 6$: By rules
  $R_{\ref{R1}},R_{\ref{R2}},R_{\ref{R3}}$ and $R_{\ref{R3+}}$, the
  face $f$ loses at most 1 for each of its vertices, hence
  \[\omega'(f)\geqslant 2\ell-6-\ell=\ell-6\geqslant 0\]
\item $\delta=3$ and $\ell=4$: Let $f=uu_1vu_2$ where $d(u)=3$. By
  $C_{\ref{C1}}$, both $u_1$ and $u_2$ are $8$-vertices. Consider the
  other neighbor $v$ of these $8$-vertices. If $v$ is a $6$-vertex or
  a $8$-vertex, then $f$ loses $2\times \frac{5}{12}$ on $u_1$ and
  $u_2$ by $R_{\ref{R2}}$ and $f$ does not lose anything on $v$ since
  $R_{\ref{R2}},R_{\ref{R3}}$ and $R_{\ref{R3+}}$ do not apply.

  If $v$ is a $7$-vertex, then $f$ loses $2\*\frac{5}{12}$ on $u_1$
  and $u_2$ by $R_{\ref{R2}}$ and $\frac{1}{12}$ on $v$ by
  $R_{\ref{R3}}$.

  Otherwise, $v$ is a $5^-$-vertex and $f$ loses $1$ on $v$ by
  $R_{\ref{R1}}$ but nothing on $u_1,u_2$. Thus the final weight of
  $f$ is at least $2-1-1=0$.
\item $\delta=3$ and $\ell=5$: Let $f=uu_1v_1v_2u_2$ where
  $d(u)=3$. By $C_{\ref{C1}}$, we have $d(u_1)=d(u_2)=8$. By
  $R_{\ref{R2}}$, the vertices $u_1$ and $u_2$ receive at most
  $\frac{5}{12}$. The three remaining vertices receive at most $1$ by
  $R_{\ref{R1}},R_{\ref{R2}},R_{\ref{R3}}$ and
  $R_{\ref{R3+}}$. Therefore, the final weight of $f$ is at least
  $4-3\* 1-2\* \frac{5}{12}=\frac 1 6> 0$.
\item $\delta=4$: By $C_{\ref{C1}}$, any $4$-vertex is adjacent to
  $7^+$-vertices. These vertices do not receive any weight from
  $f$. Therefore, $f$ loses at most $(\ell-2)\times 1$ by
  $R_{\ref{R1}},R_{\ref{R2}},R_{\ref{R3}}$ and $R_{\ref{R3+}}$. Hence
  $\omega'(f)=2\ell-6-(\ell-2)=\ell-4\geqslant 0$.
\item $\delta=5$ and $\ell=4$: If there is only one $5$-vertex $u$,
  then $f$ gives 1 to $u$ and at most $3\times \frac{1}{3}$ to the
  other vertices by $R_{\ref{R1}}$, $R_{\ref{R3}}$ and
  $R_{\ref{R3+}}$.

  Otherwise, by $C_{\ref{C3a}}$, there are two $5$-vertices and the
  two other vertices have degree at least $6$. Thus only
  $R_{\ref{R1}}$ applies, and $f$ loses $2$, giving a final weight of
  $2-2=0$.
\item $\delta=5$ and $\ell=5$. By $C_{\ref{C3a}}$, there are at most
  three $5$-vertices. If there are three such vertices, then
  $R_{\ref{R3}}$ and $R_{\ref{R3+}}$ do not apply and the final weight
  of $f$ is $4-3\* 1=1>0$. If $f$ has two vertices of degree $5$, $f$
  gives $2\* 1$ to these vertices by $R_{\ref{R1}}$ and at most
  $3\times\frac{1}{3}$ to the others by $R_{\ref{R3}}$ and
  $R_{\ref{R3+}}$. Therefore, the final weight is at least
  $4-2\*1-3\*\frac{1}{3}=1>0$.
\item $\delta>5$: Only $R_{\ref{R3}}$ applies, so $f$ loses at most
  $\ell\times\frac{1}{3}$. The final weight is
  $2\ell-6-\frac{\ell}{3}>0$ since $\ell \geqslant 4$.
\end{enumerate}

\subsection{$3$-vertices}
Let $u$ be a $3$-vertex. Note that due to $C_{\ref{C1}}$, each
neighbor of $u$ is an $8$-vertex. We consider four cases depending on
the number $n_t$ of triangular faces containing $u$. In each case, we
show that $u$ receives a weight of $3$ during the discharging
procedure, so its final weight is $0$.
\begin{enumerate}
\item $n_t=0$: by $R_{\ref{R1}}$, the vertex $u$ receives $1$ from
  each incident face.
\item $n_t=1$: the vertex $u$ receives $2$ by
  $R_{\ref{R1}}$. Moreover, $u$ is a semi-weak neighbor of two
  8-vertices. By $R_{\ref{R5}}$, it receives $2\times\frac{1}{2}$.
\item $n_t=2$: the vertex $u$ receives $1$ by
  $R_{\ref{R1}}$. Moreover, $u$ is a semi-weak neighbor of two
  8-vertices, and a weak neighbor of another $8$-vertex. By
  $R_{\ref{R4}}$ and $R_{\ref{R5}}$, it receives
  $1+2\times\frac{1}{2}$.
\item $n_t=3$: the vertex $u$ is a weak neighbor of three
  $8$-vertices. By $R_{\ref{R4}}$, it receives $3\*1$.
\end{enumerate}

\subsection{$4$-vertices}
Similarly to the previous subsection, we take a $4$-vertex $u$ and
consider several cases considering the number $n_t$ of triangular
faces incident with $u$. In each case, we show that $u$ receives at
least a weight of $2$, so ends up with non-negative weight. Recall
that, due to $C_{\ref{C1}}$, every neighbor of $u$ has degree at least
$7$.
\begin{enumerate}
\item $n_t\leqslant 2$: By $R_{\ref{R1}}$, the vertex $u$ receives
  $(4-n_t)\* 1\geqslant 2$ from incident faces.
\item $n_t=3$: In this case, $u$ receives $1$ by
  $R_{\ref{R1}}$. Moreover, $u$ is a weak neighbor of two vertices
  $w_1$ and $w_2$ and a semi-weak neighbor of two other ones $s_1$ and
  $s_2$.
  \begin{enumerate}
  \item If $d(w_1)=d(w_2)=8$ then both $w_1$ and $w_2$ give at least
    $\frac{1}{2}$ to $u$ by $R_{\ref{R6}}$, hence $u$ receives $1$.
  \item If $d(w_1)=d(w_2)=7$, then for $1\leqslant i\leqslant 2$,
    either $d(s_i)=7$ and $w_i$ gives $\frac{1}{2}$ to $u$ by
    $R_{\ref{R9}}$, or $d(s_i)=8$ and $u$ receives
    $\frac{5}{12}+\frac{1}{12}$ from $s_i$ and $w_i$ by $R_{\ref{R9}}$
    and $R_{\ref{R7}}$. In both cases, $u$ receives
    $2\* \frac{1}{2}=1$.
  \item If $d(w_1)=7$ and $d(w_2)=8$ (the other case is similar), then
    $w_2$ gives at least $\frac{7}{12}$ to $u$ by
    $R_{\ref{R6}}$. Moreover, if $d(s_1)=7$, the vertex $u$ receives
    $\frac{5}{12}$ from $w_1$ by $R_{\ref{R9}}$. Otherwise,
    $d(s_1)=8$, and $u$ receives $\frac{1}{3}+\frac{1}{12}$ from $w_1$
    and $s_1$ by $R_{\ref{R9}}$ and $R_{\ref{R7}}$. In both cases, $u$
    receives $\frac{7}{12}+\frac{5}{12}=1$.
  \end{enumerate}
\item $n_t=4$: In this case, $u$ is a weak neighbor of four
  $7^+$-vertices, say $w_1,\ldots,w_4$, sorted by increasing
  degree. We show that applying $R_{\ref{R6}}$ and/or $R_{\ref{R9}}$
  gives a weight of 2 to $u$ in any case.
  \begin{enumerate}
  \item If $d(w_1)=8$, or $d(w_4)=7$, then each $w_i$ gives
    $\frac{1}{2}$, hence $u$ receives $4\*\frac{1}{2}=2$.
  \item If $d(w_1)=7$ and $d(w_2)=8$, then $w_1$ gives $\frac{1}{3}$,
    its two neighbors among $\{w_1,\ldots,w_4\}$ give
    $2\*\frac{7}{12}$ and the remaining vertex gives $\frac{1}{2}$.
  \item If $d(w_2)=7$ and $d(w_3)=8$, then if $w_3w_4\in E(G)$, $u$
    receives $2\*\frac{7}{12}$ from $w_3$ and $w_4$, and
    $2\*\frac{5}{12}$ from $w_1$ and $w_2$. Otherwise, $u$ receives
    $2\*\frac{2}{3}$ from $w_3$ and $w_4$ and $2\*\frac{1}{3}$ from
    $w_1$ and $w_2$.
  \item If $d(w_3)=7$ and $d(w_4)=8$, then $w_4$ gives $\frac{2}{3}$
    to $u$, its neighbors among the $w_i$'s gives $2\*\frac{5}{12}$
    and the last neighbor of $u$ gives $\frac{1}{2}$.
  \end{enumerate}
\end{enumerate}

\subsection{$5$-vertices}
Take a $5$-vertex $u$. If $u$ is incident to a non-triangular face,
then it receives $1$ by $R_{\ref{R1}}$. Thus, we only have to consider
the case where $u$ is triangulated. We denote by $v_1,\ldots,v_5$ the
consecutive neighbors of $u$ in the chosen embedding of $G$. 

Note also that due to $C_{\ref{C1}}$, the minimum degree $\delta$ of
the neighborhood of $u$ is at least $5$. We distinguish three cases
depending on $\delta$. In each case, we show that $u$ receives a
weight of at least $1$, hence ends up with non-negative weight.

\begin{enumerate}
\item If $\delta\geqslant 7$: denote by $n_8$ the number of 8-vertices
  adjacent to $u$.
  \begin{enumerate}
  \item If $n_8=0$, by $R_{\ref{R10}}$, the vertex $u$ receives
    $5\*\frac{1}{5}=1$.
  \item If $n_8=1$, we may assume that $d(v_1)=8$. Then $u$ is an $E_3$-neighbor of $v_1$ and an
    $S_6$-neighbor of $v_2,v_3,v_4$ and $v_5$. By $R_{\ref{R8}}$, it receives $\frac13$ from $v_1$ and by $R_{\ref{R10}}$, it
    receives $\frac{1}{6}$ from $v_2,v_3,v_4,v_5$. At the end,
    the received weight is thus $\frac{1}{3}+4\*\frac{1}{6}=1$.
  \item If $n_8\geqslant 2$, then each neighbor of degree $8$ gives at least
    $\frac14$, while the other neighbors give $\frac16$. Thus $u$
    receives at least
    $n_8\times\frac14+(5-n_8)\times\frac16\geqslant 1$.
  \end{enumerate}
\item If $\delta=6$, we consider different cases depending on the
  number $n_6$ of $6$-vertices in the neighborhood of $u$. Note that
  $n_6\leqslant 3$ because of $C_{\ref{C5}}$.
  \begin{enumerate}
  \item If $n_6=3$, then denote by $x$ and $y$ the two neighbors of
    $u$ of degree at least $7$. Due to $C_{\ref{C5}}$, the vertices
    $x$ and $y$ are not consecutive neighbors of $u$, and moreover, we
    cannot have $d(x)=d(y)=7$. We may thus assume that
    $d(x)=8$. Therefore, $u$ receives $\frac{2}{3}$ from $x$ by
    $R_{\ref{R8}}$ and at least $\frac{1}{3}$ from $y$ by
    $R_{\ref{R10}}$ or $R_{\ref{R8}}$.
  \item If $n_6=2$, then for any $i$ such that $d(v_i)>6$, $u$ is
    either an $S_3$ or an $E_3$-neighbor of $v_i$. Thus, by
    $R_{\ref{R10}}$ or $R_{\ref{R8}}$, the vertex $u$ receives
    $3\*\frac{1}{3}$ from them.
  \item If $n_6=1$, then we may assume that $d(v_1)=6$. Thus, for
    $i=2,5$, $u$ is an $S_3$-neighbor or an $E_3$-neighbor of
    $v_i$. Thus, by $R_{\ref{R8}}$ or $R_{\ref{R10}}$, $u$ receives
    $2\*\frac{1}{3}$ from $v_2$ and $v_5$. Moreover, $u$ receives at
    least $2\*\frac{1}{6}$ by $R_{\ref{R10}}$ and $R_{\ref{R8}}$ from
    $v_3$ and $v_4$. In any case, $u$ receives at least
    $2\*\frac{1}{3}+2\times\frac{1}{6}=1$.
  \end{enumerate}
\item If $\delta=5$, note that $u$ is adjacent to only one $5$-vertex
  (due to $C_{\ref{C3a}}$). We may thus assume that $d(v_1)=5$ and
  $d(v_i)>5$ for $2\leqslant i\leqslant 5$. Moreover, we have
  $d(v_2)>6$ and $d(v_5)>6$ due to $C_{\ref{C3b}}$. We show that $v_2$
  and $v_3$ give together at least $\frac{1}{2}$ to $u$. By symmetry,
  $u$ will receive at least $2\*\frac{1}{2}$ from $v_2,v_3,v_4$ and
  $v_5$.
  \begin{enumerate}
  \item If $d(v_3)=6$, then $u$ is an $(5,6)$-neighbor of $v_2$. Thus,
    $u$ receives $\frac{1}{2}$ by $R_{\ref{R8}}$ or $R_{\ref{R10}}$.
  \item Otherwise, $u$ is either an $S_3$-neighbor or an
    $E_3$-neighbor of $v_3$. Therefore, $u$ receives $\frac{1}{3}$
    from $v_3$ by $R_{\ref{R8}}$ or $R_{\ref{R10}}$, and at least
    $\frac{1}{6}$ from $v_2$ by $R_{\ref{R8}}$ or $R_{\ref{R10}}$.
  \end{enumerate}
\end{enumerate}

\subsection{$6$-vertices}

Note that $6$-vertices do not give nor receive any weight. Moreover,
their initial weight is 0. Thus their final weight is 0, hence
non-negative.

\subsection{$8$-vertices}
Let $u$ be an $8$-vertex and $v_1,\ldots,v_8$ its neighbors in
clockwise order. A full case analysis can be found in~\cite{arxiv}. We
rather provide a program in
\href{https://github.com/tpierron/Delta8}{\url{https://github.com/tpierron/Delta8}}
to check this case.

Note that $u$ gives weight only to the $v_i$'s, and that the amount
given to each $v_i$ depends only on its degree, the degree of
$v_{i-1}$ and $v_{i+1}$, and whether $v_i$ is special (i.e. weak or
$E_3$). We consider eight types of neighbors: degree 3, degree 4,
$E_3$, weak degree 5 but not $E_3$, non-weak degree 5, degree 6,
degree 7, degree 8. The program runs in 22s on a standard computer and
does the following:
\begin{itemize}
\item[$\bullet$] It generates all $8^8$ types for $v_1,\ldots,v_8$ and keeps only
  one representative up to rotations and symmetries.
\item[$\bullet$] For each of these neighborhood types, it computes an upper bound
  on how much $u$ will lose obtained by looking at the worst case in
  the rules. It then discards all configurations where this upper
  bound is at most 2.
\item[$\bullet$] For each of the remaining neighborhood types, it generates the
  (at most) $2^8$ cases for the edges $v_iv_{i+1\mod 8}$.
\item[$\bullet$] It computes exactly the loss for $u$ in each case, and discards
  the cases where the loss is at most 2.
\item[$\bullet$] It then checks that all the remaining cases contain a reducible
  configuration.
\end{itemize}

\subsection{$7$-vertices}
Let $u$ be a $7$-vertex, and denote by $v_1,\ldots,v_7$ its
consecutive neighbors in the chosen planar embedding of $G$. 

We follow the same approach as for $8$-vertices, with some slight
adjustments regarding the types. We consider the eight following types
for the $v_i$'s: degree 4, $S_3$-neighbor, $S_5$-neighbor, other weak
5-neighbor, non-weak 5-neighbor, degree 6, degree 7, degree 8.

The rest of the program is the roughly the same as in the previous
subsection. The only difference comes from the fact that we did not
implement rule $R_{\ref{R3+}}$ in order to keep the program simple. As
a consequence, the program actually outputs a configuration which is not
reducible and for which the computed total loss is larger than 1. This
configuration is depicted in Figure~\ref{fig:last}.

\begin{figure}[!h]
  \centering
  \begin{tikzpicture}[every node/.style={draw=black,minimum size = 10pt,ellipse,inner sep=1pt}]
    \node [label=left:{$u$}] (u) at (0,0)  {$7$};
    \node [label=above right :{$v_1$}, label=right:{\scriptsize weak}] (v1) at (0:1) {$5$};
    \node [label=above right :{$v_2$}, label=right:{\scriptsize weak}] (v2) at (51.5:1) {$5$};
    \node[label=left:{$v_3$}] (v3) at (103:1) {$6$};
    \node[label=left:{$v_4$}] (v4) at (154.5:1) {$5$};
    \node[label=right:{$v_7$}] (v7) at (308.5:1) {$7^+$};
    \node[label=left:{$v_6$}] (v6) at (257:1) {$6$};
    \node[label=left:{\scriptsize weak},label=above left:{$v_5$}] (v5) at (206.5:1) {$5$};
    \draw (v5) -- (u) -- (v4) -- (v5) -- (v6) -- (v7) -- (v1); 
    \draw (v3) -- (v2) -- (v1) -- (u) -- (v2);
    \draw (v5) -- (v6) -- (u) -- (v3);
    \draw (u) -- (v7) -- (v6);
  \end{tikzpicture}
  \caption{The output of the program. There is no edge between $v_3$ and $v_4$.}
  \label{fig:last}
\end{figure}

In that case, $u$ gives $2\times \frac12$ to $v_2$ and $v_5$, and
$\frac16$ to $v_1$. Assume that $R_{\ref{R3+}}$ does not apply. This
means that the face incident to $v_3,u,v_4$ is a 4-face and the last
vertex $x$ has degree 5. Due to $C_{\ref{C3a}}$, we obtain that
$x=v_5$, which is impossible since $v_4$ has degree $5$ and is contained only in the faces $uv_4v_5$ and $uv_3xv_4$. Therefore, $R_{\ref{R3+}}$ applies and the total
loss is $2\times\frac12+\frac16-\frac16=1$.

\section{Conclusion}

To prove our result, we use the discharging method. In this case, we
have a lot of configurations to reduce. The key ideas come from the
two approaches we use to reduce them. While the Combinatorial
Nullstellensatz approach and recoloring arguments have already been
used many times in discharging proofs, the framework we present here
(the so-called color shifting graph) seems to be quite new. To our
knowledge, it was first used in~\cite{bonamy}. Here we designed a more
generic framework to use this idea. This allowed us to reduce some
configurations we cannot tackle in the usual way. While total
$9$-choosability seems out of reach from a reasonable discharging
proof, it would also be interesting to see if the methods we
introduced here can help to prove that $\chi''_\ell =\Delta+1$ for
planar graphs when $\Delta$ can be less than $12$.

Like most of the discharging proofs, our result comes along with a
linear-time algorithm to find a proper coloring of a graph $G$ given a list
assignment. Indeed, we first apply the discharging rules to $G$ in
linear time. Using the elements with negative final weight, we can
identify the reducible configurations in $G$. Then, we color
recursively the graph obtained by removing one reducible
configuration. However, instead of moving again the weights, we keep
track of what happens when we remove the configuration.

To extend the obtained coloring, observe that both the case analysis
and the recoloring approaches lead to a constant time process. For the
Nullstellensatz approach, it is trickier since the proofs are not
constructive. However, we can compute a proper coloring for each
configuration and for each list assignment just by brute-force. Since
the sizes of the configurations are bounded by a constant (except for
$C_{\ref{C2}}$, but we only used it for removing some cycles of length
4), this also takes constant time. Each recursive step thus takes
constant time and we have at most one for each element of
$G$. Therefore, if we add the initial discharging phase (which also
takes linear time), we obtain a linear-time algorithm.

\section*{Acknowledgment}

We thank the referee for their time and dedication during their careful reading of this paper.

\bibliography{total_main}

\end{document}